\documentclass[a4paper]{article}
\usepackage{graphicx} 
\usepackage{amsmath,amssymb,amsthm}
\usepackage{enumerate}
\usepackage[svgnames]{xcolor}
\usepackage{color}
\usepackage{hyperref}
\hypersetup{
    colorlinks=true,
    linkcolor=purple,
    citecolor=blue,
    urlcolor=Navy,
    filecolor=Gray,
}
\usepackage{cleveref}
\usepackage{graphicx}
\usepackage{tikz}
\usepackage{tcolorbox}
\usepackage{authblk}
\usepackage{thmtools}
\usepackage{thm-restate}
\usepackage{geometry}
\usepackage{fontawesome5}
\usepackage{orcidlink}
\usepackage{geometry}
\tcbuselibrary{skins, breakable}
\newif\ifshort
\shortfalse % for full versions

\newcommand{\defparproblem}[4]{%
  \vspace{3mm}
  \noindent\fbox{%
    \begin{minipage}{0.96\linewidth}
      \makebox[\linewidth]{%
      \noindent\textsc{#1}\hfill{\bfseries Parameter:}~#3%
      }
      \noindent{\bfseries Input:} #2

      \noindent{\bfseries Question:} #4
    \end{minipage}%
  }
\vspace{2mm}
}

\newcommand{\PITVD}{{\sc(Proper-Interval, Tree)-Vertex Deletion}}
\newcommand{\pitgraph}{{(prop-int, tree)-graph}}
\newcommand{\ctpair}{{(claw, triangle)-pair}}
\newcommand{\pig}{proper interval graph}
\newcommand{\cBB}{{\mathcal B}}
\newcommand{\cCCC}{{\mathcal C}}

\newcommand{\cGG}{{\mathcal G}}
\newcommand{\cKK}{{\mathcal K}}
\newcommand{\cOO}{{\mathcal O}}
\newcommand{\uv}{\prec_{\mathcal{V}}}
\newcommand{\euv}{\preceq_{\mathcal{V}}}

\newtheorem{reductionrule}{Reduction Rule}

\newtheorem{lemma}{Lemma}
\newtheorem{claim}{Claim}
\newtheorem{definition}{Definition}

\newtheorem{observation}{Observation}
\newtheorem{proposition}{Proposition}

\newenvironment{claimproof}{\let\oldqedsymbol\qedsymbol\renewcommand{\qedsymbol}{$\lhd$}\begin{proof}}{\end{proof}\let\qedsymbol\oldqedsymbol}

\title{A Polynomial Kernel for Vertex Deletion to the Scattered Class of Proper Interval Graph and Trees}
\author[1]{Ashwin Jacob\thanks{ashwinjacob@nitc.ac.in \href{mailto:ashwinjacob@nitc.ac.in}{\textcolor{purple}{\faEnvelope}} \href{https://ashwinjacob.github.io}{\faHome} \orcidlink{0000-0003-4864-043X}}}

\author[2]{Arpit Kumar\thanks{arpitk@iiitd.ac.in \href{mailto:arpitk@iiitd.ac.in}{\textcolor{purple}{\faEnvelope}} \href{https://sites.google.com/iiitd.ac.in/arpit-gang}{\faHome} \orcidlink{0009-0002-0741-5682}}}

\author[2]{Diptapriyo Majumdar\thanks{diptapriyo@iiitd.ac.in \href{mailto:diptapriyo@iiitd.ac.in}{\textcolor{purple}{\faEnvelope}} \href{https://diptapriyomajumdar.wixsite.com/toto}{\faHome} \orcidlink{0000-0003-2677-4648}}}

\affil[1]{National Institute of Technology Calicut, India}
\affil[2]{Indraprastha Institute of Information Technology Delhi, New Delhi, India}

\date{}

\begin{document}

\maketitle

\begin{abstract}
Vertex deletion to hereditary graph class have been well-studied in parameterized complexity.
Vertex deletion in scattered graph classes has attracted attention in recent years.
In this paper, we consider {\PITVD}, the input to which is an undirected graph $G = (V, E)$ and an integer $k$. 
The goal is to pick a set $X \subseteq V(G)$ with at most $k$ vertices such that $G - X$ is a simple graph and every connected component of $G - X$ is a proper interval graph or a tree.
When parameterized by the solution size $k$, {\PITVD} has been proved to be fixed-parameter tractable by Jacob et al. [JCSS-2023, FCT-2021].
Notably, {\PITVD} is a natural and substantial generalization of \textsc{(Clique, Tree)-Vertex Deletion}, which has been explored recently [JCSS-2026, ISAAC-2025, IPL-2025, ISAAC-2024].
The kernelization techniques for the latter exploit the finite forbidden induced subgraph characterization of cliques. In contrast, proper interval graphs and trees admit infinite forbidden induced subgraph characterizations, requiring new structural arguments.
In this paper, we consider this problem from the perspective of polynomial kernelization and provide the first nontrivial polynomial kernel for {\PITVD}, with $O(k^{33})$ vertices.
\end{abstract}
%While the latter exploits the finite forbidden characterization of cliques, our target scattered graph class has no finite forbidden family, requiring novel structural arguments.

% \tableofcontents
\ifshort
\newpage
\else
\fi

%\section{Introduction}
%\label{sec:intro}
%\input{intro.tex}

\section{Introduction}
\label{sec:intro-final}
%%%new-and-final-intro
Graph modification problems constitute one of the central themes in algorithm design and parameterized complexity.
The objective of these problems is to modify a given graph by deleting vertices or edges or sometimes adding edges so that the resulting graph belongs to a prescribed graph class.
%Among the various graph modification problems, {\sc Vertex Deletion} problems have received extensive attention over the last three decades from the perspectives of fixed-parameter tractability, approximation algorithms, and polynomial kernelization.
Among graph modification problems, {\sc Vertex Deletion} variants have been extensively studied in the contexts of approximation algorithms, parameterized complexity, and polynomial kernelization over the past few decades.
Classical examples include {\sc Vertex Cover} \cite{ChenKX10,HarrisNari2024}, {\sc Feedback Vertex Set} \cite{ChenFLLV08,KociumakaP14,RamanSS06,10.1145/1721837.1721848}, {\sc Odd Cycle Transversal} \cite{JansenPL21,KratschW14,KrithikaN13}, {\sc Chordal Vertex Deletion} \cite{AgrawalLMSZ19,CaoM15}, {\sc Interval Vertex Deletion} \cite{AgrawalLMSZ23,CaoM15} etc.
Several of these computational problems are NP-hard in general graphs but are polynomial-time solvable in special graph classes.
For example, both {\sc Vertex Cover} and {\sc Feedback Vertex Set} are NP-hard in general graphs but polynomial-time solvable on chordal~\cite{FominVillanger2010,FominVillanger2012,Golumbic2004}, AT-free~\cite{KratschMullerTodinca1999,KratschMT08}, bounded-treewidth~\cite{Courcelle1990}, bounded-mim-width~\cite{BuiXuanTelleVatshelle2013,JaffkeKwonTelle2020}, and permutation graphs~\cite{Liang1994FVS}; moreover, {\sc Vertex Cover} is polynomial-time solvable on bipartite graphs~\cite{Egervary1931,konig1931graphs,Schrijver2003}.
% Nevertheless, various domination, graph separation, and graph connectivity problems are NP-Complete even when restricted to split graphs which is a subclass of chordal graphs, e.g. {\sc Length Bounded Cut} \cite{BentertHK22}, {\sc Geodetic Set} \cite{DouradoPRS10,EkimEHHM12}, {\sc Steiner Tree} \cite{WhiteFarberLiang1985}, {\sc Hamiltonian Path} \cite{Keil1985Hamiltonian,RenjithS22}, {\sc Dominating Set}, and some of its variants are NP-Complete in chordal graphs while polynomial-time solvable in proper interval graphs and forests \cite{Bertossi84,Chang1998Domination,Chang1998Handbook,HenningPT23,KusumP24,MerouaneC15,PandeyP17,TripathiKP25}.
Nevertheless, several domination, separation, and connectivity problems exhibit a sharp complexity jump: they are polynomial-time solvable on proper interval graphs and forests but NP-complete on chordal graphs, often remaining NP-complete even on split graphs.
Examples include {\sc Length Bounded Cut}~\cite{BentertHK22}, {\sc Geodetic Set}~\cite{DouradoPRS10,EkimEHHM12}, {\sc Steiner Tree}~\cite{WhiteFarberLiang1985}, {\sc Hamiltonian Path}~\cite{Keil1985Hamiltonian,RenjithS22}, {\sc Role Assignment} \cite{HeggernesHP12} and {\sc Dominating Set} and its variants \cite{Bertossi84,Chang1998Domination,Chang1998Handbook,HenningPT23,KusumP24,MerouaneC15,PandeyP17,TripathiKP25}.
%, which are NP-complete on chordal graphs but polynomial-time solvable on proper interval graphs and forests~\cite{Bertossi84,Chang1998Domination,Chang1998Handbook,HenningPT23,KusumP24,MerouaneC15,PandeyP17,TripathiKP25}.

\ifshort
Assuming a uniform graph class for the entire graph is often overly restrictive in practice, as real-world networks typically exhibit heterogeneous structure---some regions dense, others sparse or tree-like.
This motivates the study of {\sc Scattered Graph Classes}, in which each connected component may belong to one of several prescribed classes.
Such a framework broadens the scope of tractable instances while retaining the algorithmic benefits of the constituent classes: if a problem is polynomial-time solvable on each individual class, it remains so on the corresponding scattered class by processing components independently.
\else
However, assuming that an entire graph belongs to a single graph class is overly restrictive in many practical settings.
Real-world networks often consist of structurally heterogeneous regions: some parts exhibit dense interactions, whereas others possess sparse or tree-like connectivity.
This observation motivates the study of {\sc Scattered Graph Classes}, where every connected component of the graph is allowed to belong to one of several graph classes.
Such classes substantially enlarge the family of tractable instances while preserving the algorithmic advantages of their constituent graph classes.
In particular, if a computational problem admits a polynomial-time algorithm on each constituent graph class individually, then it immediately remains polynomial-time solvable on the corresponding scattered graph class by processing every connected component independently.
Consequently, scattered graph classes provide a natural framework for extending tractability beyond uniform graph classes.
\fi

The algorithmic study of scattered graph classes was initiated by Ganian et al.~\cite{GanianRamanujanSzeider2017} in the context of backdoors for constraint satisfaction.
Subsequently, Jacob et al. \cite{JacobKroonMajumdarRaman2023} built on their work and studied the {\sc $(\cGG_1,\dots,\cGG_d)$-Vertex Deletion} problem, which asks whether at most $k$ vertices can be deleted so that every connected component of the remaining graph belongs to some $\cGG_i$.
While numerous pairs, including $(\textsc{Clique},\textsc{Tree})$, $(\textsc{Interval},\textsc{Tree})$, and $(\textsc{Proper-Interval},\textsc{Tree})$, have been studied for fixed-parameter tractability \cite{JacobKroonMajumdarRaman2023,JACOB2023280,JansenKW25}, polynomial kernelization remains largely unexplored.
To date, the only nontrivial polynomial kernel is for {\sc $(\textsc{Clique},\textsc{Tree})$-Vertex Deletion}, due to Jacob et al. \cite{JacobMZ26,JacobMZ2024isaac}, with subsequent improvements \cite{Kumabe2025isaac,Tsur25}.
These developments naturally raise the following question:
% The algorithmic study of scattered graph classes was initiated by Ganian et al. \cite{GanianRamanujanSzeider2017} in the context of backdoors for constraint satisfaction problems.
% Subsequently, Jacob et al. \cite{JacobKroonMajumdarRaman2023} introduced the {\sc $(\cGG_1,\ldots,\cGG_d)$-Vertex Deletion} problem from the perspective of parameterized complexity.
% Given a graph $G$ and an integer $k$, the objective of {\sc $(\cGG_1,\ldots,\cGG_d)$-Vertex Deletion} is to determine if one can delete at most $k$ vertices so that every connected component of the remaining graph belongs to $\mathcal{G}_i$ for some $i \in [d]$.
% Since then, numerous pairs of scattered graph classes have been investigated from the viewpoint of fixed-parameter tractability, including $(\textsc{Clique},\textsc{Tree})$, $(\textsc{Interval},\textsc{Tree})$, and $(\textsc{Proper-Interval},\textsc{Tree})$, among others see \cite{JACOB2023280}.
% %%%%
% While there has been substantial progress on designing fixed-parameter algorithms for vertex deletion to scattered graph classes \cite{JacobKroonMajumdarRaman2023,JACOB2023280,JansenKW25}, polynomial kernelization remains largely unexplored.
% Until recently, the only nontrivial polynomial kernel was obtained for {\sc (Clique, Tree)-Vertex Deletion} by Jacob et al. \cite{JacobMZ26,JacobMZ2024isaac} where the kernel size has subsequently been improved through a sequence of works \cite{Kumabe2025isaac,Tsur25}.
% These results naturally raise the following question: 
\begin{center}
	``Can polynomial kernels be designed for scattered graph classes when each constituent graph class admits polynomial kernelization through different structural characterizations?''
	% and reduction rules?''
	%``Can polynomial kernelization be extended to scattered graph classes whose constituent graph classes require fundamentally different structural techniques?''	
\end{center}
Addressing this question is considerably more challenging than a direct substitution of component classes.
Existing kernels for {\sc $(\textsc{Clique},\textsc{Tree})$-Vertex Deletion}~\cite{JacobMZ26,JacobMZ2024isaac,Kumabe2025isaac,Tsur25} rely on the structural simplicity of cliques, where neighborhood interactions remain relatively tractable. Proper interval graphs, by contrast, exhibit a far richer structure, admitting arbitrarily many overlapping maximal cliques arranged along a proper interval ordering, which necessitates advanced techniques such as clique partitions and clique bypassing \cite{KE2018109}. Hence, the reduction rules devised for clique components do not transfer directly to proper interval components.
% This question is considerably more challenging than simply replacing one graph class by another.
% Existing kernels for {\sc (Clique, Tree)-Vertex Deletion} \cite{JacobMZ26,JacobMZ2024isaac,Kumabe2025isaac,Tsur25} exploit that the clique components are simple enough and makes neighborhood interactions comparatively easy to analyze.
% Proper interval graphs, possess a substantially richer structure.
% A proper interval graph may contain arbitrarily many overlapping maximal cliques arranged according to a proper interval ordering, requiring techniques based on clique partitions, and structural graph operations such as clique bypassing.
% Consequently, the reduction rules designed for clique components no longer extend directly for the components that are proper interval graphs.
Trees are sparse, acyclic, and may contain arbitrarily many claws, whereas proper interval graphs are dense, contain arbitrarily large cliques and numerous cycles, yet are claw-free and have no induced cycles of length exceeding three.
Consequently, structural tools effective on one class fail on the other.
Unlike prior scattered pairs where one constituent exhibits relatively simple structure, no single decomposition captures both proper interval graphs and trees.
Thus, obtaining a polynomial kernel demands a unified framework integrating distinct structural perspectives, rather than a uniform reduction strategy.

% Tree components are sparse, acyclic, and may contain arbitrarily large number of claws.
% Proper interval graphs, on the other hand, are dense, contain arbitrarily large cliques, admit numerous cycles, yet are claw-free and has no induced cycles of length more than three.
% Thus, structural tools effective on one side become ineffective on the other.
% Unlike previous scattered graph classes where one constituent admits comparatively simple structure, no single decomposition simultaneously captures both proper interval graphs and trees.
% As a result, obtaining a polynomial kernel requires combining distinct structural viewpoints within a single framework rather than applying a uniform reduction strategy.

\subparagraph*{Our Problem and Result.}
In this paper, we investigate the {\PITVD} from the perspective of polynomial kernelization.
Even though the input to our problem is a simple undirected graph, some of our reduction rules can create parallel edge.
Therefore, we consider the more general formulation of the problem as follows.

\defparproblem{{\PITVD} (PITVD)}{An undirected (multi)graph $G = (V, E)$ and an integer $k$.}{$k$}{Does there exist $X \subseteq V(G)$ such that $G - X$ is a simple graph and every connected component of $G - X$ is a proper interval graph or a tree?}

\noindent Although the problem admits an FPT algorithm running in $7^k n^{\cOO(1)}$ time, the existence of a polynomial kernel remained open.
Our work provides the first nontrivial kernel for PITVD.

\begin{restatable}{theorem}{FinalMainThm}
    \label{thm:final-result}
    {\PITVD} admits a kernel with $\cOO(k^{33})$ vertices.
\end{restatable}

% While our kernel size is large, our result above is the first nontrivial polynomial kernel for {\sc ($\cGG_1, \cGG_2$)-Vertex Deletion} in which both $\cGG_1$ and $\cGG_2$ are characterized by infinite forbidden families.
% The significance of this result lies not merely in obtaining another polynomial kernel for a scattered graph class, but in demonstrating that kernelization remains possible even when the constituent graph classes require substantially different structural techniques.
% Our proof cannot be viewed as a straightforward extension of the existing kernel for {\sc (Clique, Tree)-Vertex Deletion}.
% Instead, it combines several ingredients that interact in a nontrivial manner.

While the resulting kernel is large, it constitutes the first nontrivial polynomial kernel for {\sc $(\cGG_1,\cGG_2)$-Vertex Deletion} where both classes are characterized by infinite forbidden families.
Its significance lies not in the kernel size itself, but in establishing that polynomial kernelization remains feasible even when the two constituent classes demand fundamentally distinct structural techniques.
Our proof is not a straightforward extension of the existing kernel for {\sc ({Clique}, {Tree})-Vertex Deletion}; rather, it combines several ingredients in a nontrivial interplay such as {\em Sunflower Lemma, Expansion Lemma, Clique bypassing}, {\em handling claws in trees} etc.
\ifshort
Due to lack of space, we refer to the full version (attached at the end) and to \cite{CyganFKLMPPS15,downey2013fundamentals} for more formal definitions of parameterized complexity and kernelization.
\else
\fi

\subparagraph*{Related Work.}
%Ganian et al. \cite{GanianRamanujanSzeider2017} considered backdoors to scattered classes of CSPs and provided FPT algorithms.
%Jacob et al. \cite{JacobKroonMajumdarRaman2023} considered {\sc $(\mathcal{G}_1,\ldots,\mathcal{G}_d)$-Vertex Deletion} for fixed constance $d$ and designed FPT algorithm with running time $2^{k^{\cOO(d)}}n^{\cOO(1)}$.
Jansen et al. \cite{JansenKW25} studied {\sc $(\mathcal{G}_1,\ldots,\mathcal{G}_d)$-Vertex Deletion} when every $\mathcal{G}_i$ is defined by finite forbidden families and designed a single exponential time algorithm for the same.
%Subsequently, pairs of scattered graph classes were considered and Jacob et al. \cite{JACOB2023280} and they designed $2^{\cOO(k)}n^{\cOO(1)}$-time FPT algorithms for various {\sc $(\mathcal{G}_1, \mathcal{G}_2)$-Vertex Deletion} problems where $(\mathcal{G}_1, \mathcal{G}_2)$ is {\sc (Clique, Tree), (Split, Bipartite), (Proper Interval, Tree), (Interval, Tree), (Claw-Free, Triangle-Free)} etc.
For {\sc (Clique, Tree)-Vertex Deletion}, Jacob et al. \cite{JacobMZ26} provided a kernel with $\cOO(k^5)$ vertices which was later improved to $\cOO(k^4)$ vertices by Tsur \cite{Tsur25} and to $\cOO(k^2)$ vertices by Kumabe \cite{Kumabe2025isaac}.
The first polynomial kernel for {\sc Proper Interval Vertex Deletion} of  $\cOO(k^{53})$ vertices was provided by Fomin et al. \cite{fomin2013polynomial}, and was subsequently improved to $\cOO(k^4)$ vertices by Ke et al. \cite{KE2018109}.
{\sc Feedback Vertex Set} is known to admit a kernel with $4k^2$ vertices due to Thomasse \cite{10.1145/1721837.1721848}.

%\section{Introduction 2}
%\input{intro-2.tex}

\section{Preliminaries and Notations}
\label{sec:prelims}
%%%preliminaries
\subparagraph{Sets and Graph Theory.} 
We use standard graph theoretic terminologies from the book of Diestel \cite{diestel2017graph}.
Throughout the paper, we refer to only undirected graphs.
If every connected component of a simple undirected graph $G$ is a proper interval graph or a tree, then we say that $G$ is a \emph{\pitgraph}.
Throughout the paper, we use ``a simple {\pitgraph}'' for the sake of simplicity most of the time.
A {\em path} in a graph is a sequence of distinct vertices $(v_1, v_2, \ldots, v_{\ell})$ such that for all $1 \leq i \leq {\ell - 1}$, $v_i v_{i+1}$ is an edge.
A {\em cycle} in a graph is a sequence of vertices $(v_1, v_2, \ldots, v_{\ell}, v_1)$ such that $\ell \geq 2$, $v_1,\ldots, v_{\ell}$ are distinct, for every $1 \leq i \leq \ell - 1$, $v_i v_{i+1}$ is an edge and $v_{\ell} v_1$ is an edge. 
A path $P = (v_1, v_2, \ldots, v_{\ell})$ is a {\em degree-2-path} in a graph $G$ if every internal vertex has degree exactly 2 in graph $G$.
A degree-2-path is a {\em degree-2-tail} in $G$ if $d_G(v_1) > 2$ and $d_G(v_{\ell}) = 1$.
A degree-2-path is a {\em degree-2-overbridge} in $G$ if $d_G(v_1), d_G(v_{\ell}) > 2$.
Given an undirected graph $G$ and $v \in V(G)$, we define {\em attach a degree-2-tail of length $\ell$ at $v$} as follows: 
{create new vertices $u_1, u_2, \dots, u_{\ell}$ and add edges $vu_1$ and $u_iu_{i+1}$ for every $1\le i\le \ell-1$.}
Let $q$ be a positive integer and $G$ be a bipartite graph with vertex bipartition $(A, B)$.
For $\widehat{A}\subseteq A$ and $\widehat{B}\subseteq B$, a set $M \subseteq E(G)$ of edges is called a {\em $q$-expansion} of $\widehat{A}$ into $\widehat{B}$ if
\begin{enumerate}[(i)]
  \item every vertex of $\widehat{A}$ is incident to exactly $q$ edges in $M$; and
  \item exactly $q|\widehat{A}|$ vertices of $\widehat{B}$ are incident to the edges in $M$.
\end{enumerate}
The vertices of $\widehat{A}$ and $\widehat{B}$ that are the endpoints of the edges of $M$ are {\em saturated} by $M$.

%%%%%%%%%%%%%%%%%%%%%%%%%%%%%%%%%%%%%%%%%%%%%%%%%%%%%%%%%%%%%%%%%%%%%%%%%%%%%%%%%%%%%%%%%%%%%%%%%%%%%%
% Moreover, we use the concept of {\em $q$-expansion}, along with its algorithmic tools.
% It is important to note that by definition of $q$-expansion $M$ of $\widehat{A}$ into $\widehat{B}$, all vertices of $\widehat{A}$ are saturated by $M$, and $|B| \geq q|A|$.
% But not all vertices of $\widehat{B}$ are guaranteed to be saturated by $M$.

We use the concept of {\em $q$-expansion}, along with its algorithmic tools. 
It is important to note that by definition of a $q$-expansion $M$ of $\widehat{A}$ into $\widehat{B}$, every vertex of $\widehat{A}$ is saturated by $M$, and $|\widehat{B}| \geq q|\widehat{A}|$. However, a vertex of $\widehat{B}$ need not be saturated by $M$.

%\begin{definition}[$q$-expansion]
%  Let $q$ be a positive integer and $G$ be a bipartite graph with vertex bipartition $(A, B)$.
%  For $\widehat{A}\subseteq A$ and $\widehat{B}\subseteq B$, a set $M \subseteq E(G)$ of edges is called a {\em $q$-expansion} of $\widehat{A}$ into $\widehat{B}$ if
%  \begin{enumerate}[(i)]
%    \item every vertex of $\widehat{A}$ is incident to exactly $q$ edges in $M$; and
%    \item exactly $q|\widehat{A}|$ vertices of $B$ are incident to the edges in $M$.
%  \end{enumerate}
%  The vertices of $\widehat{A}$ and $\widehat{B}$ that are the endpoints of the edges of $M$ are {\em saturated} by $M$. 
%\end{definition}
%%%%%%%%%%%%%%%%%%%%%%%%%%%%%%%%%%%%%%%%%%%%%%%%%%%%%%%%%%%%%%%%%%%%%%%%%%%%%%%%%%%%%%%%%%%%%%%%%%
\begin{proposition}[$q$-Expansion Lemma \cite{CyganFKLMPPS15, 10.1145/1721837.1721848}]
{\label{lem:q-expansion}}
  Let $q$ be a positive integer and $G$ be a bipartite graph with vertex bipartition $(A, B)$ such that $|B|\ge q|A|$, and there is no isolated vertex in $B$. 
  Then there exist non-empty {sets} $\widehat{A}\subseteq A$ and $\widehat{B}\subseteq B$ such that
  \begin{enumerate}[(i)]
    \item\label{expansion-prop-1} There is a $q$-expansion $M$ of $\widehat{A}$ into $\widehat{B}$ in $G$.
    \item\label{expansion-prop-2} $N(\widehat{B})\subseteq \widehat{A}$.
  \end{enumerate}
  Furthermore, the sets $\widehat{A}, \widehat{B}$ and the $q$-expansion $M$ can be computed in time $\mathcal{O}(mn^{1.5})$.
\end{proposition}

%%%%%%%%%%%%%%%%%%%%%%%%%%%%%%%%%%%%%%%%%%%%%%%%%%%%%%%%%%%%%%%%%%%%%%%%%%%%%%%%%%%%%%%%%%%%%%%%%%%%
% Recently, Fomin et al. \cite{fomin2019subquadratic} have designed the following generalization of the $q$-Expansion Lemma.
\noindent
Later, Fomin et al. \cite{fomin2019subquadratic} designed the following generalization of the $q$-Expansion Lemma.

\begin{proposition}[New $q$-Expansion Lemma \cite{babu2024packing,fomin2019subquadratic, a16030144}]
{\label{lem:new_q-expansion}}
  Let $q$ be a positive integer and $G$ be a bipartite graph with vertex bipartition $(A, B)$. Then, there exist non-empty sets $\widehat{A}\subseteq A$ and $\widehat{B}\subseteq B$ such that
  \begin{enumerate}[(i)]
    \item\label{new-exp-1} There is a $q$-expansion $M$ of $\widehat{A}$ into $\widehat{B}$ in $G$.
    \item\label{new-exp-2} $N(\widehat{B})\subseteq \widehat{A}$.
    \item\label{new-exp-3} $|B \setminus \widehat{B}|\le q|A \setminus \widehat{A}|$.
  \end{enumerate}
  Furthermore, the sets $\widehat{A}, \widehat{B}$ and the $q$-expansion $M$ can be computed in polynomial-time.
\end{proposition}

% Observe that the \Cref{lem:new_q-expansion} statement does not require the condition that B has no isolated vertex and $|{B}|\ge q|{A}|$ that were required for the \Cref{lem:q-expansion}.
% In particular, if $|{B}|> q|{A}|$, then it must be that {$|\widehat{B}|> q|\widehat{A}|$ and} $\widehat{B}$ contains some {vertex that is} not saturated by the q-expansion $M$.
Note that \Cref{lem:new_q-expansion} omits two conditions required by \Cref{lem:q-expansion}: $B$ has no isolated vertex and $|B|\ge q|A|$.
In particular, if $|B|>q|A|$, then $|\widehat{B}|>q|\widehat{A}|$ and $\widehat{B}$ contains a vertex not saturated by the $q$-expansion $M$.
%{create new vertices $u_1,u_2,u_3,\ell,u_{\ell-1}$ and edges $vu_1$ and $u_i u_{i+1}$ for every $2 \leq i \leq \ell - 3$.}
%%%%In addition, we use the following observation that is a structural property.
%{\color{red}\begin{observation}
%    \label{obs:claw-free-tree}
%    If $G$ is a tree and $G$ does not contain any claw as an induced subgraph, then $G$ is a path.
%\end{observation}}
%%%}

% For a graph $G$ and a vertex $v \in G$, a {\em $v$-flower} is the structure formed by a family of $\ell$ cycles $C_1, C_2, \ldots C_{\ell}$ in $G$ such that each of these $\ell$ cycles contain $v$ and no two distinct cycles $C_i$ and $C_j$ share any vertex except $v$.
% We refer to the $C_i$'s as the {\em petals} and to $v$ as the {\em core}.
% The number of cycles $\ell$ is the {\em order} of the $v$-flower.
% We use the following proposition that is related to the $v$-flower structure.
Moreover, For a graph $G$ and a vertex $v \in V(G)$, a {\em $v$-flower} is a family of $\ell$ cycles
$C_1, C_2, \ldots, C_{\ell}$ in $G$ such that each cycle contains $v$, and any two distinct
cycles $C_i$ and $C_j$ intersect only at $v$. We refer to the $C_i$'s as the {\em petals}, and to $v$ as the
{\em core}.
The number of cycles $\ell$ is the {\em order} of the $v$-flower.
We use the
following proposition on $v$-flower.
%\begin{definition}[$v$-flower]
%{\label{def:flower}}
% For a graph $G$ and a vertex $v \in G$, a {\em $v$-flower} is the structure formed by a family of $\ell$ cycles $C_1, C_2, \ldots C_{\ell}$ in $G$ such that each of these $\ell$ cycles contain $v$ and no two distinct cycles $C_i$ and $C_j$ share any vertex except $v$.
% We refer to the $C_i$'s as the {\em petals} and to $v$ as the {\em core}.
% The number of cycles $\ell$ is the {\em order} of the $v$-flower.
%\end{definition}
%
%%%%%%%%%%%%%%%%%%%%%%%%%%%%%%%%%%%%%%%%%%%%%%%%%%%%%%%%%%%%%%%%%%%%%%%%%%%%%%%%%%%%%%%%%%%%%%
\begin{proposition}[\cite{CyganFKLMPPS15}, Lemma 9.6]
  {\label{prop:flower}}
  {Given a graph $G$, a vertex $v \in V(G)$, and an integer $k$, there is a polynomial-time algorithm that either finds a $v$-flower of order $k+1$ or computes a set $Z \subseteq V(G)\setminus\{v\}$ of size at most $2k$ that intersects every cycle of $G$ passing through $v$.}

  % Given a graph $G$ with $v \in V (G)$ and an integer $k$, there exists a polynomial-time algorithm that either provides a $v$-flower of order $k+1$ or {computes} a set $Z \subseteq V (G) \setminus \{v\}$ with at most $2k$ vertices such that $Z$ intersects every cycle of $G$ that passes through $v$.

  % and there are at most $2k$ edges incident to $v$ and with second endpoint in $Z$.
\end{proposition}

% We use the concept of $v$-flower only for cycles.
% However, we use a more generic concept of `sunflower' for our kernelization algorithm that is defined as follows.
% A \emph{sunflower} with $\ell$ petals and core $Y$ is a family of sets $\{S_1,\dots, S_{\ell}\}$ such that $S_i\cap S_j= Y$ for all $i\neq j$. The sets $S_i\setminus Y$ are petals, and we require none of them to be empty.
We use $v$-flower only for cycles. For our kernelization algorithm, we also use a more generic concept of {\em sunflower}: a family $\{S_1,\dots,S_\ell\}$ with core $Y$ such that $S_i\cap S_j=Y$ for all $i\neq j$. The sets $S_i\setminus Y$ are petals, and we require none of them to be empty.

\begin{proposition}[Lemma 3.2 of \cite{fomin2013polynomial}, Lemma 3.3 of \cite{EibenHR19}]
  \label{prop:sunflower-apply}
  Let $\mathcal{F}$ be a family of sets of cardinality at most $d$ over a universe $\mathcal{U}$ and $k$ be a positive integer.
  Then there is an $\mathcal{O}(|\mathcal{F}|(k+|\mathcal{F}|))$ time algorithm that finds a non-empty set $\mathcal{F}'\subseteq \mathcal{F}$ such that
  \begin{enumerate}[(i)]
    \item For every $Z\subseteq \mathcal{U}$ of size at most $k$, $Z$ is a minimal hitting set of $\mathcal{F}$ if and only if $Z$ is a minimal hitting set of $\mathcal{F}'$; and
    \item $|\mathcal{F}'|\le d!(k+1)^d$.
  \end{enumerate}
\end{proposition}

\ifshort
\else
\begin{figure}[ht]
    \centering
    \label{fig:obst-pig}
    \includegraphics[scale=0.6]{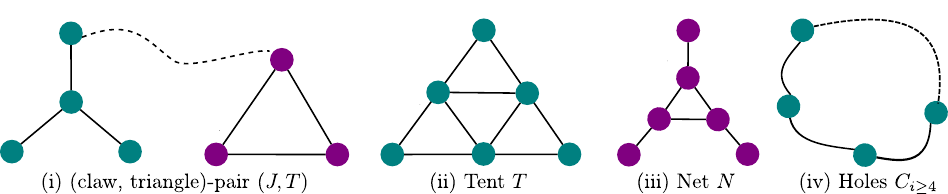}
    \caption{Forbidden Induced Subgraphs or Obstructions for {\PITVD} problem. The dashed path in {\ctpair} $(J, T)$ is a path connecting the claw $J$ and the triangle $T$.}
\end{figure}
\fi
\subparagraph{Forbidden Subgraph Characterization.}
% \todo[inline]{Diptapriyo: I am wondering if we should change this paragraph heading.}
It is well-known due to \cite{JACOB2023280} that a simple undirected graph is {\pitgraph} if and only if it does not contain any net, tent, and holes as induced subgraph, and it has no component that contains a claw and a triangle.
We call these graph structures {\em obstructions} for {\pitgraph}.
%{It is well-known that a simple undirected graph is {\pitgraph} if and only if it has no connected component that has both a claw and a triangle.}
If a connected graph contains a claw $J$ and triangle $T$, there is a (shortest) path $P$ that connects $J$ and $T$. 
We call $J$ and $T$ a {\em \ctpair}.
If the vertices of a claw $J$ and that of a triangle $T$ have a nonempty intersection, then $(J, T)$ is a {\em pairwise-intersecting} {\ctpair}.

\subparagraph*{Proper Interval Graphs.}
\ifshort
An {\em interval graph} is the intersection graph of intervals on $\mathbb{R}$. 
An {\em interval representation} is ${I_v=[\mathsf{lp}_v,\mathsf{rp}_v]\mid v\in V(G)}$ with $I_u\cap I_v\neq\emptyset \text{ if and only if } uv\in E(G)$.
A {\em proper interval graph} is an interval graph with no interval properly containing another. 
A {\em proper interval ordering} \cite{looges1993optimal} is an ordering $v_1,\dots,v_n$ such that $\mathsf{lp}_{v_i}<\mathsf{lp}_{v_j}$ (hence $\mathsf{rp}_{v_i}<\mathsf{rp}_{v_j}$).
The following proposition from \cite{KE2018109} is a well-known property of proper interval graphs.
\else
An \emph{interval graph} is an intersection graph of intervals on the real line.
The set of intervals corresponding to the vertices of an interval graph is called an \emph{interval representation} or \emph{interval model} of the graph which can be specified by a set of intervals $\{I_v:=[{\sf lp}_{v}, {\sf rp}_{v}] \mid v \in V(G)\}$ such that $I_u \cap I_v \neq \emptyset$ if and only if $uv \in E(G)$.
A \emph{proper interval graph} is an interval graph that has an interval representation in which no interval is properly contained in another interval.
A \emph{proper interval ordering} \cite{looges1993optimal} of a proper interval graph $G$ is a linear ordering of its vertices $v_1, v_2, \ldots, v_n$ such that for every $1 \leq i < j \leq n$, ${\sf lp}_{v_i} < {\sf lp}_{v_j}$ (and ${\sf rp}_{v_i} < {\sf rp}_{v_j}$ by definition of {\pig}).
It is well-known that given an undirected graph, there exists a polynomial-time algorithm that either outputs that the graph is not a proper interval graph or outputs a proper interval ordering in which every interval is of unit length.
Hence, we can assume that all the intervals representing a proper interval graphs are closed intervals of unit length.
The following proposition due to Cao et al. \cite{KE2018109} is well-known property of proper interval graphs.
\fi

% \begin{figure}[ht!]
%   \centering
%   \label{fig:proper-interval-ordering}
%   \includegraphics[scale=0.4]{FIGURES/proper-interval-ordering.pdf}
%   \caption{Schematic Diagram of \Cref{lemma:degree-2-tail-attach-to-pendant}}
% \end{figure}

\begin{proposition}[\cite{KE2018109}]
  {\label{prop:pig-clique}}
  Let $v_{1}, \dots, v_{n}$ be a proper interval ordering of a proper interval graph $G$.
  Then, for every $1 \leq i < j \leq n$, if $v_{i}v_{j} \in E(G)$, then $v_{i}, \dots, v_{j}$ is a clique.
\end{proposition}

\subparagraph{Construction of Clique Partition.}
Let $G$ be a proper interval graph, and we fix a proper interval ordering $\mathcal{V} :=\{v_1, v_2, \dots, v_n\}$ of {\pig} $G$.
We partition the vertex set of $G$ into a sequence of cliques as follows.
We start with the first vertex $v_1 \in \mathcal{V}$ and put all the vertices adjacent to $v_1$ (including $v_1$) in the first clique $K_1$, i.e., $K_1 =\{v_1, \dots, v_i\}$ such that ${\sf lp}_{v_i} \in I_{v_1}$ and ${\sf lp}_{v_{i+1}}\notin I_{v_1}$ where $v_{i+1}\in \mathcal{V}$.
Next, we remove all the vertices of $K_1$ from $\mathcal{V}$ and repeat the same procedure to form the next clique $K_2$ from the remaining vertices.
We continue this process until all the vertices of $G$ are assigned to some clique.
Let $\cKK := \{K_1, K_2, \dots, K_t\}$ be the sequence of cliques obtained from this process.
{Note that the indices of the cliques in the sequence $\cKK$ correspond to their order not the number of vertices in the cliques.}
{We say that $\cKK$ is a {\em clique partition} of $G$.}
There are some important properties of this sequence of cliques $\cKK$ as stated in the following propositions.

\begin{proposition}[\cite{KE2018109}, Proposition 2.2]
  {\label{prop:pig-clique-nbd}}
  Let $\cKK = \{K_1, K_2, \dots, K_t\}$ be a clique partition of a proper interval graph $G$ obtained from a proper interval ordering $\mathcal{V}$ of $G$.
  Then, for each $1 < i < t$, $N(K_i)\subset K_{i-1}\cup K_{i+1}$. Moreover, $N(K_1)\subseteq K_2$ and $N(K_t)\subseteq K_{t-1}$.
\end{proposition}

%{\color{red}\begin{proposition}[\cite{KE2018109}, Corollary 2.3]
%  {\label{prop:pig-clique-dist}}
%  Let $\mathcal{K} = \{K_1, K_2, \dots, K_t\}$ be a clique partition of a proper interval graph $G$.
%  For each pair of vertices $u\in K_i$ and $v\in K_j$ with $1\le i \le j\le t$,
%  \begin{itemize}
%    \item $d_{G}(u, v)\ge j-i$, and
%    \item if $u$ is the first vertex of $K_i$, then $d_{G}(u, v)\ge j-i+1$.
%  \end{itemize}
%  where $d_{G}(u, v)$ is the distance between $u$ and $v$ in $G$.
%\end{proposition}}

Now, we define one of the crucial techniques for our problem.
Consider a clique partition $\cKK = \{K_1, K_2,\ldots,K_t\}$ of a proper interval graph.
Suppose that there is an $\ell$ such that $2 \leq \ell \leq t-1$ and $N(K_{\ell}) \cap K_{\ell-1}, N(K_{\ell}) \cap K_{\ell + 1} \neq \emptyset$.
We borrow the concept of a graph operation from \cite{KE2018109} as follows.
We say that we {\em bypass} $K_{\ell}$ if we remove $K_{\ell}$ from the graph, and add edges between every pair of vertices $u\in N(K_{\ell})\cap K_{\ell -1}$ and $v\in N(K_{\ell})\cap K_{{\ell}+1}$.
It is well-known due to \cite{KE2018109} that if a clique from a proper interval graph is bypassed, then the resulting graph is also a proper interval graph. 
%%%%
% \begin{figure}[ht!]
%     \centering
%     \label{fig:bypassing}
%     \includegraphics[scale=0.4]{FIGURES/bypassing.pdf}
%     \caption{Schematic Diagram of Bypassing a clique $K_{\ell}$ in a proper interval graph $G$.}
% \end{figure}
%%%
%{\color{red}
%\begin{proposition}[\cite{KE2018109}, Proposition 2.4]
%  {\label{prop:bypass-pig}}
%  Let $\mathcal{K}=\{K_1, K_2, \dots, K_t\}$ be a clique partition of a {\pig} $G$.
%  For each $1 < i < t$, the graph $G'$ obtained by bypassing the clique $K_i$ is a {\pig}.
%\end{proposition}
%%%%
%\begin{definition}[Sunflower]
%  \label{def:sunflower}
%  \end{definition}
%%%%%
%\begin{proposition}[Sunflower Lemma, \cite{erdos1960intersection, CyganFKLMPPS15}]
%  \label{prop:sunflower}
%  Let $\mathcal{F}$ be a family of sets over a universe $\mathcal{U}$, such that each set in $\mathcal{F}$ has cardinality exactly $d$. If $|\mathcal{F}| > d!(k - 1)d$, then $\mathcal{F}$ contains a sunflower with $k$ petals and such a sunflower can be computed in polynomial time in $|\mathcal{F}|$, $|\mathcal{U}|$ and $k$.
%\end{proposition}
%%%
%%%
%{\color{red}\begin{proposition}[\cite{villanger2010proper}]
%  \label{prop:poly-algo-pit}
%  Let $G$ be a simple, undirected, and connected graph that does not contain any net, tent and holes of length at most 6 as an induced subgraph.
%  Then, there is a polynomial time algorithm computing a set X of minimum size such that $G-X$ is a proper interval graph.
%\end{proposition}}
%%
%Let $G$ be a proper interval graph.
The following two lemma statements ensure us that there are two graph operations under which {\pig}s are closed.

\ifshort
\begin{lemma}[$\star$]\footnote{Due to lack of space, the detailed proofs of safeness of the reduction rules, and the detailed proofs of every lemma, observation etc. that are omitted or marked $\star$ can be found in the full version.}
\else
\begin{lemma}
  \fi
  \label{lemma:degree-2-tail-attach-to-pendant}
  Let $G$ be a simple \pitgraph, $v$ be a pendant vertex in $G$, and $\widehat{G}$ be the graph obtained by attaching a degree-2-tail $P= (v, u_1, u_2, \dots, u_{\ell})$ at $v$.
  Then $\widehat{G}$ is a simple \pitgraph.
\end{lemma}

\ifshort  
\else
\begin{figure}[ht!]
  \centering
  \label{fig:attaching_tail}
  \includegraphics[scale=0.6]{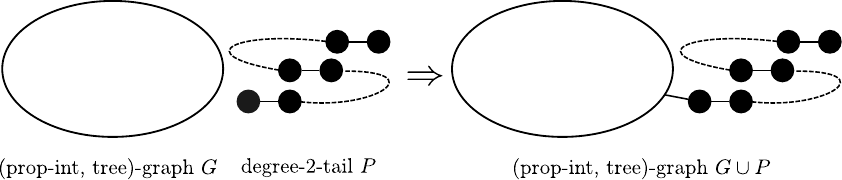}
  \caption{Schematic Diagram of \Cref{lemma:degree-2-tail-attach-to-pendant}}
\end{figure}
\fi

\ifshort
\else
\begin{proof}
Suppose that the premise of the statement is true.
We obtain $\widehat{G}$ by attaching a degree-2-tail at a pendant vertex $v$ of $G$.
It is sufficient to prove that $\widehat{G}$ is simple graph every connected component of $\widehat{G}$ is a proper interval graph or a tree.
First, observe that this graph operation does not create any self-loop.
Note that $G$ is a simple graph, and the degree-2-tail $P$ is a simple graph.
The edges of $G$ and the edges of $P$ are pairwise disjoint.
Therefore, attaching a degree-2-tail to a vertex does not create any parallel edge or self loop.
{Hence, $\widehat{G}$ does not contain a self loop.}
%{Hence, $\widehat{G}$ is a self loop.}

Now, consider the connected component $C$ to which this degree-2-tail $P$ is attached.
If $C$ is a tree, then note that $P$ is a tree on its own, and attaching $P$ into a tree $C$ does not create any cycle.
Hence, the newly constructed connected component is also a tree.
In case $C$ is a proper interval graph, then what is crucial is that the vertex $v$ to which the degree-2-tail is attached to is a pendant vertex.
As $v$ is a pendant vertex of $G$, let $u$ be the unique neighbor of $v$ in $G$.
Observe that, after construction of $\widehat{G}$, $v$ is a degree-2-vertex in $\widehat{G}$ and is not part of any cycle in $\widehat{G}$.
Hence, there does not exist any tent or net in $\widehat{G}$ containing $v$ and a subset of $P$.
Furthermore, observe that degree of $v$ and any vertex $u_i$ of $P$ is at most 2.
Hence, there cannot exist a claw containing $v$ and some vertex of $P$.
Hence, the connected component constructed by attaching $P$ to a pendant vertex of $C$ is a proper interval graph.
This completes the proof of the lemma.
\end{proof}
\fi

%%%%{\color{blue}
\ifshort
\begin{lemma}[$\star$]
\else
\begin{lemma}
\fi
\label{lem:subdiving-edges-pig}
Let $G$ be a simple {\pig}, and $P=\{v_1, v_2, v_3, v_4\}$ be a degree-2-path in $G$.
Let $G'$ be the graph obtained from $G$ by subdividing the edge $v_2 v_3$ by new vertices $Z=\{u_1,\dots, u_{\ell}\}$ such that $\ell \ge 1$ and $P'=\{v_1, v_2, u_1, \dots, u_{\ell}, v_3, v_4\}$.
Then, $G'$ is also a simple {\pig}. 
\end{lemma}

\ifshort
\else
\begin{proof}
  For the sake of contradiction, suppose that $G'$ is not a simple {\pig}.
  Then, $G'$ contains an induced cycle, a claw, a net, or a tent as an induced subgraph which intersects $Z$.
  Note that $d_{G'}(u_i)=2$ for all $i \in \{1,\dots,\ell\}$ and $G'[Z]$ induces a path.
  First, we consider the case when there is an induced cycle $C_i$ with $i \ge 4$ in $G'$ such that $C_i$ intersects $Z$.
  Then, $C_i$ must contain all vertices of $Z$ and $V(P)\subseteq C_i$.
  By contracting the edges of $P'$ one by one such that we get back the edge $v_2 v_3$, we obtain an induced cycle $C_i'=(C_i\setminus Z)\cup \{v_2v_3\}$ in $G$.
  In particular, $C_i'$ contains all vertices $v_1, v_2, v_3, v_4$.
  This contradicts that $G$ is a simple {\pig}.
  
  Next, we consider the case when $G'$ has a net or a tent or a claw.
  Note that the claws and nets of $G'$ cannot contain any vertex from $Z$ as every vertex of $Z$ and $N_{G'}(Z) (= \{v_2, v_3\})$ is of degree exactly 2 in $G'$.
  Now, consider that there is a tent $T$ in $G'$ such that $T$ intersects $Z$.
  Let's say that $u_j \in V(T) \cap Z$ for some $j \in \{1,\dots, \ell\}$.
  As $d_{G'}(u_j)=2$ and the only possible degree for a vertex in a tent is $2$ or $4$, $u_j$ must be a degree-$2$ vertex in $T$.
  It implies that both neighbors of $u_j$ in $G'$, i.e., $u_{j-1}$ and $u_{j+1}$ (or $v_2$ if $j=1$ and $v_3$ if $j=\ell$) must be in $T$.
  By the structure of a tent, the neighbor of a degree-$2$-vertex is a degree-$4$-vertex.
  Since no vertex of $Z \cup \{v_2, v_3\}$ has degree at least 4, therefore, no tent can also intersect $Z$.
  
  Finally, observe that the construction of $G'$ does not create any double-edge in a simple graph $G$.
  Therefore, $G'$ is a simple {\pig}.
\end{proof}
\fi

\ifshort
\else
\subparagraph*{Parameterized Complexity and Kernelization.} A parameterized problem $L$ is a set of instances
$(x, k) \in \Sigma^{*}\times \mathbb{N}$ where $\Sigma$ is a finite alphabet and $k \in \mathbb{N}$ is a parameter.
%The notion of `tractability’ in parameterized complexity is defined as follows.
A parameterized problem $L \subseteq \Sigma^* \times \mathbb{N}$ is fixed-parameter tractable (or FPT) if given $(x, k) \in \Sigma^{*}\times \mathbb{N}$, there is an algorithm $\mathcal{A}$ that correctly decides if $(x, k) \in L$ in $f(k)|x|^{\mathcal{O}(1)}$-time for some computable function $f : \mathbb{N}\rightarrow \mathbb{N}$. This algorithm $\mathcal{A}$ is called fixed-parameter algorithm (or FPT algorithm) for the problem $L$.
%%
%%%
%{\color{red}\begin{definition}[Fixed-ParameterTractability, \cite{CyganFKLMPPS15}]
%  A parameterized problem $L$ is said to be fixed-parameter tractable (or FPT) if given $(x, k) \in \Sigma^{*}\times \mathbb{N}$, there is an algorithm $\mathcal{A}$ that correctly decides if $(x, k) \in L$ in $f(k)|x|^{\mathcal{O}(1)}$-time for some computable function $f : \mathbb{N}\rightarrow \mathbb{N}$. This algorithm $\mathcal{A}$ is called fixed-parameter algorithm (or FPT algorithm) for the problem $L$.
%\end{definition}}
%%%
Observe that we allow combinatorial explosion in the parameter $k$ while the algorithm runs in polynomial-time with respect to $|x|$.
We say that two instances $(x, k)$ of $L$ and $(x', k')$ of $L$ are equivalent when $(x, k) \in L$ if and only if $(x', k') \in L$.
The notion of kernelization (preprocessing) for in the area of parameterized complexity is defined as follows.
A parameterized problem $L \subseteq \Sigma^{*} \times \mathbb{N}$ is said to admit a kernelization if there exists an algorithm that, given an instance $(x, k)$ of $L$, outputs in time polynomial in $|x| + k$ an equivalent instance $(x', k')$ of $L$ such that $|x'| + k' \leq g(k)$ for some computable function $g : \mathbb{N} \rightarrow \mathbb{N}$. The instance $(x', k')$ is called a kernel for the instance $(x, k)$. If $g(k)$ is $k^{\mathcal{O}(1)}$, then we say that $L$ admits a polynomial kernel.
%%%
%%%
%\begin{definition}[Kernelization, \cite{CyganFKLMPPS15}]
%  A parameterized problem $L \subseteq \Sigma^{*} \times \mathbb{N}$ is said to admit a kernelization if there exists an algorithm that, given an instance $(x, k)$ of $L$, outputs in time polynomial in $|x| + k$ an equivalent instance $(x', k')$ of $L$ such that $|x'| + k' \leq g(k)$ for some computable function $g : \mathbb{N} \rightarrow \mathbb{N}$. The instance $(x', k')$ is called a kernel for the instance $(x, k)$. If $g(k)$ is $k^{\mathcal{O}(1)}$, then we say that $L$ admits a polynomial kernel.
%  
%\end{definition}
%%%
A kernelization algorithm usually consists of a collection of reduction rules that have to be applied exhaustively. A reduction rule is called safe if given an instance $(x, k)$ of $L$, one application of the reduction rule outputs an equivalent instance $(x', k')$ of $L$. 
It is well-known due to \cite{downey2013fundamentals,cygan2015parameterized,fomin2019subquadratic} that a decidable parameterized problem is FPT if and only if it admits a kernelization.
\fi

\section{Polynomial Kernel for {\PITVD}}
\label{sec:pitvd-kernel}

This section is devoted to the proof of our main result.
Let $(G, k)$ be the input instance.
We assume without loss of generality that $G$ has no self-loop, as otherwise, we can always remove them.
Our first step is to invoke a polynomial-time algorithm by Jacob et al. \cite{JACOB2023280} that outputs a set $\widehat{S}$ such that $G - \widehat{S}$ is a {\pitgraph} and $S$ is a 7-approximate solution.
Clearly, if $|\widehat{S}| > 7k$, we can immediately output that $(G, k)$ is a no-instance.
Hence, we can assume that $|\widehat{S}| \leq 7k$.
%%
%\begin{proposition}[\cite{JACOB2023280}, Theorem 12]
%\label{prop:approximation}
%{{\PITVD}} admits a $7$-approximation algorithm. 
%\end{proposition}
%%%
Now, with the help of this 7-approximate solution $\widehat{S}$ and Proposition \ref{prop:sunflower-apply}, we first focus on the small obstructions, i.e., nets, tents, and induced cycles $C_\ell$, $4\le \ell \le 6$ to obtain a set ${S}$ such that $G - {S}$ is a {\pitgraph} and $\widehat{S} \subseteq S$.
It allows us to forget about the small obstructions of size at most $6$ and focus on the large obstructions like {\ctpair s} and holes of size at least $7$ in the next step.
Let us define a new family $\mathcal{F}$ of obstruction such that it contains net, tent, and holes $C_4$, $C_5$ and $C_6$.
Our next lemma prepares us a set $S$ such that $\widehat{S} \subseteq {S}$ and ensures that it is sufficient to intersect $C_4, C_5, C_6$, nets, and tents in $G[S]$ only.

\ifshort
\begin{lemma}[$\star$]
\else
\begin{lemma}
  \fi
  \label{lem:small-obstruction}
  Let $(G,k)$ be an instance to {\PITVD}.
  Then there is a polynomial time algorithm that either concludes that $(G, k)$ is a no-instance for {\PITVD}, or
  finds a non-empty set $S \subseteq V(G)$ such that
  \begin{itemize}
    \item $G - S$ has no net, tent, and induced cycles $C_\ell$, $4\le \ell \le 6$;
    \item Every set $Z \subseteq V(G)$ of size at most $k$ is a minimal hitting set for nets, tents, and induced cycles $C_\ell$, $4\le \ell \le 6$, in $G$ if and only if it is a minimal hitting set for nets, tents, and induced cycles $C_\ell$, $4\le \ell \le 6$, contained in $G[S]$; and
    \item $|S| \le 6\cdot 6!(k+1)^6 + 7k$.
  \end{itemize}
\end{lemma}

\ifshort
\else
\begin{proof}
  In the first step, we invoke the 7-approximation algorithm by Jacob et al. \cite{JACOB2023280} and obtain a vertex subset $\widehat{S}$ such that $G - \widehat{S}$ is a {\pitgraph}.
  If $|\widehat{S}| > 7k$, then we output that $(G, k)$ is a no-instance.
  Otherwise, $|\widehat{S}| \leq 7k$.
  Next, we apply \Cref{prop:sunflower-apply} on $\mathcal{F}$, and it outputs a set $\mathcal{F}'$ such that
  \begin{itemize}
      \item $X$ is a minimal hitting set of $\mathcal{F}$ of size at most $k$ if and only if $X$ is a minimal hitting set of $\mathcal{F}'$ of size at most $k$; and
      \item $|\mathcal{F}'|\le 6!(k+1)^6$.
  \end{itemize}
  Let $S^* = \bigcup\limits_{F \in \mathcal{F}'} V(F)$.
  Clearly, $|S^*|\le 6\cdot 6!(k+1)^6$.
  It implies that every set $Z \subseteq V(G)$ of size at most $k$ is a minimal hitting set for nets, tents, and induced cycles $C_\ell$, $4\le \ell \le 6$, in $G$ if and only if it is a minimal hitting set for nets, tents, and induced cycles $C_\ell$, $4\le \ell \le 6$, contained in $G[\widehat{S}]$.
  Then, $G- S^*$ does not contain any net, tent, and induced cycles $C_\ell$, $4\le \ell \le 6$.
  We set $S = S^* \cup \widehat{S}$.
  Observe that $G - S$ is a {\pitgraph}, and it satisfies the properties that hitting every net, tent, and holes with at most 6 vertices of $G[\widehat{S}]$ is good enough.
  This completes the proof of the lemma, and we have that $|S| \leq 6\cdot 6!(k+1)^6 + 7k$.
\end{proof}
\fi
\ifshort
We henceforth fix the set $S$ obtained from the lemma.
Then, $|S| = \cOO(k^6)$ and $G-S$ is a \pitgraph.
We will use the following immediate observation throughout.
\else
From now onward we work with the set $S$ that we obtained from the above lemma.
Observe that $|S|$ is $\cOO(k^6)$ and $G - S$ is a {\pitgraph}.
%%%Now, we describe our preprocessing rules as follows.
We use the following observation throughout the paper the proof of which is trivial.
\fi
\begin{observation}{\label{obs:reduction}}
  For any subset $Z\subseteq V(G)$, if $(G, k)$ is a yes-instance for PITVD with solution $X$, then $(G - Z, k)$ is a yes-instance for PITVD with solution $X \setminus Z$.
\end{observation}

% \begin{figure}[ht!]
%     \centering
%     \label{fig:G-Z-instance}
%     \includegraphics[scale=0.4]{FIGURES/G-Z-instance.pdf}
%     \caption{Schematic Diagram of \Cref{obs:reduction}}
% \end{figure}
We assume without loss of generality that $G$ contains no connected component that is a proper interval graph or a tree, since any such components can be removed. 
Hence, we assume that every component of $G - S$ (when it is a proper interval graph or a tree) is adjacent to at least one vertex in $S$.
Let $V_1$ denote the set of vertices of the connected components of $G - S$ that have a cycle.
Let $V_2$ denote the set of vertices of the connected components of $G - S$ that are trees.
Observe that every connected component of $G[V_1]$ is a proper interval graph and has at least one triangle.
The details of our kernelization algorithm are divided into 3 critical parts.
In Section \ref{sec:initial-rules}, we illustrate some reduction rules that ensure us some important structure to the entire graph.
Subsequently, in Section \ref{sec:bounding-tree-components}, we provide some reduction rules that ensure us that $G[V_2]$ has $\cOO(|S|^2)$ vertices and in Section \ref{sec:bounding-pig-components2}, we provide some reduction rules that ensure us that $G[V_1]$ has $\cOO(k^{33})$ vertices.

\subsection{Initial Preprocessing Rules}
\label{sec:initial-rules}
%%%initial-rules
We have some preprocessing rules that are common for the entire graph that we describe in this section.
It is important to note that some of our subsequent reduction rules may generate components in $G - S$ that is a tree or a proper interval graph but has no neighbor in $S$.
%We call such components {\em isolated} in $G - S$.
Therefore, for the sake of completeness, we state the following reduction rule that removes every connected component of $G$ that is a proper interval graph or a tree.
%%% remove any connected component $C$ in $G - S$ that has no neighbor in $S$.
This deletion is safe because every obstruction for PITVD is connected induced subgraph, and intersect with $S$. 
Hence, no obstruction can be entirely contained within a component $C$ in $G - S$ that has no neighbor in $S$.
%%\todo[inline]{Diptapriyo: should we use `obstruction' or `forbidden subgraph'?}
%%\todo[color=green]{Arpit: I think `obstruction' sounds better. It is easy to understand and write.}
%%%%%%%%%%%%%%%%%%%%%%%%%%%%%%%%%%%%%%%%%%%%%%%%%%%%%%%%%%%%%%%%%%%%%%%%%%%%%%%%%%%%%%%%%%%%%
\begin{reductionrule}
\label{rul:isolated-component-reomval}
  If there exists a connected component $C$ in $G$ that is a simple proper interval graph or a tree, then delete $C$ from $G$.
  The new instance is $(G - C, k)$.
\end{reductionrule}

{Our next two reduction rules also apply to the {\sc (Clique, Tree)-Vertex Deletion} problem. 
One is the standard buss-kernelization rule for {\sc Vertex Cover} whose safeness is trivial.
The other reduces edge multiplicity to at most 2 (standard for {\sc Feedback Vertex Set}) to obtain a simple graph after deleting at most $k$ vertices. See \cite{CyganFKLMPPS15,downey2013fundamentals} for details.}
%The authors also used this rule.

\begin{reductionrule}
  \label{rule:multiplicity-reduction}
  If there exists an edge in $G$ with multiplicity more than 2, then reduce the multiplicity of that edge to exactly 2.	
\end{reductionrule}

\begin{reductionrule}
  \label{rule:buss-rule-double-edge}
  If there is a vertex $v$ such that $v$ is adjacent to $k+1$ vertices $u_1, u_2,\ldots,u_{k+1}$ such that for every $i \in [k+1]$, $vu_i$ is a double-edge, then the new instance is $(G - v, k-1)$.	
\end{reductionrule}

%%%%\todo[inline, color=green]{Arpit: I have swaped the position of \Cref{rule:multiplicity-reduction} and \Cref{rule:buss-rule-double-edge}}
%\noindent
{Our next two reduction rules contract the degree-2-paths to at most length $3$: degree-2-tails are contracted to length $1$ and the other degree-2-paths contracted to length $3$.}

\begin{reductionrule}
  {\label{rule:degree-2-tail}}
  Let $P = \{v_1, v_2, \ldots, v_{\ell}\}$ be a {degree-2-tail} of with $\ell \geq 3$ vertices in $G$ and
  % such that every internal vertex $v_i$ ($2 \leq i \leq \ell-1$) has degree exactly 2 in $G$ and $d_G(v_1) > 2$.
   $Z = \{v_3, v_4, \ldots, v_{\ell}\}$.
  Then, the new instance is $(G - Z, k)$.
  % is a yes-instance if and only if $(G, k)$ is a yes-instance.
\end{reductionrule}

%\ifshort
%\begin{lemma}[$\star$]
%\else
%\begin{lemma}
%\fi
%  \Cref{rule:degree-2-tail} is safe.
%\end{lemma}

% \begin{figure}[ht!]
%     \centering
%     \label{fig:degree-2-tail}
%     \includegraphics[scale=0.4]{FIGURES/degree-2-tail.pdf}
%     \caption{Schematic Diagram of \Cref{rule:degree-2-tail}}
% \end{figure}

\ifshort
%{\color{blue}
The vertices of $Z$ in the degree-2-tail $P$ does not need to contain any vertex of a solution: after deleting the vertices of $Z$, $v_2$ has degree exactly one. 
% , and the tail can therefore be reattached as a pendant path.
% Such an attachment preserves the required component class, by \Cref{lemma:degree-2-tail-attach-to-pendant}.
For both trees and proper interval graphs, due to \Cref{lemma:degree-2-tail-attach-to-pendant}, attaching a degree-2-tail preserves the structure of the respective graph classes.
Hence, the reduction rule is safe.
Our next reduction rule contracts a degree-2-path to length at most 3 as follows.

\else
\begin{lemma}
  \Cref{rule:degree-2-tail} is safe.
\end{lemma}
\begin{proof}
  From \Cref{obs:reduction}, if $(G, k)$ is a yes-instance, then $(G-Z, k)$ is also a yes-instance.
  Hence, the forward direction ($\Rightarrow$) follows.

  For the backward direction ($\Leftarrow$), let $X$ be a solution to $(G-Z, k)$. 
  Note that $G[Z]$ is a degree-2-tail.
  What is crucial here is that in the graph $G - (X \cup Z)$, $v_2$ either does not exist or is an isolated vertex {or} a pendant vertex.
  Subsequently, $G - X$ is constructed by attaching a pendant tail $Z$ at $v_2$.
  In case, ${v_2} \in X$, then ${v_2}$ does not exist in $G - (X \cup Z)$.
  Then $Z$ is a connected component in $G - X$ that is a path which is a tree.
  In the second case, if ${v_2}$ is isolated in $G - (X \cup Z)$, then $Z \cup \{v_2\}$ forms a connected component in $G - X$ that is a path which is a tree.
  In the third case, {$v_2$ is a pendant vertex in $G-(X\cup Z)$.}
  As $G - (X \cup Z)$ is a {\pitgraph}, and a degree-2-tail $Z$ is attached to $v_2$, it follows from \Cref{lemma:degree-2-tail-attach-to-pendant} that $G - X$ is a {\pitgraph}.
  Hence, $X$ is a solution to $(G, k)$.
\end{proof}
\fi

%%%%%%%%%%%%%%%%%%%%%%%%%%%%%%%%%%%%%%%%%%%%%%%%%%%%%%%%%%%%%%%%%%%%%%%%%%%%%%%%%%%%%
\begin{reductionrule}{\label{rule:degree-2-overbridge}}
  Let $ P = (v_1, v_2, \dots, v_{\ell}) $ be a degree-2-path that is not a degree-2-tail in $ G $ such that $ \ell \geq 5$. %and $v_1 \neq v_{\ell}$.
  % and every internal vertex $ v_i $ ($ 2 \leq i \leq \ell-1 $) has degree exactly 2 in $ G $.
  Then, define $Z = \{v_3, v_4, \dots, v_{\ell-2}\} $ and consider $G'$ be the graph obtained from $G$ as follows: remove all vertices in $Z$; then add the edge $v_2 v_{\ell - 1}$.
%  \begin{itemize}
%    \item Remove all vertices in $Z $;
%    \item Add the edge $ v_2v_{\ell-1} $.
%  \end{itemize}
 %% Then, the instance $ (G, k) $ is a yes-instance for the PITVD if and only if the instance $ (G', k) $ is a yes-instance.
\end{reductionrule}

%\todo[inline]{Diptapriyo: I feel that keeping it degree-2-path is better.}
%\todo[color=green, inline]{Arpit: I think we have already defined degree-2-overbridge in prelim section. Hence, it implicitly says that $v_1$ and $v_{\ell}$ are distinct. If we keep just degree-2-path then there will be ambiguity between tail or overbridge. I think it's better to keep it as overbridge.}

\ifshort
\else
\begin{figure}[ht!]
    \centering
    \label{fig:degree-2-overbridge}
    \includegraphics[scale=0.6]{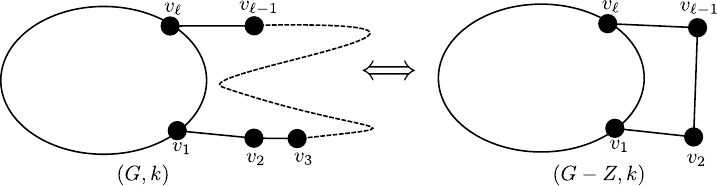}
    \caption{Schematic Diagram of \Cref{rule:degree-2-overbridge}}
\end{figure}
\fi
The safeness of the above reduction rule crucially relies on the next lemma.

\ifshort
\begin{lemma}[$\star$]
\else
\begin{lemma}
\fi
  \label{lem:degree-2-overbridge}
  Let $G'$ be the graph obtained from \Cref{rule:degree-2-overbridge} and $Z$ be the same vertex subset defined in \Cref{rule:degree-2-overbridge}.
%  {Let $ P = (v_1, v_2, \dots, v_\ell) $ is a \emph{degree-2-overbridge} in $ G $ such that $ \ell \geq 5 $, where every internal vertex $ v_i $ ($ 2 \leq i \leq \ell-1 $) has degree exactly~2 in~$ G $. Define $Z = \{v_3, v_4, \dots, v_{\ell-2}\}$.
%  Let $ G' $ be the graph obtained from $ G $ by performing the following operations:
%  \begin{itemize}
%      \item Remove all vertices in $ Z $;
%      \item Add the edge $ v_2v_{\ell-1} $.
%  \end{itemize}}
  Then, the following holds:
  \begin{enumerate}[(i)]
    \item\label{overbridge-Z_disjoint} Every claw, net, tent, and double-edge of $G$ is disjoint from $Z$.
    \item\label{overbridge-edge_disjoint} Every claw, net, tent, and double-edge of $G'$ is disjoint from edge $v_2v_{\ell-1}$.
    \item\label{overbridge-same-obstruction} The set of claws, nets, tents, and double-edges of $G$ and $G'$ are the same.
%    \item\label{double-edges-same} The set of double edges of $G$ and $G'$ are the same.
  \end{enumerate}
\end{lemma}

\ifshort
\else
\begin{proof}
  We prove the items in the given order.
  \begin{enumerate}[(i)]
    \item Observe that every vertex of $Z$ has degree exactly $2$ in $G$ and the vertices of $N_G(Z)=\{v_2, v_{\ell -1}\}$ also have degree $2$ in $G$.
    In short, every vertex of $N_G[Z]$ has degree exactly 2 in $G$.
    %Since tents, claws, and nets of a graph have at least one degree-3-vertex and every vertex of $Z$ is degree exactly 2 in $G$, hence there does not exist any claw, net, or tent that intersect $Z$.
    First we consider the case that there is a double edge $xy$ of $G$ that intersects $Z$.
    But that is not possible because every $z\in Z$ has two distinct neighbors along $P$, hence a double-edge incident to $z$ would already give $d_G(z)\ge3$.
    Hence, every double-edge of $G$ is disjoint from $Z$.
    
    Next, we consider the case that there is a claw $J$ in $G$ that intersects $Z$.
    A claw $J$ contains a degree-3-vertex $x$ and a pendant vertex $y$.
    If $J$ intersects $Z$, then the pendant vertex $y$ must be in $Z$ and its unique neighbor $x$ in $J$ must be in $N_G[Z]$.
    But this implies that a vertex of $N_G[Z]$ which is from $\{v_2, v_{\ell - 1}\}$ has degree at least 3 in $G$.
    This leads to a contradiction that $v_2$ and $v_{\ell - 1}$ have degree exactly 2 in $G$.
    Hence, $J$ cannot intersect $Z$.

    Now, we consider the case that there is a net $N$ in $G$ that intersects $Z$.
    Then, $N$ has an edge $xy$ with two endpoints $x, y$ being from $N_G[Z]$.
    If both $x$ and $y$ are of degree at least 3 in the net, then both $x$ and $y$ must have degree at least 3 in $G$, leading to a contradiction that $x$ and $y$ have degree exactly 2 in $G$.
    %\replace{Without loss of generality, let $x$ be a pendant in net, then y must be a vertex at least 3 in $G$, leading to a contradiction that $y$ is a vertex with degree exactly 2 in $G$.}
    If $x$ is a pendant in net, then $y$ must be a vertex of degree at least 3 in $G$, leading to a contradiction that $y$ is a vertex with degree exactly 2 in $G$.
    Similarly, if $y$ is a pendant, then $x$ must be of degree at least 3 in $G$ that leads to a contradiction that $y$ is of degree exactly 2 in $G$.
    Hence, $N$ cannot be isomorphic to a net that intersect $N_G[Z]$, implying that $N$ {is disjoint from} $Z$.
    
    Finally, we consider the case when a tent $T$ that intersects $Z$.
    Then, $T$ has an edge $xy$ with two endpoints $x, y$ being from $N_G[Z]$.
    Note that when $xy$ is an edge of a tent in $G$, then it must be that either both $x$ and $y$ must have degree at least 4 in $G$; or one of them has degree 2 and the other has degree at least 4 in $G$.
    But, that contradicts the fact that both $x$ and $y$ have degree exactly 2 in $G$.
    Hence, $T$ cannot be a tent in $G$ that intersect $N_G[Z]$, implying that $T$ is disjoint from $Z$.
    
    Therefore, our claim follows.    
    %\todo[inline]{Diptapriyo: the above proof looks shortened too much. Should we elaborate it on case by case basis as done in the next item?}
    \item For the sake of contradiction, suppose that there is a double-edge, or a claw, or a net, or a tent in $G'$ that contains the edge $v_2v_{\ell-1}$.
    By construction of $G'$, note that both $v_2$ and $v_{\ell - 1}$ have degree exactly 2 in $G'$.
     
    First, note that if a double-edge of $G'$ has to contain the edge $v_2 v_{\ell - 1}$, then there must be two edges between $v_2$ and $v_{\ell - 1}$, but that is not possible because $v_2v_{\ell-1}\notin E(G)$.
    Hence, the double-edges of $G'$ cannot contain $v_2 v_{\ell - 1}$. 
    
    Next, we consider the case that there is a claw in $G'$ that contains the edge $v_2v_{\ell-1}$.
    Then, either $v_2$ or $v_{\ell - 1}$ is the center of this claw.
    {Let $v_2$ be the center. Then, $v_2$ must have degree at least 3 leading to a contradiction that the degree of $v_2$ is exactly 2 in $G'$.
    By similar arguments, we can argue that $v_{\ell - 1}$ cannot be the center of the claw.}
    %{If $v_2$ is the center, then $v_2$ must have degree at least 3 leading to a contradiction that the degree of $v_2$ is exactly 2 in $G'$.
    %Similarly, if $v_{\ell - 1}$ is the center of the claw, then $v_{\ell - 1}$ must have degree at least 3 leading to a contradiction that $v_{\ell - 1}$ is a vertex of degree exactly 2 in $G'$.}
    
    Now, we consider the case that there is a net that contains the edge $v_2 v_{\ell - 1}$.
    Every vertex of a net is either a pendant or a vertex of degree three.
    If both $v_2$ and $v_{\ell -1}$ are of degree at least 3 in the net, then both $v_2$ and $v_{\ell - 1}$ must have degree at least 3 in $G'$, leading to a contradiction that $v_2$ and $v_{\ell -1}$ have degree exactly 2 in $G'$.
    {Without loss of generality, let $v_2$ be a pendant in net. Then, $v_{\ell - 1}$ must be a vertex of degree at least 3 in $G'$, leading to a contradiction that $v_{\ell - 1}$ is a vertex with degree exactly 2.}
    %%{If $v_2$ is a pendant in net, then $v_{\ell - 1}$ must be a vertex of degree at least 3 in $G'$, leading to a contradiction that $v_{\ell - 1}$ is a vertex with degree exactly 2.
    %%Similarly, if $v_{\ell - 1}$ is a pendant, then $v_2$ must be of degree at least 3 in $G'$ that leads to a contradiction that $v_2$ is of degree exactly 2 in $G'$.}
    Hence, $v_2 v_{\ell - 1}$ cannot be part of a net.
    
    Finally, we consider the case when $v_2 v_{\ell - 1}$ is an edge of a tent.
    Note that if $v_2 v_{\ell - 1}$ is an edge of a tent, then it must be that either both $v_2$ and $v_{\ell-1}$ must have degree at least 4 in $G'$; or one of them has degree 2 and the other has degree at least 4 in $G'$.
    But, that contradicts the fact that both $v_2$ and $v_{\ell - 1}$ have degree exactly 2 in $G'$.
    Hence, there cannot exist any tent in $G'$ containing $v_2 v_{\ell - 1}$. 
    \item We use the correctness of items-(\ref{overbridge-Z_disjoint}) and (\ref{overbridge-edge_disjoint}) to prove this item.
    It follows from the item-(\ref{overbridge-Z_disjoint}) that all claws, nets, and tents of $G$ are disjoint from $Z$.
    From the item-(\ref{overbridge-edge_disjoint}), it follows that the set of all claws, nets, and tents of $G'$ are disjoint from the edge $v_2v_{\ell-1}$.
    
    Now, we consider a vertex subset $O$ that induces a double-edge, or a claw, or a net, or a tent in $G$.
    Due to item-(\ref{overbridge-Z_disjoint}), $O$ is disjoint from $Z$.
    Hence, $O$ cannot contain any of the vertices in $Z$.
    In particular, $O$ cannot contain the edges $v_2 v_3$ and $v_{\ell - 2} v_{\ell -1}$.
    Therefore, $O$ is contained in $G'$.
    
    Similarly, due to item-(\ref{overbridge-edge_disjoint}), if $O$ is a vertex subset that induces a double-edge, or a claw, or a tent, or a net of $G'$, then $O$ cannot contain the edge $v_2 v_{\ell - 1}$.
    Then, the edges of $G[O]$ are disjoint from the edges of $\{v_2 v_3, v_3 v_4,\ldots,v_{\ell - 2} v_{\ell -1}\}$.
    It means $O$ is disjoint from $Z$.
    
    This completes the justification that the set of claws, nets, and tents of $G$ and $G'$ are the same.
%    
%    \item Observe that a double edge in $G$ cannot have any vertex of $N_G[Z]$ as endpoint.
%    Note that 
  \end{enumerate}
This completes the proof of the lemma.
\end{proof}
\fi
%%%%%%%%%%%%%%%%%%%%%%%%%%%%%%%%%%%%%%%%%%%%%%%%%%%%%%%%%%%%%%%%%%%%%%%%%%%%%%%%%%%%%
% \begin{reductionrule}{\label{rule:degree-2-overbridge}}
%   Let $ P = (v_1, v_2, \dots, v_\ell) $ is a \emph{degree-2-overbridge} in $ G $ such that $ \ell \geq 5 $, where every internal vertex $ v_i $ ($ 2 \leq i \leq \ell-1 $) has degree exactly 2 in $ G $. Define $ Z = \{v_3, v_4, \dots, v_{\ell-2}\} $. Let $ G' $ be the graph obtained from $ G $ by performing the following operations:
%   \begin{itemize}
%       \item Remove all vertices in $ Z $;
%       \item Add the edge $ v_2v_{\ell-1} $.
%   \end{itemize}
%   Then, the instance $ (G, k) $ is a yes-instance for the PITVD if and only if the instance $ (G', k) $ is a yes-instance.
% \end{reductionrule}
%\todo[inline, color=green]{Arpit: I have moved the \Cref{rule:degree-2-overbridge} before \Cref{lem:degree-2-overbridge}.}
%\ifshort
%\begin{lemma}[$\star$]
%\else
%\begin{lemma}
%\fi
%  \Cref{rule:degree-2-overbridge} is safe.
%\end{lemma}

\ifshort
A degree-2-path, that is not a degree-2-tail, can be shortened to four vertices because its internal vertices do not participate in any obstruction other than holes and paths connecting a claw and a triangle.
What is crucial is that even after shortening a degree-2-path to a degree-2-path of 4 vertices, the every hole still remains a hole of length at least 4.
%Hence, replacing the internal vertices of the degree-2-path by the edge $v_2v_{\ell-1}$ does not create or destroy any of the obstructions.
%Moreover, any induced cycle or a path that connects a \ctpair\ that uses the new edge can be expanded by subdividing it with the contracted degree-2 vertices, and conversely, any such structure in the original graph can be compressed to the new edge.
Therefore, a solution can be transferred between $G$ and $G'$ without increasing its size.

\else
\begin{lemma}
  \Cref{rule:degree-2-overbridge} is safe.
\end{lemma}
\begin{proof}
  Note that $P$ is a degree-2-path but not a degree-2-tail.
  Since the {degree-2-path} $P$ with at least 5 vertices has been converted to a degree-2-path of four vertices in $G'$, let $P'= (v_1, v_2, v_{\ell-1}, v_{\ell})$ be the {degree-2-path} of four vertices.

  For the forward direction ($\Rightarrow$) of the proof, let $X$ be a solution to $(G, k)$.
  We define $X'$ as follows:
  % \replace{
  %   \[X' = (X\setminus N_G[Z])\cup A_X\]
  %   where, $A_X = (\{v_1\} \text{ if } X\cap(Z\cup \{v_2\}\neq \emptyset))\cup (\{v_\ell\}\text{ if } v_{\ell-1}\in X)$
  % }
  \begin{itemize}
    \item If $v_2 \in X$, set $X' = (X \setminus \{v_2\}) \cup \{v_{1}\}$.
    \item If $v_{\ell-1} \in X$, set $X' = (X \setminus \{v_{\ell-1}\}) \cup \{v_{\ell}\}$.
    \item If $v_2, v_{\ell-1}\in X$, set $X' = (X \setminus \{v_2, v_{\ell-1}\}) \cup \{v_{1}, v_{\ell}\}$.
    \item If $X\cap Z\neq\emptyset$, set $X' = (X \setminus Z) \cup \{v_{1}\}$.
    \item Else, set $X' = X$.
  \end{itemize}
  
  Clearly, $|X'|\le |X|$.
  Observe that $v_2, v_{\ell-1}\notin X'$ and $Z\cap X'=\emptyset$.
  In particular, $N[Z] \cap X' = \emptyset$.
  We claim that $X'$ is a solution to $(G', k)$.
  Targeting a contradiction, suppose that $G' - X'$ is not a {\pitgraph}.
  Then, $G'-X'$ contains a double-edge or some connected component of $G' - X'$ is neither a tree nor a proper interval graph.
  Due to item-(\ref{overbridge-same-obstruction}) of \Cref{lem:degree-2-overbridge}, the set of double-edges, nets, and tents of $G$ and $G'$ are the same.
  In particular, due to item-(\ref{overbridge-edge_disjoint}) of \Cref{lem:degree-2-overbridge}, the nets, and tents of $G'$ are disjoint from the edge $v_2 v_{\ell - 1}$.
  Hence, $X'$ intersects the set of double-edges, nets, and tents of $G'$. 

  Consider an induced cycle $C$ with at least 4 vertices in $G'-X'$.
  Note that $C$ was not present in $G - X$ but is present in $G' - X'$.
  Then, $C$ must contain $v_2$ or $v_{\ell - 1}$.
  As $C$ is chordless, $C$ must contain both $v_1$ and $v_{\ell}$ as well.
  In particular, being an induced cycle that contains $v_2$ implies that the same induced cycle must also contain $v_{\ell - 1}$.
  It implies that $v_1, v_{\ell} \notin X'$.
  But then $v_1, v_{\ell} \notin X$ as well.
  In particular, $(Z \cup \{v_2, v_{\ell-1}\}) \cap X = \emptyset$.
  Then, $C$ contains all $v_1, v_2, v_{\ell-1}, v_{\ell}$.
  Consider the cycle $C^*$ that is obtained from $C$ by subdividing the edge $v_2 v_{\ell - 1}$ by the vertices $v_3, v_4,\ldots, v_{\ell-2}$ (vertices from $Z$).
  This cycle $C^*$ is also an induced cycle and appears in $G - X$ since $X$ is disjoint from the entire degree-2-overbridge.
  This leads to a contradiction that $G - X$ is a {\pitgraph}.
  
%  If $C_i$ contains $v_2$ or $\replace{v_{\ell-1}}{v_3}$, then it must contain both $v_2$ and $\replace{v_{\ell-1}}{v_3}$.
%  Then, $C_i$ must contain all of $v_1, v_2, \replace{v_{\ell-1}}{v_3}, \replace{v_{\ell}}{v_4}$.
%  As $C_i$ is in $G' - X'$, $v_1, v_{\ell} \notin X'$.
%  By construction of $X'$, note that $X \cap P = \emptyset$.
%  Consider the cycle $D$ that we get from $C_i$ by subdividing the edge $v_2 v_{\ell - 1}$ by $\ell - 4$ vertices.
%  Observe that $D$ is an induced cycle of length at least 4 in $G$ that contains $P$ and $X$ does not contain any vertex from $P$.
%  Also, by construction, $X \setminus P = X' \setminus P$.
%  Hence, $X'$ does not contain any vertex from $D \setminus P$.
%  Combining these two facts, $X$ is disjoint from $D \cap P$ and $X$ is disjoint from $P$.
%  Then, $D$ is an induced cycle of length at least 4 in $G - X$, contradicting that $G - X$ is a {\pitgraph}.
  Finally, we consider the case when $G' - X'$ contains a connected component $D$ that has both a claw $J$ and a triangle $T$ as induced subgraphs.
  Consider a shortest path $Q$ between $J$ and $T$ in $D$.
  If $Q$ contains $v_2$ (or $v_{\ell - 1}$), then note that the claws of $G'$ do not contain the edge $v_2 v_{\ell-1}$.
  Similarly, $v_2$ and $v_{\ell - 1}$ are not part of any triangles of $G'$.
  Hence, $Q$ contains $v_1, v_2, v_{\ell - 1}, v_{\ell}$.
  By construction of $X'$ from $X$, $v_1, v_{\ell} \notin X'$ means that $v_2, v_{\ell - 1} \notin X$ and $Z$ is disjoint from $X$.
  Hence, $N_G[Z]\cap X = \emptyset$.
  Then, we subdivide the path $Q$ by inserting vertices of $Z$ into the edge $v_2 v_{\ell - 1}$ and obtain a path $\widehat{Q}$ in $G$.
  Also, $X \setminus P = X' \setminus P$.
  Therefore, $J \cup T$ is also disjoint from $X$.
  Hence, there is a connected component of $G - X$ that contains triangle $T$, claw $J$ and a path $\widehat{Q}$ between them.
  This contradicts that $G - X$ is a {\pitgraph}.
 
 	As we have argued for all exhaustive cases, $X'$ is a solution to $(G', k)$.

	For the backward direction ($\Leftarrow$), let $X'$ be a solution of size at most $k$ to $(G', k)$. 
	We claim that $G - X'$ is a {\pitgraph}.
  %For the sake of contradiction, suppose that $G-X'$ is not a {\pitgraph}.
  %Then, $G-X'$ contains a net, or a tent, or a cycle $C_i$ with $i \ge 4$ as an induced subgraph or has a connected component that has a claw and a triangle as induced subgraphs.
  %Such an obstruction also exists in $G$. 
  From \Cref{lem:degree-2-overbridge}, item-(\ref{overbridge-same-obstruction}), the set of nets, tents, claws, and double-edges in $G$ and $G'$ are the same.
  Therefore, $X'$ intersects all the nets, tents, and double-edges of $G$.
  Let us examine how the graphs $G' - X'$ and $G - X'$ are different.
  The vertices of $Z$ are the only ones that are in $G - X'$ but not in $G' - X'$.
  \begin{description}
    \item[Case-(a)] If $v_1, v_{\ell} \in X'$, then one of the following situations arise.
    One situation is when $Z$ forms an additional connected component of $G - X'$ itself that is a path.
    Other components of $G - X'$ are unchanged.
    Hence, $G - X'$ is a {\pitgraph}.
    If the vertices of $Z\cup \{v_2\}$, or $Z\cup \{v_{\ell-1}\}$, or $Z\cup \{v_2, v_{\ell - 1}\}$ forms a connected component of $G-X'$, then a single connected component of $G'-X'$ is converted into a path and the other components remains unchanged.
  %  {If the vertices of $Z \cup \{v_2\}$ forms a connected component of $G - X'$, then a {singleton} connected component of $G' - X'$ is converted to a path.
  %  Similarly, if the vertices of $Z \cup \{v_{\ell -1}\}$ or the vertices of $Z \cup \{v_2, v_{\ell - 1}\}$ forms a connected component of $G - X'$, then that component also is converted into a path from a single vertex or an edge.
  %  The other components remain unchanged.}
    Hence, $G - X'$ is a {\pitgraph}. 
    \item[Case-(b)] If $v_1 \in {X'}$ but $v_{\ell} \notin X'$.
    Consider the connected component $C$ of $G - X'$ that contains $Z$.
    If $Z = C$, then the other connected components of $G - X'$ and $G' - X'$ are the same.
    This takes care of the case when $v_{\ell - 1} \in X'$.
    In that case, $C$ is a path.
    Hence, $G - X'$ is a {\pitgraph}.
    
    Consider when $v_{\ell - 1} \notin X'$.
    If $Z \subseteq C$ and $C$ contains $v_{\ell}$, then $C$ contains $v_{\ell -1}$.
    Consider the connected component $D$ of $G' - X'$ that contains $v_{\ell}$ and $v_{\ell - 1}$.
    If $v_2 \notin D$, then $v_{\ell-1}$ is a pendant vertex in $D$.
    In such a case, we obtain a component $C$ of $G - X'$ by attaching the degree-2-tail $Z \cup \{v_{\ell-1}\}$ into the pendant vertex $v_{\ell - 1}$.
    As $G' - X'$ is a {\pitgraph}, due to \Cref{lemma:degree-2-tail-attach-to-pendant}, attaching a pendant tail to $G' - X'$ creates $G - X'$ that is also a {\pitgraph}.
    In case $v_2 \in D$, then $v_2 v_{\ell - 1}$ is an edge and $v_2$ is a pendant vertex.
    We obtain the component $C$ by subdividing the edge $v_2 v_{\ell - 1}$ into the component $D$.
    This procedure is equivalent to first deleting the pendant vertex $v_2$, and then attaching a degree-2-tail $Z \cup \{v_2, v_{\ell - 1}\}$ into the vertex $v_{\ell - 1}$.
    Note that deleting $v_2$ creates a {\pitgraph} (due to hereditary property) and $v_{\ell - 1}$ is a pendant vertex.
    Subsequently, due to lemma \ref{lemma:degree-2-tail-attach-to-pendant}, attaching a degree-2-tail to a pendant vertex ensures that $G - X'$ is a {\pitgraph}.
    %Note that $C' = C \setminus Z$ is a component in $G' - X'$.
    %Note that $P \setminus \{v_1, v_{\ell}\}$ and $P \setminus \{v_1, v_2, v_{\ell}\}$ both are degree-2-tails attached to $v_{\ell - 1}$.

    \item[Case-(c)] If $v_1 \notin X'$ but $v_{\ell} \in {X'}$, then the case is symmetric to subcase-(b).
    
    \item[Case-(d)] If $v_1, v_{\ell} \notin X'$, then we can have the following situations.
    %%%\todo[inline]{Diptapriyo: this case is yet to be verified by me.}
    In case, $v_2, v_{\ell - 1} \notin {X'}$, then $v_1, v_2, v_{\ell - 1}, v_{\ell}$ is part of a connected component $D'$ of $G' - X'$.
    If $D'$ is a tree, then subdividing the edge $v_2 v_{\ell - 1}$ of $D'$ by the vertices of $Z$ forms a tree $D$ that is a component in $G - X'$.
    Other components remain unchanged.
    So, $G - X'$ is a {\pitgraph}.
    If $D'$ is a proper interval graph, and it contains a degree-2-path $P'$, therefore there cannot exist any path between $v_1$ and $v_{\ell}$ in $D$ avoiding the vertices $v_2$ and $v_{\ell - 1}$.
    As otherwise, we will have an induced cycle of length at least 4 in $D'$.
    %\todo[inline]{Diptapriyo: is this statement trivial or does it follow directly?}
    %\todo[inline, color=green]{Arpit: Yes, This statement is trivial. Direct implication of $X'$ being the solution and cycle will be created with both paths having $\{v_2, v_{\ell-1}\}$.}
    Therefore, {$P'$ is a unique path between $v_1$ and $v_{\ell}$ in $D'$ such that $v_2, v_{\ell-1}\in P'$.}
    %{$P'$ is a unique path between $v_1$ and $v_{\ell}$ in $D'$}.
    We create $D$ by subdividing the edge $v_2 v_{\ell - 1}$ using the vertices of $Z$.
    Note that $D$ is a proper interval graph by \Cref{lem:subdiving-edges-pig}.
    %\todo[inline]{Diptarpiyo: need a lemma that this preserves proper interval property.}
    %\todo[inline, color=green]{Arpit: Yes, I have added a \Cref{lem:subdiving-edges-pig} for this in preliminaries.}
    
    In case, without loss of generality, $v_2 \notin X'$ but $v_{\ell - 1} \in X'$ (or $v_2 \in X'$ but $v_{\ell - 1} \notin X'$), then $v_2$ (or $v_{\ell - 1}$) is a pendant vertex in $G' - X'$.
    %\todo[inline, color=green]{Arpit: $v_2$ can not be isolated in $E(G'-X')$ as $v_1v_2\in E(G'-X')$}
    %{\color{red}If $v_2$ is an isolated vertex in $G' - X'$, then $Z \cup \{v_2\}$ is a component in $G - X'$ that is a path.
    %Other components of $G-X'$ remains unchanged.}
    %{Other components of $G - X'$ are the same as components of $G' - X'$ except the vertex $v_2$.}
    As $v_2$ (or $v_{\ell - 1}$) is a pendant vertex in $G' - X'$, then by \Cref{lemma:degree-2-tail-attach-to-pendant}, attaching a degree-2-tail $Z$ to the pendant vertex $v_2$ (or $v_{\ell - 1}$) in $G' - X'$ creates $G - X'$ that is also a {\pitgraph}.
    Hence, $G' - X'$ is a {\pitgraph}.
    %  Consider an induced cycle $D$ of length at least 4 in $G - X'$.
    %  If $D$ contains any vertex from $Z \cup \{v_2, v_{\ell - 1}\}$, then $D$ contains the degree-2-overbridge $P$, then contracting the edges of $P$ one by one, we get a degree-2-overbridge $P' = (v_1, v_2, v_{\ell - 1}, v_{\ell})$ in $G' - X'$.
    %  This creates a cycle $D'$ in $G' - X'$ that is an induced cycle with at least 4 vertices, contradicting that $X'$ is a solution to $(G', k)$.
    %  If $D$ is disjoint from $Z$, then $D$ also cannot contain the edges $v_1 v_2$ and $v_{\ell - 1} v_{\ell}$.
    %  Hence, $D$ is an induced cycle of length at least 4 in $G' - X'$ as well.
    %  In this case also, we get a contradiction that $G' - X'$ is a {\pitgraph}.
    %  
    %  Now, we consider the case that $G - X'$ contains a connected component $D$ that contains a claw $C$ and a triangle $T$ as induced subgraph.
    %  Let $Q$ be a shortest path between $C$ and $T$ in $D$.
    %  If $Q$ contains a vertex from $Z$, then $Q$ contains $P$.
    %  In that case, contracting the edges of $P$ creates a connected component $D'$ in $G' - X'$ containing both claw and triangle as induced subgraphs.
    %  This contradicts that $G' - X'$ is a 
  \end{description}
  Therefore, $X'$ is a solution to $(G, k)$, completing the proof of the lemma.
\end{proof}
\fi

% Let $x$ be a cut vertex in $G$ such that $G - x$ has a component $C$ that is a tree.
% If $C \cup \{x\}$ induces a tree, then $C$ is a {\em pendant tree} in $G$ and $C$ is {\em attached to $x$}.
% It is well-known from \cite{JacobMZ2024isaac} that $C$ is a pendant tree attached to $x$ if and only if $E(V(G) \setminus C, C)$ is a singleton set.
% Consider a pendant tree $C$ {of $G$} attached to $x\in V(G)$ such that $x$ is a cut vertex.
% Suppose that there is at least one vertex in $C$ that has degree at least 3 in $G$.
% Our next reduction rule chooses to keep only a few vertices of this pendant tree.
We borrow the idea of `pendant tree' defined in \cite{JacobMZ2024isaac}.
Formally, let $x$ be a cut vertex in $G$ and $C$ a tree component of $G-x$.
If $C\cup\{x\}$ induces a tree, then $C$ is a {\em pendant tree} attached to $x$.
%It is known that $C$ is a pendant tree attached to $x$ if and only if $E(V(G) \setminus C, C)$ is a singleton set \cite{JacobMZ2024isaac}.
If a pendant tree $C$ attached to cut vertex $x$ has a vertex of degree at least 3 in $G$, then our next reduction rule keep a claw in $C$ that is closest $x$.

\begin{reductionrule}
  \label{rul:pendant-tree}
  Let $C$ be a pendant tree attached to $x$.
  % \replace{Let $C$ be a pendant tree attached to $x$.}{Let $C$ be a pendant tree of $G\left[V_{2}\right]$ such that $x\in S$ is a unique $S$-neighbor of $C$.}
  % There exists a unique vertex $w\in C$ such that $d_{G}(w)\ge 3$ and is closest to $x$.
  If there is a vertex in $C$ that has degree at least three in $G$, then we consider the set of vertices $H_C = \{v \in C \mid d_G(v) \geq 3\}$.
  %\todo[inline]{Diptapriyo: this $x$ in bracket is not required, once $C$ is fixed, it is clear.}
  Choose $v \in H_C$ to be a vertex that has the smallest distance from $x$ among all vertices of $H_C$ (since $C \cup \{x\}$ induces a tree, $v$ is unique).
  Choose a set $D$ of vertices that contains the unique path $P_{x, v}$ between $v$ and $x$ in $G[C \cup \{x\}]$, and two arbitrary neighbors of $v$ in $C$ not appearing in $P_{v, x}$.
  Then, delete all vertices of $C \setminus D$ from $G$ and $G' = G - (C \setminus D)$.
  % $P_{xw} \cup {y} \cup {z}$ where $y, z\in (N(w)\cap C)\setminus P_{xw}$.
  The new instance is $(G', k)$.
\end{reductionrule}

%\ifshort
%\begin{lemma}[$\star$]
%\else
%\begin{lemma}
%\fi
%  \label{lem:pendant-tree-safe}
%  \Cref{rul:pendant-tree} is safe.
%\end{lemma}

\ifshort
\else
\begin{lemma}
  \Cref{rul:pendant-tree} is safe.
\end{lemma}
\begin{proof}
  The forward direction ($\Rightarrow$) of this reduction rule is trivial due to \Cref{obs:reduction}.

  We give the backward direction ($\Leftarrow$) of the proof.
  For the ease of notation, let $\widehat{D} = C \setminus D$, and $X$ be a set of at most $k$ vertices such that $G - (\widehat{D} \cup X)$ is a {\pitgraph}.
  We argue that either $G - X$ is also a {\pitgraph} or there exists $\widehat{X}$ of size at most $k$ such that $G - \widehat{X}$ is a {\pitgraph}.
  
  If $G - X$ is a {\pitgraph}, then we are done.
  Else consider the case that $G - X$ is a not a {\pitgraph}.
  {It implies that there exists a vertex in $\widehat{D}$ that intersects with an obstruction.}
  First and most crucial observation is that $\widehat{D} \subseteq C$ such that $C$ is a pendant tree.
  Hence, no double-edge is incident to the vertices of $\widehat{D}$.
  As $G - (X \cup \widehat{D})$ has no double-edge, it follows that $G - X$ cannot have a double-edge.
  Additionally, the vertices of $\widehat{D}$ do not participate in any hole of $G$.
  As $G - (\widehat{D} \cup X)$ does not contain a hole, there cannot exist any hole in $G - X$ that contains the vertices of $\widehat{D}$.
  Similarly, the vertices of $\widehat{D}$ do not participate in any triangle of $G$.
  Hence, the vertices of $\widehat{D}$ does not participate in any triangle of $G - X$.
  Then, the only possible case that $G - X$ is not a {\pitgraph} is that there is a component $\widehat{C}$ of $G - X$ that contains a claw $J$ as well as a triangle $T$.
  We analyze the structure of the component $\widehat{C}$ along with the $J$ and $T$.

  First crucial observation is that the triangle $T$ appears in $G - (\widehat{D} \cup X)$ and $T$ cannot intersect any vertex of $C$.
  Consider a shortest path $P_{J, T}$ between $T$ and $J$ in $G - X$.
  Additionally, due to the structure that $G[C \cup \{x\}]$ is a tree, $x$ separates $C$ from $G - (C \cup \{x\})$, and any vertex of $C$ is a cut-vertex of the connected component of $G$ containing $C$, any path from $T$ to any vertex $w \in \widehat{D}$ contains the vertices of $P_{x, w}$ that is a unique path between $x$ and $w$ in $G[C \cup \{x\}]$.
  Due to the construction of $G'$, observe that the vertices $D$ contains a path $P_{x, v}$ from $xS$ to a vertex $v \in C$ such that $v \in H_C$.
  %\todo[inline, color=green]{Arpit: This red colored text is not needed. $x$ is a cut-vertex that connects pendant tree to the rest of the graph. It need not be from $S$. I think, we have also missed mentioning that $X$ does not intersect $D$ otherwise $\{v, x_1, x_2, x_3\}$ does not induce claw in $G-X$ or the above blue colored text suffice already. Kindly confirm it.}
  In addition, the choice of $D$ ensures that $v$ has three neighbors in $D$ and each of those three neighbors are in $C \cup \{x\}$.
  Let $x_1, x_2, x_3$ be the neighbors of $v$ in $D$ such that $x_1, x_2, x_3 \in C \cup \{x\}$.
  Then, $\{v, x_1, x_2, x_3\}$ induces a claw.
  In addition, the choice of $v \in H_C$ is such that no other vertex of $H_C$ is closer to $x$ in $G$.
  It implies that $v$ separates $x$ to any other vertex of $H_C$ in $G$ and $D \cup \{x\}$ induces a connected subgraph which is a tree.
  Hence, the shortest path $P_{J, T}$ from triangle $T$ to claw $J$ must pass through $P_{x, v}$.
  \begin{description}
    \item[Case-(i):] $J$ is the claw $\{v, x_1, x_2, x_3\}$.
    In this case, $J \subseteq D$.
    As $D$ is a connected subgraph and {$D\subseteq C$}, observe that the triangle $T$, the claw $J$, and the path $P_{J, T}$ completely avoids the vertices of $\widehat{D}$.
    Then, the vertices of $J \cup T$ and the path $P_{J, T}$ is in a connected component of $G - (\widehat{D} \cup X)$ that is neither a tree or nor a proper interval graph.
    This is a contradiction to the assumption that $X$ is a solution to the graph $G - \widehat{D}$.
    
    \item[Case-(ii):] $J$ is not the claw $\{v, x_1, x_2, x_3\}$ but a claw that intersects $\widehat{D}$.
    But then $P_{J, T}$ passes through $P_{x, v}$.
    Then, $P_{J, T}$ also contains one of the neighbors $x_1$ or $x_2$ or $x_3$.
    Suppose that $P_{J, T}$ contains $x_1$.
    Observe that $\{v, x_1, x_2, x_3\}$ is a claw present in $G - \widehat{D}$ and is also connected to $T$ by a path from $T$ to $v$ which contains the subpath $P_{x, v}$.
    Hence, $X$ must have contained $x_2$ or $x_3$.
    We consider $\widehat{X} = (X \setminus \{x_2, x_3\}) \cup \{x\}$.
    Observe that $|\widehat{X}| \leq k$.
    Note that in the graph $G - \widehat{X}$, every component that gets created by deletion of $x$ are induced subgraphs of $G - (X \cup \widehat{D})$.
    The other components of $G - X$ are the pendant trees attached to $x$ or some induced subgraph of them.  
    %Then, we consider $J' = \{v, x_1, x_2, x_3\}$.
    %Note that $J'$ also induces a claw and any path between $T$ and $J$ must intersect $J'$.
    %Hence, $J'$ appears in shorter distance to $T$ than $J$ in terms of distance in $G$.
    %But, then the triangle $T$, the claw $J'$, and the unique path $P_{x, v}$ is part of the graph $G - (\widehat{D} \cup X)$.
    %Then, $G - (X \cup \widehat{D})$ has a connected component that contains both claw and a triangle.
    %This contradicts that $X$ is a solution to $G - \widehat{D}$.
    Hence, $(G, k)$ has a solution of size at most $k$.
  \end{description}
  As both the cases above lead to a contradiction, it implies that $G - X$ is a {\pitgraph}.
  Hence, the reduction rule is safe.
\end{proof}
\fi

\ifshort
\begin{lemma}[$\star$]
\else
\begin{lemma}
\fi
  \label{lem:C-pendant-tree-property}
  Let $C$ be a pendant tree attached to $x$.
  % \replace{Let $C$ be a pendant tree attached to $x$.}{Let $C$ be a pendant tree of $G\left[V_{2}\right]$ with unique $S$-neighbor $x \in S$.}
  If $G[C \cup \{x\}]$ is not a path and Reduction Rules \ref{rule:degree-2-overbridge}--\ref{rul:pendant-tree} are not applicable, then all but at most 5 vertices are deleted from $C$.
  % \replace{If $G[C \cup \{x\}]$ is not a path and none of the above reduction rules are applicable, then  all but at most 5 vertices are deleted from $C$.}{Then, the following statements hold true.}
  %\begin{enumerate}[(i)]
  %	\item\label{when-pendant-tree-path} If $G[C\cup \{x\}]$ is a path and none of the Reduction Rules \ref{rul:isolated-component-reomval}-\ref{rule:non-nbr-leaves} are not applicable, then the \Cref{rul:pendant-tree} is not applicable to $C$ and $C$ is a singleton vertex.
  %	
  %	\item\label{when-pendant-tree-has-deg-3-vertex} 
  %	If $G[C \cup \{x\}]$ is not a path and none of the above reduction rules are applicable, then  all but at most 5 vertices are deleted from $C$.
  %\end{enumerate}
\end{lemma}

\ifshort
\else
\begin{proof}
  If $C$ is a pendant tree attached to $x$ such that $G[C \cup \{x\}]$ is not a path, then $C$ contains a vertex of degree at least 3 in $G$.
  Let $v \in C$ be the vertex of degree at least 3 that is nearest to $x$ and $P_{x, v}$ be the unique path between $x$ and $v$ in $G$.
  Hence, every vertex of $P_{x, v} \setminus \{x, v\}$ has degree exactly 2 in $G$.
  \Cref{rul:pendant-tree} chooses to keep the vertices of $P_{x, v}$ and 2 additional neighbors $w_1, w_2$ of $v$ in $C$ that are not in $P_{x, v}$.
  Subsequently, \Cref{rul:pendant-tree} deletes the other vertices from $C$.
  Note that after deletion of the other vertices, $w_1$ and $w_2$ become pendant vertices in $G$.
  Now, we focus on the number of vertices in $P_{x, v}$.
  If $P_{x, v}$ has exactly $r+2$ vertices including $x$ and $v$ for some $r \geq 0$, then $P_{x, v}$ has exactly $r$ internal vertices (and exactly $r$ vertices).
  In particular, each internal vertex of $P_{x, v}$ has degree exactly 2 in $G$, because $v \in C$ is the nearest vertex to $x$ that has degree at least 3 in $G$.
  Hence, $P_{x, v}$ is a degree-2-path of $G$ with $r$ internal vertices.
  As \Cref{rule:degree-2-overbridge} is not applicable, $r \leq 2$.
  Subsequently, as \Cref{rul:pendant-tree} is not applicable, in addition to the vertices of $P_{x, v}$, two additional neighbors of $v$ are kept, and the other vertices are deleted from $C$.
  As $P_{x, v}$ contains at most 3 vertices that includes $v$ but excludes $x$, and two additional neighbors of $v$ in $C$ are kept, therefore, at most 5 vertices in total are kept in $C$ after application of \Cref{rul:pendant-tree}.
  Therefore, all but at most 5 vertices are deleted from $C$.
\end{proof}
\fi

Consider a cut vertex $x$.
Our next reduction rules ensures that it is sufficient to keep at most 3 pendant trees attached to $x$.
It also ensures us that a claw containing $x$ and the vertices of pendant trees attached to $x$ is preserved.

\begin{reductionrule}
  \label{rul:pendant-trees}
  Let $x$ be vertex such that there are at least $\ell \geq 4$ different pendant trees attached to $x$ and $Z=\cup_{j=4}^{\ell} C_j$.
  Then, the new instance is $(G - Z, k)$.
  % \replace{Let $x$ be vertex such that there are at least $\ell \geq 4$ different pendant trees attached to $x$ and $Z=\cup_{j=4}^{\ell} C_j$.
  % Then, the new instance is $(G - Z, k)$.}{Let a vertex $x \in S$ be the unique $S$-neighbor of at least $\ell$ different pendant trees $C_1, \dots, C_\ell$ of $G\left[V_{2}\right]$ such that $\ell \geq 4$ and $Z=\cup_{j=4}^{\ell} C_j$.
  % Then, the new instance is $(G - Z, k)$.}
\end{reductionrule}

%\ifshort
%\begin{lemma}[$\star$]
%\else
%\begin{lemma}
%\fi
%\label{lemma:rule-3-pendant-trees-safe}
%\Cref{rul:pendant-trees} is safe.
%\end{lemma}

\ifshort
The forward direction ($\Rightarrow$) of the safeness of the \Cref{rul:pendant-trees} follows from \Cref{obs:reduction}.
The backward direction ($\Leftarrow$) follows from the following intuitive reason.
Suppose that $X$ is a solution of size at most $k$ to $(G-Z,k)$ but not to $(G,k)$.
Then, an obstruction must use some vertex of $Z$.
Since the vertices of each $C_j$ induce a pendant tree attached only to $x$, they cannot participate in holes, tents, double-edges.
Thus, the obstruction can only arise from a net or from a component containing both a claw and a triangle.
However, at least three pendant trees are attached to $x$ in $G'$ and $x \notin X$. 
Therefore, whenever such an obstruction; net, or a (claw, triangle)-pair appears in $G - X$ containing $x$, such an obstruction also appears in $G' - X$.
Then, at least one vertex from the 3 pendant trees $C_1 \cup C_2 \cup C_3$ is already present in $X$ while $x \notin X$.
By arguments similar to the case of \Cref{rul:pendant-tree}, $(X \setminus (C_1 \cup C_2 \cup C_3)) \setminus \{x\}$ is also a solution of size at most $k$ to $(G, k)$
% .}{uses one of the deleted pendant trees, another pendant tree attached to $x$ can serve the same purpose.
% Hence, deleting all but three pendant trees cannot affect the existence of a solution, which gives a contradiction.}

\else
\begin{lemma}
  \Cref{rul:pendant-trees} is safe.
\end{lemma}
\begin{proof}
 If $X$ is a solution to $(G, k)$, then due to \Cref{obs:reduction}, $X$ is a solution to $(G - Z, k)$.
  Therefore, the forward direction ($\Rightarrow$) follows.

  Now, to prove the backward direction ($\Leftarrow$), assume that $X$ is an optimal solution to $(G - Z, k)$ of size at most $k$.
  We prove that $(G, k)$ has a set of at most $k$ vertices whose deletion results in a {\pitgraph}.
  Suppose for the sake of contradiction that $X$ is not a solution to $(G, k)$.
  Then, there exists an obstruction $O$ in $G - X$ which intersect $Z$.
  Observe that the vertices of $O \cap Z$ cannot participate in any cycle.
  Hence, $O$ cannot become a tent or an induced cycle of length at least 4 or a double-edge.
  It means that $O$ is a net, or a claw $J$ along with a triangle $T$ appearing in the same connected component of $G - X$.
  \begin{description}
  	\item[Case-(i):] $O$ is a net.
  	Then, $O$ must contain a pendant vertex the neighbor of which participates in a triangle.
  	It implies that $O$ contains $x$ that participates in a triangle.
  	Then, $O$ contains a vertex $y \in Z$ such that the edge $xy$ separates $Z$ from $G - Z$.
  	In such a case, choose $w \in C_1$ such that $xw \in E(G)$.
  	Observe that $(O \setminus \{y\}) \cup \{w\}$ is a net in $G - Z$.
  	As $X$ is a solution to $(G - Z, k)$ and $x \notin X$, it must be that $w \in X$.
  	We construct $\widehat{X} = (X \setminus \{w\}) \cup \{x\}$.
  	
  	Observe that any component of $(G - (Z \cup \widehat{X})$ is $C_1$ or an induced subgraph of another component $D$ of $G - (X \cup Z)$.
  	Since $G - (X \cup Z)$ is a {\pitgraph}, and it is hereditary, it follows that  $(G - (Z \cup \widehat{X})$ is a {\pitgraph}.
  	Since $|\widehat{X}| = |X|$, therefore $\widehat{X}$ is a solution to $(G - Z, k)$.
  	
  	Now, observe that $G - \widehat{X}$ is a disjoint union of $G - (Z \cup \widehat{X})$ and the connected components in $G[Z]$.
  	Hence, $G - \widehat{X}$ is a {\pitgraph}.
  	Therefore, there exists a solution of size $k$ to $(G, k)$.  	
  	\item[Case-(ii):] There is a connected component $D$ of $G - X$ that contains a triangle $T$, and a claw $J$.
  	Since $G - (\widehat{X} \cup Z)$ did not have any component with both claw and a triangle, observe that $D$ has been created because of the vertices of $Z$ being added.
  	Furthermore, note that $x$ separates the vertices of $Z$ from the vertices of $G - (X \cup {Z})$ and $G[Z]$ is a forest.
  	Hence, $D$ contains $x$.
  	As $G[Z]$ is a forest and the vertices of $Z$ does not intersect in any triangle, it means that there are vertices in $Z$ that intersects the claw $J$ in $G[D]$.
  	Observe that every path from $J$ to $T$ passes through $x$.
  	But, then we consider a claw $J'$ in $G$ with center $x$ and its three neighbors $y_1 \in C_1, y_2 \in C_2$, and $y_3 \in C_3$.
  	Observe that $C_1, C_2, C_3$ are in $G - Z$ also, hence this claw $J'$ is also present in $G - Z$.
  	Since the triangle $T$ is disjoint from $X$, a path between $T$ and $x$ is present in $G - (X \cup Z)$, and $X$ is a solution to $(G - Z, k)$, it must be that $y_1 \in X$ or $y_2 \in X$ or $y_3 \in X$.
  	We construct $\widehat{X} = (X \setminus (C_1 \cup C_2 \cup C_3)) \cup \{x\}$ and claim that $\widehat{X}$ is a set of at most $k$ vertices such that $G - \widehat{X}$ is a {\pitgraph}.
  	
  	As $x \in \widehat{X}$, every pendant tree $C_i$ is a connected component of $G - \widehat{X}$.
  	Subsequently, let us compare the graphs $G - (X \cup Z)$ and $G - \widehat{X}$ other than the components $C_1,\ldots,C_{\ell}$.
  	Note that the connected component of $G - (X \cup Z)$ that contains $x$ has been broken into multiple components when we consider the graph $G - \widehat{X}$.
  	That particular component of $G - (X \cup Z)$ containing $x$ is a proper interval or a tree.
  	Due to hereditary property, after being broken also, it remains a proper interval graph or a tree.
  	But the other components of $G - (X \cup Z)$ that do not intersect $C_1 \cup C_2 \cup C_3$ are the same as the components of $G - \widehat{X}$ that are disjoint from $C_1,\ldots,C_{\ell}$.
  	Since every component of $G - \widehat{X}$ is a proper interval graph or a tree, $G - \widehat{X}$ is a {\pitgraph}.
  	
  	We consider $|\widehat{X}|$. 
  	Note that $X \cap (C_1 \cup C_2 \cup C_3) \neq \emptyset$.
  	Hence, at least one vertex will be removed from $X$ and exactly one vertex is added when constructing $\widehat{X}$.
  	Therefore, $|\widehat{X}| \leq |X| \leq k$.
  	This completes the proof that $G$ has a set of at most $k$ vertices the deletion of which results in a {\pitgraph}.
  \end{description}
As the cases are mutually exhaustive, this completes the proof of the lemma.
\end{proof}
\fi

%\section{Kernelization Algorithm}

%\input{initial-rules-2.tex}
\subsection{Bounding the Tree Components of \texorpdfstring{$G - S$}{G - S}}
\label{sec:bounding-tree-components}
%%% tree-part-bounding
This section is devoted to provide an upper bound on the number of vertices in $G\left[V_{2}\right]$.
Note that $G\left[V_{2}\right]$ is the collection the tree components of $G - S$.
%We achieve this upper bound by obtaining that both $N_G(S) \cap V_2$ and $V_2 \setminus N_G(S)$ have $\cOO(|S|^2)$ vertices.
%%
Let $F_1 = N_G(S) \cap V_2$, and $F_2 = V_2 \setminus F_1$.
First, we design a reduction rule that helps us to bound the number of vertices in $F_2$.
We partition the vertices of $F_2$ as follows.
Consider two distinct vertices $x, y \in F_1$ such that there is a (unique) path between $x$ and $y$ in $G[V_2]$.
We say that a vertex $u \in F_2$ is {\em $(x, y)$-connecting} if $u$ appears in the unique path between $x$ and $y$ in the graph $G[V_2]$.
We use $F_3$ to denote the vertices that are $(x, y)$-connecting for any two vertices of $F_1$, and focus on the vertices of $F_3$ and $F_2 \setminus F_3$ separately.

% \ifshort
% \else
\begin{figure}[ht]
  \centering
  \includegraphics[scale=0.5]{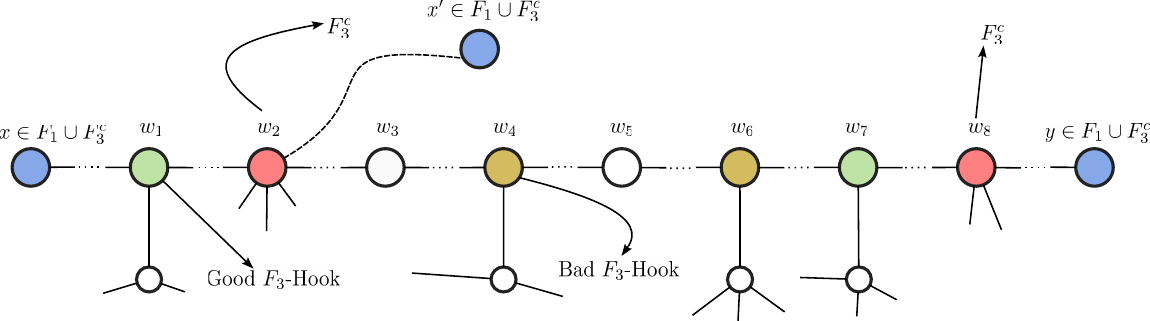}
  \caption{An illustration of good (green) and bad (yellow) $F_3$-hooks. White vertices: those on the $(x,y)$-path are in $F_3$; others are $w_i$'s hanger. Red vertices are in $F_3^c$.}
  \label{fig:hook}
\end{figure}
% \fi

%%%
Now, we consider the graph $G[F_3 \cup F_1]$.
As every vertex of $F_3$ is $(x, y)$-connecting for $x, y \in F_1$, if a vertex $w$ is from $F_3$, then $w$ must have at least two neighbors in $G[F_1 \cup F_3]$.
Hence, the number of leaves in $G[F_3 \cup F_1]$ is at most $|F_1|$.
If a vertex $w \in F_3$ has at least three neighbors in $G[F_3 \cup F_1]$, then we say that $w$ is {\em $F_3$-critical}.
We use $F^c_3$ to denote the set of all $F_3$-critical vertices.
It is not hard to observe that $|F_3^c| \leq |F_1|$.
Additionally, it is possible that a vertex $w \in F_3$ that has degree exactly 2 in the graph $G[F_1 \cup F_3]$, but there is a pendant tree $C$ attached to $w$ such that $C \subseteq G[V_2]$.
If $w \in F_3$ has degree exactly 2 in $G[F_1 \cup F_3]$ and there is a pendant tree $C$ attached to $w$ such that $C \subseteq G[V_2]$, then {we say that $w$ is an {\em $F_3$-hook} and $C$ is a {\em $w$-hanger}.}
Observe that the set of all $F_3$-hooks is precisely a subset of $F_3 \setminus F_3^c$.
%%%
Consider an $F_3$-{hook} $w$ that appears in the unique path $P_{x, y}$ such that $x, y \in F_1 \cup F_3^c$.
We say that $w$ is {\em good} if
\begin{itemize}
	\item $w$ is closest to $x$ among all possible $F_3$-{hooks} that appear in $P_{x, y}$; or
	\item $w$ is  closest to $y$ among all possible $F_3$-{hooks} that appear in $P_{x, y}$.
\end{itemize} 

If an $F_3$-{hook} $w$ is not good, then we say that $w$ is a {\em bad} $F_3$-{hook}.
We refer to \Cref{fig:hook} for an illustration.
Observe that if $w$ is a bad $F_3$-hook, there exist $x, y \in F_1 \cup F_3^c$ such that $w$ appears in the unique path $P_{x, y}$ between $x$ and $y$ in $G[V_2]$.
In addition, there are two good $F_3$-hooks $v$ and $u$ that also appear in $P_{x, y}$.
Hence, there is a claw incident at $v$ (respectively, incident at $u$) that appears closer to $x$ (respectively, closer to $y$) than the claw incident to $w$.
Our next reduction rule ensures that the claw incident to $v$ and the claw incident to $u$ are preserved.
This is crucial to the intuition that our next reduction rule is safe.
%It marks important hangers and makes a modification to the graph as follows.

\begin{reductionrule}
  \label{rule:remove-F3-hangers}
  Let ${w} \in F_3$ be a bad $F_3$-{hook} and $C$ be a $w$-hanger.
  Then, we delete $C$ from $G$.
  % \replace{$C$ from $G$}{every $w$-hanger from $G$}.
  The new instance is $(G', k)$.
\end{reductionrule}

%\ifshort
%\begin{lemma}[$\star$]
%\else
%\begin{lemma}
%\fi
%\label{lem:remove-F3-hangers-safe}
%  \Cref{rule:remove-F3-hangers} is safe.
%\end{lemma}

\ifshort
%{\color{blue}The forward direction ($\Rightarrow$) of the safeness of the \Cref{rule:remove-F3-hangers} is obvious.
%%%
%The backward direction ($\Leftarrow$) follows from the following intuitive reason.
%Suppose that $X$ is a solution to $(G',k)$..
%Then, an obstruction in $G-X$ must intersect $C$.
%Since $C$ is a $w$-hanger attached to the bad $F_3$-hook $w$, the vertices of $C$ cannot participate in any hole, net, tent, or double-edge.
%Hence, the only possible obstruction is a component containing both a claw and a triangle.
%However, by the definition of a bad $F_3$-hook, any such claw involving $C$ can be replaced by the corresponding structure through $w$, which is retained in $G'$.
%Thus, the same obstruction would also exist in $G'-X$, contradicting that $X$ is a solution to $(G',k)$.
%\todo[inline]{Diptapriyo: this is not a good written intuition process. It is not that simple enough.}
%}

\else
\begin{lemma}
  \Cref{rule:remove-F3-hangers} is safe.
\end{lemma}
\begin{proof}
  %%%% Let $\mathcal{W_H}= \{W_1, \dots, W_{\ell}\}$ such that $\ell\le 3$ (\Cref{rul:pendant-trees}) be the collection of all $w$-hangers in $G$ and $W_i$, for $i\in [\ell]$, is a pendant tree.
  Let $C$ be a $w$-hanger such that $w \in F_3$ is a bad $F_3$-hook.
  Clearly, from \Cref{obs:reduction}, if $(G, k)$ is a yes-instance, then $(G- C, k)$ or $(G', k)$ is also a yes-instance.
  Hence, the forward direction ($\Rightarrow$) follows.

  Now, for the backward direction ($\Leftarrow$), we assume that $X$ be a solution of size $k$ such that $G' - X$ is {\pitgraph}.
  % \replace{a solution of size $k$ such that $G' - X$ is {\pitgraph}.}{an solution to $(G', k)$.}
  Observe that $G - (C \cup X)$ is the same graph as $G' - X$.
  If $G - X$ is a {\pitgraph}, then we are done.
  {Otherwise, when $G - X$ is not a {\pitgraph}, we modify $X$ to construct an alternate solution $X'$ of size at most $|X|$.}

  Consider the case when $G - X$ is not a {\pitgraph}.
  Then, there is a connected component $D$ of $G - X$ that is neither a tree nor a proper interval graph.
  Since $G' - X$ is a {\pitgraph}, but $G - X$ is not, the component $D$ intersects $C$.
  As $C$ is a pendant tree attached to $w \in D$, therefore it follows that $D$ is the only connected component that is neither a proper interval graph nor a tree.
  As $w \in F_2$ and $w$ is $(x, y)$-connecting for $x, y \in F_1$, it must be that $w$ is not part of a triangle.
  Hence, the vertices of $C$ cannot participate in any net or a tent.
  In addition, the vertices of $C$ cannot participate in any induced cycle of length at least 4.
  Therefore, $D$ must have a triangle as well as a claw such that the claw intersects the vertices of $C$.
  Consider a {\ctpair} $(J, T)$ that is {closest} in the connected component $D$.
  Note that $J$ must intersect $C$ and a path from $J$ to $T$ passes through a vertex of $F_1$.
  Let $x \in F_1$ be a vertex such that there is a path from the claw $J$ to the triangle $T$ via $w$ and the vertex $x \in F_1$.
  %%%\todo[inline]{Diptapriyo: edit to continue.}

  Additionally, observe that $C$ is a pendant tree attached to $w \in D$.
  Then, $D \setminus C$ is also connected and $D \setminus C$ is a connected component of $G' - X$.
  Since $D \setminus C$ has a triangle, it must be that $D \setminus C$ is a proper interval graph.

  %{\color{red}We consider the proper interval ordering of the vertices of $D \setminus C$.
  %Since $G[V_2]$ is a tree, the vertices of $(D \setminus C) \cap V_2$ induces a disjoint union of paths.}
  %\todo[inline, color=green]{Arpit: I think this red colored statement is not needed as it is not being used anywhere.}
  Additionally, $(D \setminus C) \cap V_2$ contains a vertex $x \in F_1$ and a bad $F_3$-hook $w$.
  Then, there is a good $F_3$-hook $v \in F_2$ such that $v$ appears in a path between $x$ and $w$ in $D \setminus C$.
  We prove the following claim that is very crucial to our proof.

  \begin{claim}
  \label{claim:v-separates-w-from-x}
  Let $x \in F_1 \cap D$ be a vertex such that there is a path from the claw $J$ to a triangle of $D$ passing through $x$ in the connected component $D$.
  %Then, $1 \leq |\widehat{F_1}| \leq 2$.
  Then, there exists a good $F_3$-hook $v \in D \cap {F_3}$ such that removal of $v$ makes $D$ disconnected.
  In particular, $v$ disconnects $w$ from $x$. 	
  \end{claim}

  \begin{claimproof}
    Let $x \in F_1 \cap D$ be a vertex such that there is a path from claw $J$ to a triangle of $D$ passing through ${x}$ in the connected component $D$.
    Observe that claw $J$ intersects $C$ such that $C$ is attached to a bad $F_3$-hook $w$.
    In addition, observe that $w$ separates $C$ from $D \setminus (C \cup \{w\})$ and there is a path from $J$ to a triangle $\widehat{T}$ in $D$ via the vertex $x$.
    Let $P(\widehat{T}, J)$ denotes the path.
    Then, $P(\widehat{T}, J)$ uses the unique path $P_{x,w}$ from $w$ to $x$ in $G[V_2]$.
    As $w$ is a bad $F_3$-hook and $w$ is $(x, y)$-connecting {for some $y\in F_1$}, there exists a good $F_3$-hook $v$ such that $v$ appears closer to $x$ than $w$ in the path $P_{x, w}$.

    Now, we consider the graph $D \setminus C$.
    We claim that removal of $v$ makes $D$ disconnected.
    Suppose that this is not the case.
    Then, even after removal of $v$, there is another path between every pair of vertices in $D \setminus \{v\}$.
    In particular, there is a path between $x$ and $w$ in $D \setminus \{v\}$.
    Such a path cannot pass through the vertices of $C$, as $C$ is a pendant tree.
    Observe that $D \setminus C$ is connected (because removal of a pendant tree does not disconnect the graph). 
    Additionally, $D \setminus C$ has a triangle.
    Hence, $D \setminus C$ is a proper interval graph.
    The path $P_{x, w}$ is an induced subgraph of $D \setminus C$ and contains a good $F_3$-hook $v$ that is an internal vertex.
    But $w$ and $v$ have no neighbor in $S$.
    If another path from $w$ to $x$ exists avoiding $v$, then such a path must pass through some vertex of $S$ and, it must pass through a vertex $y \in F_1$.
    %\replace{
    Since $w$ and $v$ are $(x, y)$-connecting, it implies that there is an $x-y$ induced path in $D\setminus S$ which contains $w, v$ both.
        %}
    %{Then, $w$ and $v$ are $(x, y)$-connecting and that path exists in $D \setminus C$.
    %Also, the path between $x, y \in F_1$ in $G[D]$ contains $v$ and $w$.
    %If another path exists in $D \setminus C$ between $x$ and $y$ that passes through the vertices of $S$,}
    Then, it {creates} an induced cycle of length at least 4 in $D \setminus C$.
    This leads to a contradiction that $D \setminus C$ is a proper interval graph.
    Hence, removal of $v$ makes $D \setminus C$ disconnected.
    Since $C$ is attached to only one vertex $w \in D \setminus C$, therefore, $D \setminus \{v\}$ induces a disconnected subgraph.

    Now, let us justify that $v$ separates $x \in F_1$ from $w$ in $D$.
    Observe that the graph $D \setminus C$ is a proper interval graph.
    The vertex $v$ appears in a path $P_{x, w}$ between $x$ and $w$ in $D \setminus C$.
    In particular, that path $P_{x, w}$ is a unique path between $x$ and $w$ in $G[V_2]$.
    Since $x$ and $w$ are not adjacent, therefore any path between $x$ and $w$ must pass through $v$.
    Therefore, removal of $v$ makes $D$ disconnected.
  \end{claimproof}

  We follow a modification procedure to prove that there exists a set $X'$ of at most $k$ vertices in $G$ such that $G - X'$ is a {\pitgraph}.
  In the first step, we initialize $X' := X$.
  Hence, initially $G - X$ and $G - X'$ are the same graphs.
  %Then, there exists a connected component $D$ of $G - X'$ that is neither a proper interval graph, nor a tree.
  %Then, $D$ contains the vertices of the $w$-hanger $C$.
  %
  %First, we observe that the component $D$ is unique containing the vertices of $C$.
  %The reason is that a pendant tree can be only inserted to a single component of $G - (C \cup X)$ and does not merge two connected components into a single one.
  %Note that $w$ is a bad $F_3$-hook.
  %Then, $w \in F_2$ and $w$ is $(x, y)$-connecting when $x, y \in F_1$.
  %Hence, $w$ does not participate in any triangle of $G$.
  %Since $w \in G - (X \cup C)$ and $G - (X \cup C)$ is a {\pitgraph}, then $w$ does not participate in any hole.
  %Since $w$ does not participate in any cycle of $G - (C \cup X)$, then $w$ is a pendant vertex of $G - (X \cup C)$ or $w$ is a cut vertex of $G - (X \cup C)$.
  %\todo[inline]{Is the above sentence correct? Any reference?}
  %{Then, the vertices of $w$-hanger $C$ cannot participate in a triangle of $G$ either.
  %Therefore, the vertices of $C$ do not participate in any tent or net.
  %In addition, the vertices of $C$ cannot participate in any cycle.
  %Hence, the vertices of $C$ cannot participate in any hole of $G$.
  %Then, some of vertices of $C$ must induce a claw $J$ and there is a triangle in $D$ itself.
  %
  %\todo[inline]{Diptapriyo: the above blue paragraph can be ignored.}

  Then, any path from any triangle of $D$ to the claw $J$ passes through some vertices of $F_1$.
  In addition, $J$ contains $w$ itself, or any path from $J$ to a triangle in $D$ must also pass through $w$ and some vertex of $F_1$.
  Consider any vertex $x \in F_1$ such that there exists a path from a triangle of $D$ to claw $J$ via $x$.
  Let $P^x_J$ be the path from $x$ to the claw $J$ that is a subpath from a triangle of $D$ to the $J$.
  Observe that $P^x_J$ contains $w$.
  As $w$ is a bad $F_3$-hook, there exists a good $F_3$-hook $v \in F_2$ that appears in $P^x_J$ such that the distance of $v$ from $x$ in $G[V_2]$ is smaller than the distance of $w$ from $x$ in $G[V_2]$.
  But $v$ is a good $F_3$-hook. 
  Hence, $v$ has at least three neighbors in $G$.
  Two of those 3 neighbors of $v$ appear in $P^x_J$ and at least one other neighbor $v_1$ of $v$ appears in a pendant tree attached to $v$.
  Let $v_2$ and $v_3$ the other two neighbors of $v$ that appear in the path $P^x_J$.
  Since $v_2, v_3, v \notin X$ and there is a path from $(v, v_1, v_2, v_3)$ to a triangle of $D$, it must be that $v_1 \in X$.
  We set $X' := (X' \setminus \{v_1\}) \cup \{v\}$.
  Observe that $|X'| \leq |X|$.
  Additionally, the vertices of $C \cup \{w\}$ are not in $X'$.{}

  Next, we claim that there is no path from $w$ to $x$ in $G - X'$.
  In case there is a path $P_1$ from $w$ to $x$ in $G - X'$, then $P_1$ passes through some other vertices of $(F_1 \cup S) \setminus \{x\}$.
  The other path from $x$ to $w$ goes via $v$ in the graph $G' - X$.
  This creates a hole in leading to a contradiction that $G' - X$ is a {\pitgraph}.
  Hence, our modification process ensures that in $G - X'$, there is no path from $x$ to $w$. 

  Next, we consider the graph $G - X'$.
  If $G - X'$ is not a {\pitgraph}, then we focus on the component $D'$ that contains a claw formed by the vertices of $C$ and contains a triangle.
  We repeat this process for every $y \in F_1$ such that there is a path from the claw $J$ intersecting $C$ to a triangle in $D'$.
  Observe that this step of replacement does not decrease the size of the modified set.
  In addition, this modification process terminates by iterating over every possible vertices of $F_1$.

  It remains to prove that $G - X'$ is a {\pitgraph}.
  Observe that every step of replacement involves choosing to {include} a good $F_3$-hook.
  Due to Claim \ref{claim:v-separates-w-from-x}, if a good $F_3$-hook $v$ is chosen, then $v$ separates $w$ from $x \in F_1$ in $D$.
  Since $w$ is separated from every $x \in F_1 \cap D$, therefore, the components of $D$ that do not contain $w$ are proper interval graphs.
  In addition, the component containing $w$ forms a tree since $C$ is a $w$-hanger attached to a bad $F_3$-hook $w$.
  Other components of $G - X'$ are the same as the components of $G - X$.
  Therefore, $G - X'$ is a {\pitgraph}.
\end{proof}
\fi

Our above reduction rule ensures us that a helps us to bound the number of leaves in $G[V_2]$ that are non-neighbors of $S$.
%\input{non-nbr-leaves.tex}
% \begin{observation}
%     \label{obs:claw-P-dist}
%     Let $C$ be a pendant tree of $G[V_2]$ and \Cref{rule:degree-2-overbridge} is not applicable on $(G, k)$.
%     Then, for any $v \in C$, if there exists claw $J \subseteq C$, and $P_{v}^{J}$ is a path from $v$ to $J$ in $C$, then $|P_{v}^{J}|\le 3$.
% \end{observation}

% \todo[inline]{Diptapriyo: it seems that the statement is hard to prove.}

% \begin{proof}
% Let $C$ be a pendant tree and $(G, k)$ be an instance at which Reduction Rule \ref{rule:degree-2-overbridge} is not applicable.
% \end{proof}

%Observe that the \Cref{rul:pendant-tree} is applicable when $C$ is a pendant tree containing a vertex with degree at least 3 in $G$.
%Otherwise, every vertex of $C$ has degree at most 2 in $G$ and $C \cup \{x\}$ induces a tree.
%Additionally, if the above reduction rule is applicable, then we also analyze the number of vertices in $D$ that we decide to retain from $C$.
%Our next lemma analyzes both these things properly.

\subparagraph*{Bounding the Neighbors of $S$ in $V_2$.}
Now, we focus on designing the reduction rules that bound $N_G(S) \cap V_2$.
Our first reduction rule is based on the idea that if there is a $v$-flower of order $4k + 3$, then any optimal solution of size at most $k$ must contain $v$.

\begin{reductionrule}
  \label{rul:flower}
  If there exists a vertex $ v \in S $ such that the induced subgraph $G[\{v\} \cup V_2]$ contains a \emph{$v$-flower} of order $4k+3$, then every solution of size at most $k$ must include $ v $, and the instance reduces to $ (G - v, k - 1) $.
\end{reductionrule}

%\ifshort
%\begin{lemma}[$\star$]
%\else
%\begin{lemma}
%\fi
%\label{lem:flower}
%  \Cref{rul:flower} is safe.
%\end{lemma}

\ifshort
The backward direction ($\Leftarrow$) of the safeness of the above rule is obvious.
The forward direction ($\Rightarrow$) follows from the following intuitive reason.
Suppose that $X$ is a solution of $(G, k)$ of size at most $k$ and $v \notin X$.
Since there are $4k+3$ pairwise vertex-disjoint cycles in the $v$-flower, which share only $v$, and $X$ contains at most $k$ vertices, at most $k$ cycles can be induced cycles of length at least four or double-edges.
Hence, at least $3k+3$ cycles are triangles.
It implies that there $k+1$ triples of triangles are available.
Each triple induces a \ctpair\.
Therefore, there must be at least one vertex from each triple, which violates the budget requirement. 
Every solution of size at most $k$ must contain $v$, and we can safely delete $v$ and decrease the parameter by one.

\else
\begin{lemma}
  \Cref{rul:flower} is safe.
\end{lemma}
\begin{proof}
  Clearly, if $X$ is a solution to $(G - v, k-1)$ instance then $X \cup \{v\}$ is a solution to $(G, k)$ instance.
  Therefore, the backward direction ($\Leftarrow$) follows.

  Now, to prove the forward direction ($\Rightarrow$), we prove that if $X$ is a solution to $(G, k)$ instance then we claim that $X$ must include $v$.
  Suppose for the sake of contradiction that $v \notin X$. Consider the set of cycles $C_1, C_2, \dots, C_{4k+3}$ in the $v$-flower.
  Suppose at least $k+1$ of these cycles are double-edges or induced cycles of length at least 4.
  Then, $X$ must include at least one vertex from each of these cycles since $v \in X$, which implies that $|X| \geq k+1$.
  This contradicts the assumption that $X$ is a solution of size at most $k$.
  
  Hence, at most $k$ cycles can be induced cycles of length at least 4 or double-edges.
  Therefore, at least $3k+3$ cycles must be of length 3, i.e., they are triangles.
  We arrange them in $(k+1)$ pairwise-disjoint triplets arbitrarily and focus on one such triplet, say $ (C_i, C_j, C_{\ell}) $.
  Suppose that for all $p \in \{i, j , \ell\}$, $C_p = \{v, a_{p}, b_{p}\}$.
  %Therefore, the only obstructions can be possible in this setting are tents and nets for PITVD.
  Since $G\left[V_{2}\right]$ is a forest and hence, $G[\{a_i, b_i, a_j, b_j, a_{\ell}, b_{\ell}\}]$ is a bipartite graph.
  In particular, the edges $\{a_i b_i, a_j b_j, a_{\ell} b_{\ell}\}$ form a maximum matching of $G[\{a_i, b_i, a_j, b_j, a_{\ell}, b_{\ell}\}]$.
  
  Since $G[\{a_i, b_i, a_j, b_j, a_{\ell}, b_{\ell}\}]$ has a maximum matching of size 3, due to K{\H{o}}nig's Theorem \cite{konig1931graphs} the minimum vertex cover size of $G[\{a_i, b_i, a_j, b_j, a_{\ell}, b_{\ell}\}]$ is exactly 3.
  Therefore, maximum independent set size of $G[\{a_i, b_i, a_j, b_j, a_{\ell}, b_{\ell}\}]$ is 3.
  Without loss of generality, if $a_i, a_j, a_{\ell}$ is an independent set of size of 3 in $G[\{a_i, b_i, a_j, b_j, a_{\ell}, b_{\ell}\}]$, then $(v, a_i, a_j, a_{\ell})$ forms a claw.
%  Due to Konig's Theorem ({\color{red} cite}), $G[\{a_i, b_i, a_j, b_j, a_{\ell}, b_{\ell}\}]$ has  
	Therefore, from each triplet of these triangles $(C_i, C_j, C_{\ell})$, there is a claw $(v, a_i, a_j, a_{\ell})$ with center $v \in S$ and three other vertices from $C_i \cup C_j \cup C_{\ell} \subseteq V_2$.
	Also, there is a triangle $C_i$ that contains $v$.
	Observe that $((v, a_i, a_j, a_{\ell}), C_i)$ is a closest {\ctpair} in $G[V_2 \cup \{v\}]$.
	{As $G[V_2 \cup \{v\}]$ is a subgraph of $G$, $((v, a_i, a_j, a_{\ell}), C_i)$ is a closest {\ctpair} in $G$ itself.
  %If there are edges between the vertices of these $3k+3$ triangles other than $v$, then there is a net.
  %%%
  Therefore, there are $k+1$ {{\ctpair}s} $(J_1, T_1), (J_2, T_2),\ldots,(J_{k+1}, T_{k+1})$ in $G$ each of whom are closest {\ctpair} in $G$.}
  Furthermore, for every $i \neq j$, $(J_i \cup T_i) \cap (J_j \cup T_j) = \{v\}$.
  Therefore, if $v\notin X$ then $X$ must contain at least one vertex from each of these $k+1$ claw triangle pairs that are disjoint apart from $v$, which implies that $|X| \geq k+1$.
  This contradicts the assumption that $X$ is a solution of size at most $k$.
  Hence, $v$ must be part of any optimal solution of size at most $k$.
\end{proof}
\fi

Now, we analyze the structure of the graph when $G[\{v\} \cup V_2]$ has no $v$-flower of order $4k+3$.
%We prove the following lemma.
Our next lemma provides this characterization.
% \todo[inline]{Diptapriyo: we should have an observation/lemma. If $G[(V \setminus S) \cup \{v\}]$ has a tent, then it should be fully in $G[\{v\} \cup V_1]$.}
% \todo[inline, color=green]{Arpit: Done in the \Cref{lem:tentin}}

\ifshort
\begin{lemma}[$\star$]
\else
\begin{lemma}
\fi
  \label{lem:flower2}
  If there is no $v$-flower of order $4k+3$ in $G[\{v\} \cup V_2]$ for every vertex $ v \in S $, then there exists $Z_v\subseteq V_2$ with $|Z_v|\le 8k+4$ that can be obtained in polynomial-time such that $Z_v$ intersects every net and every cycle in $G[\{v\}\cup V_2]$ passing through $v$.
  {Furthermore, $Z_v$ intersects every pairwise-intersecting {\ctpair} of $G[V_2 \cup \{v\}]$.}
\end{lemma}

\ifshort
\else
\begin{proof}
  Since {\Cref{rul:flower}} is not applicable, graph $G[\{v\} \cup V_2]$ has no $v$-flower of order $4k + 3$.
  It follows from the \Cref{prop:flower} that we obtain a vertex subset $Z_v \subseteq V_2$ of at most $2(4k + 2)$ vertices in polynomial-time that intersects all cycles of $G[\{v\} \cup V_2]$ that pass through $v$. 
  Therefore, $Z_v$ trivially intersects {all double-edges} and all induced cycles of $G[\{v\} \cup V_2]$ containing at least $4$ vertices.

%\todo[inline]{Diptapriyo: my edit has to resume from here.}
  Consider a net $N$ in $G[\{v\} \cup V_2]$ such that $v\in N$.
  {As $N$ has a triangle and $G\left[V_{2}\right]$ has no triangle, it must be that $v$ is part of the triangle of $N$}, as otherwise it violates the fact that $G\left[V_{2}\right]$ {is a forest}.
  %that passes through $v$ is intersected by $Z_v$.
  Hence, $N$ must have a triangle in the form of $T = \{v, a, b\}$ where $a, b \in V_2$.
  %{\color{blue} As $Z_v$ passes through all cycles passing through $v$, $Z_v$ intersects $T$, and hence $Z_v$ intersects the tent $N$ passing through $v$.}
  %To form a net, there must be three pendant edges $e_1, e_2, e_3$ incident to $v, a, b$ respectively and the second endpoint of each edge in $V_2$.
  Since $Z_v$ intersects triangle $T$, it must satisfy that $Z_v \cap N  \neq \emptyset$.
  
  Now, we consider {\ctpair} $(J, T)$  where $J$ is a claw, $T$ is a triangle in $G[\{v\} \cup V_2]$, and $J \cap T \neq \emptyset$.
  By similar reasoning that each pairwise intersecting {\ctpair} of $G[V_2 \cup \{v\}]$ have a triangle of the form of $T = \{v, a, b\}$ where $a, b \in V_2$ which intersects $Z_v$.
  Therefore, $Z_v$ intersects every pairwise intersecting {\ctpair} in $G[\{v\} \cup V_2]$.

  This completes the proof that $Z_v$ intersects every cycle, every closest {\ctpair} and every net in $G[\{v\} \cup V_2]$ that passes through $v$.
\end{proof}
\fi

%When the above-mentioned rules are not applicable, we have the following lemma which bounds the number of connected components of the graph $G[V_2\setminus Z_v]$ that is adjacent to $v$.

\subparagraph*{Construction of Auxiliary Bipartite Graph.} If none of the above rules are applicable, we exploit the structural properties of the graph using the above-mentioned lemmas and construct an auxiliary bipartite graph that we use in some reduction rules later.
Let $\mathcal{C}$ denotes the set of connected components of $G[V_2 \setminus Z_v]$ that are adjacent to $v$. 
In other words, if $D \in \mathcal{C}$, then $D$ is a connected component of $G[V_2 \setminus Z_v]$ such that $v$ is adjacent to $D$.

\begin{definition}
  {\label{def:auxiliary}}
  Given $v \in S$, let $Z_v$ denote the set of at most $8k + 4$ vertices obtained by \Cref{lem:flower2} to the graph $G[\{v\} \cup Z_v]$.
  {Consider the graph $G[V_2 \setminus Z_v]$ that is a forest.}
  We define an auxiliary bipartite graph $\mathcal{Z}_v = (Z_v \cup(S \setminus \{v\}), \mathcal{C})$ where $Z_v \cup (S \setminus\{v\})$ is on one side, and $\mathcal{C}$ on the other side. The set $\mathcal{C}$ contains a vertex for each connected component $C$ of $G[V_2 \setminus Z_v ]$ that has a vertex adjacent to $v$. We add an edge between $z \in Z_v \cup (S \setminus\{v\})$, and connected component $C \in \mathcal{C}$ if $z$ is adjacent to a vertex in component $C \in \mathcal{C}$.
\end{definition}

\ifshort
\begin{observation}[$\star$]
\else
\begin{observation}
\fi
  \label{obs:auxiliary}
  Let $\mathcal{Z}_v = (Z_v\cup (S\setminus \{v\}), \mathcal{C})$ be the auxiliary bipartite graph defined in \Cref{def:auxiliary}.
  If $v$ has degree more than $7(|S| + |Z_v|) +5$ in $G[\{v\}\cup V_2]$, then $\mathcal{C}$ has more than $5(|S| + |Z_v|) + 5$ components.
\end{observation}

% \begin{observation}
%   Let $\mathcal{Z}_v = (Z_v\cup (S\setminus \{v\}), \mathcal{C})$ be the auxiliary bipartite graph defined in \Cref{def:auxiliary}.
%   If $v$ has degree more than $105(k+1)$ in $G[\{v\}\cup V_2]$, then $\mathcal{C}$ has more than \replace{$5(|S| + |Z_v|) + 50$ components}{$5(|S| + |Z_v|)$ components}.
% \end{observation}
% \todo[inline, color=green]{Arpit: I beleive this $+50$ factor is not that needed because this was required just to show $|\mathcal{C}|\ge 5(|S| + |Z_v|)$ so that we can apply $5$-expansion Lemma.}

\ifshort
\else
\begin{proof}
  Note that Reduction Rule \ref{rule:multiplicity-reduction} is not applicable, there can be at most $2|Z_v|$ edges incident to $v$ such that the other endpoint is in $Z_v$.
  Hence, there are at least $7(|S| + |Z_v|) + 5 -  2|Z_v|$ edges incident to $v$ such that the other endpoint is in $V_2 \setminus Z_v$.
  Hence, there are at least $5(|S| + |Z_v|) + 5$ edges that are incident to $v$ and the other endpoint is in $Z_v \setminus V_2$. 
  Since $|Z_v| \le 8k + 4$, and Reduction Rule \ref{rule:multiplicity-reduction} is not applicable, for every $w \in Z_v \cup (S \setminus\{v\})$, there can be two edges between $v$ and $w$.
  Hence, it follows that at most {$2(|S| + |Z_v| - 1)$} edges are incident to $v$ with other endpoint being in $Z_v \cup (S \setminus \{v\})$. 
  Therefore, $v$ is adjacent to at least {$7(|S| + |Z_v|)+5 - 2(|S| + |Z_v| - 1)= 5(|S| + |Z_v|) + 5$} vertices other than $Z_v$ in $G[\{v\}\cup V_2]$.
 %%%% \todo[inline]{Diptapriyo: I believe that subtracting $2|Z_v|$ should be sufficient enough.}
  It follows from \Cref{lem:flower2} that $Z_v$ intersects all cycles  and all nets of $G[V_2 \cup \{v\}]$ that pass through $v$.
  In addition, $Z_v$ intersects all (claw, triangle)-pairs of $G[\{v\}\cup V_2]$ that contain $v$ and are pairwise intersecting.
  Hence, there cannot exist any cycle $C^*$ containing vertices only from $v$ and some connected component $C \in \mathcal{C}$. 
  Therefore, $v$ is adjacent to at most one vertex in each connected component of $\mathcal{C}$.
  Observe that $v$ is incident to at least $5(|S| + |Z_v|) + 5$ edges with other endpoint being in $V_2 \setminus Z_v$, and for each such edge, the other endpoint appears in a connected component $\mathcal{C}$ of $G[V_2 \setminus Z_v]$.
  Due to Definition \ref{def:auxiliary}, $v$ is adjacent to one vertex in every connected component of $G[V_2 \setminus Z_v]$.
  Hence, the number of connected components in $\mathcal{C}$ is at least {$5(|S| + |Z_v|) + 5$}.
\end{proof}
\fi

\subparagraph{Applying the New Expansion Lemma.} 
If there is a vertex $v \in S$ such that there are at least {$7(|S| + |Z_v|) +5$} edges incident to $v$ with the other endpoints being in $V_2$, then \Cref{obs:auxiliary} implies that $|\mathcal{C}|> 5(|S|+ |Z_v|)$.
Suppose that we apply new expansion lemma (\Cref{lem:new_q-expansion} with q = 5) on $\mathcal{Z}_v$, we obtain $A  \subseteq Z_v \cup (S \setminus \{v\})$, $B \subseteq \mathcal{C}$ with a 5-expansion $M$ of $A$ into $B$. 
Then, it satisfies that {$|\mathcal{C}\setminus B|\le 5|((S\setminus \{v\}) \cup Z_v) \setminus A|$} and $N_{\mathcal{Z}_v}(B) \subseteq A$.
As $|\mathcal{C} \setminus B| \le 5|((S\setminus \{v\}) \cup Z_v) \setminus A|$ and $|\mathcal{C}|> 5(|S|+ |Z_v|)$, it must be that $|B|> 5|A|$.
%\todo[color=green]{Arpit: I believe this blue colored, $\setminus {v}$ must be written as $v$ does not belong to $\mathcal{Z}$.}
Then, there must be a component $C^*\in B$ such that $C^*$ is not  saturated by $M$. 
Let $\widehat{B}\subseteq B$ denote the components of $B$ that are saturated by $M$. 
As some component of $B$ is not in $\widehat{B}$, it must be that $\widehat{B} \subsetneq B$.
Our next lemma not only proves a more concrete relation between $\widehat{B}$ and $B$, but also proves that $A \neq \emptyset$.

\ifshort
\begin{lemma}[$\star$]
\else
\begin{lemma}
\fi
  \label{lemma:expansion-apply}
  Let $v \in S$ be a vertex with degree at least $7(|S| + |Z_v|) +5$ in $G[\{v\}\cup V_2]$ and $\mathcal{Z}_v$ be the auxiliary bipartite graph as illustrated in \Cref{def:auxiliary}. 
  We invoke the algorithm provided by \Cref{lem:new_q-expansion} (i.e. new $q$-expansion lemma with $q = 5$) to compute sets $A \subseteq Z_v \cup (S \setminus\{v\})$ and $B\subseteq \mathcal{C}$ such that A has a $5$-expansion $\widehat{M}$ into $B$ in $\mathcal{Z}_v$ and $N_{\mathcal{Z}_v}(B) \subseteq A$. 
  Then, the following statements hold true:
  \begin{enumerate}[(i)] 
  	\item\label{Z-v-5-comp-extra} there exist at least 5 connected components in $B$ that are not saturated by $\widehat{M}$,
  	\item\label{Z-v-C-neighbor} for any component $C \in B$, $N_G(C) \subseteq A \cup \{v\}$, and
  	\item\label{Z-v-A-not-empty} $A \neq \emptyset$.
  \end{enumerate} 
\end{lemma}

\ifshort
\else
\begin{proof}
  Consider a vertex $v \in S$ be a vertex with degree at least {$7(|S| + |Z_v|) +5$} in $G[\{v\}\cup V_2]$ and let $\mathcal{Z}_v$ be the auxiliary bipartite graph as illustrated in Definition \ref{def:auxiliary}.
  It follows from Observation \ref{obs:auxiliary} that $\mathcal{C}$ has more than $5(|S| + |Z_v|) + 5$ components in $\cCCC$.
  Due to item-(\ref{new-exp-3}) of Proposition \ref{lem:new_q-expansion}, $|\mathcal{C} \setminus B| \leq 5|(S \cup Z_v) \setminus (A \cup \{v\})|$.
  Then, $|B| > 5|A|$.
  As $|\mathcal{C}| > 5(|S| + |Z_v|) + 5$ and $|\mathcal{C} \setminus B| \leq 5(|S| + |Z_v| - |A| - 1)$, it must be that $|B| > 5|A| + 5$.
  Subsequently, note that $\widehat{M}$ is a 5-expansion of $A$ onto $B$.
  Then, exactly $5|A|$ components of $B$ are saturated by $\widehat{M}$.
  Therefore, at least 5 components of $B$ are not saturated by $\widehat{M}$.
  This completes the proof of item-(\ref{Z-v-5-comp-extra}).

  Consider any component $C \in B$.
  Note that $N_{\mathcal{Z}_v}(B) \subseteq A$.
  Hence, the component $C$ can have neighbors only in the vertices of $A$ and the vertex $v \in S$.
  Therefore, $N_G(C) \subseteq A \cup \{v\}$, completing the proof of item-(\ref{Z-v-C-neighbor}).

  Now, we prove item-(\ref{Z-v-A-not-empty}).
  Suppose for the sake of contradiction that $A = \emptyset$.
  Then, due to item-(\ref{Z-v-C-neighbor}), any $C \in B$ can neighbor only in $A \cup \{v\}$.
  But $A = \emptyset$.
  Then, any $C \in B$ is not saturated by $\widehat{M}$.
  It implies that $C$ is a component of $G[V_2 \setminus Z_v]$ that is a tree and $N_G(C) = \{v\}$.
  But $Z_v$ intersects all cycles of $G[V_2 \cup \{v\}]$ passing through $v$.
  Hence, observe that $C \cup \{v\}$ cannot contain any cycle.
  It means that $v$ has exactly one neighbor $y$ in $C$ and $vy$ is not a double-edge.
  It implies that $C$ is a pendant tree in $G$.
  But due to item-(\ref{Z-v-5-comp-extra}), there are 5 such components in $B$ that are not saturated by $\widehat{M}$ and each such component is a pendant tree.
  Then, it implies that Reduction Rule \ref{rul:pendant-trees} is applicable which is a contradiction.
  Therefore, $A \neq \emptyset$, completing the proof of item-(\ref{Z-v-A-not-empty}).
\end{proof}
\fi
%%%%%%%%%%%%%%%%%%%%%%%%%%%%%%%%%%%%%%%%%%%%%%%%%%%%%%%%%%%%%%%%%%%%%%%%%%%%%%%%%%%%%%%%%%%

Now, we describe our next reduction rule that helps us to bound the degree of any $v \in S$ in $G[V_2 \cup \{v\}]$.

\begin{reductionrule}
  \label{rule:expansion}
  Let $v \in S$ be a vertex with degree at least $7(|S| + |Z_v|) +5$ in $G[\{v\}\cup V_2]$ and let $\mathcal{Z}$ be the auxiliary bipartite graph as illustrated in \Cref{def:auxiliary}. 
  We invoke the algorithm provided by \Cref{lem:new_q-expansion} (new $q$-expansion lemma with $q = 5$) to compute sets $A \subseteq Z_v \cup (S \setminus\{v\})$ and $B\subseteq \mathcal{C}$ such that A has a $5$-expansion $\widehat{M}$ into $B$ in $\mathcal{Z}$ and $N_{\mathcal{Z}}(B) \subseteq A$. 
  Let $\widehat{B}\subseteq B$ denote the vertices of $B$ that are saturated by $\widehat{M}$.
  % (endpoints of $\widehat{M}$ in $B$).
  Remove the edges between $v$ and every connected component of $\widehat{B}$ in $G$ and create a double edge between $v$ and every vertex in $A$ to obtain the graph $G'$. 
  The new instance is $(G', k)$. 
\end{reductionrule}

The above reduction rule's safeness crucially depends on the next lemma.

\ifshort
\begin{lemma}[$\star$]
\else
\begin{lemma}
\fi
\label{lemma:exp-lemma-rule-property}
  Let $X$ be an optimal solution of size at most $k$ to $(G, k)$ and $A, B, \widehat{B}$ denote the vertex subsets obtained from \Cref{rule:expansion}.
  Then, $v \in X$ or $A \subseteq X$.
\end{lemma}

\ifshort
\else
\begin{proof}
  Consider a set $X$ of size at most $k$ such that $G - X$ is a {\pitgraph}.
  As Proposition \ref{lem:new_q-expansion} (New $q$-Expansion Lemma) has been applied with $q = 5$ in Reduction Rule \ref{rule:expansion}, we have the obtained sets that are $A \subseteq Z_v \cup (S \setminus \{v\})$ and $B \subseteq \cCCC$ such that $N_{\mathcal{Z}}(B) \subseteq A$.
  Due to item-(\ref{Z-v-A-not-empty}) of Lemma \ref{lemma:expansion-apply}, $A \neq \emptyset$ and due to item-(\ref{Z-v-C-neighbor}) of Lemma \ref{lemma:expansion-apply}, for every connected component $C \in B$, it holds that $N_G(C) \subseteq A \cup \{v\}$.
  By definition of $\cCCC$, every connected component $C \in \cCCC$ has a vertex that is adjacent to $v$.
  Since $Z_v$ intersects every cycle of $G[V_2 \cup \{v\}]$ passing through $v$, the connected components of $\cCCC$ are disjoint from $Z_v$.
  In addition, there is exactly one edge with one endpoint $v$ and other endpoint in a component of $\cCCC$.

  Suppose for the sake of contradiction that $v \notin X$ and $A \not\subseteq X$.
  We initialize $\widehat{X} := X$ and modify $\widehat{X}$ as follows.
  Then, there exists $u \in A$ such that $u \notin X$.
  For every $u \in A \setminus X$, we proceed with a modification step as follows.

  \noindent Since there is a 5-expansion of $A$ into $B$, there are 5 connected components $C^u_1,C^u_2,\dots,C^u_5$ $\in \widehat{B}$ such that for every $i \in [5]$, $(u, C^u_i) \in \widehat{M}$.
  Observe that for every $i \in [5]$, $v$ is adjacent to exactly one vertex of $C_i$ and $v$ is adjacent to one vertex in $C_i$.
  As we assume that $u, v \notin X$, there is no double-edge between $u$ and $v$.
  %Then, there is a claw $J$ whose center is $v$ and the other 3 vertices appear in $C_1 \cup C_2 \cup C_3 \cup C_4 \cup C_5$.
  \begin{description}
    \item[Case-(i):] $u$ and $v$ are adjacent and there are at least 3 components from $C^u_1,C^u_2,C^u_3,C^u_4,C^u_5$ such that there is an induced cycle of length at least 4 containing only $u, v$ and the vertices of one component $C^u_j$.
    Without loss of generality, we assume that $C^u_1, C^u_2$,  and $C^u_3$ are 3 such components of $B$.
    In particular, there is an induced cycle of length at least 4 exists that contains $u, v$ and the other vertices from $C^u_1$.
    Similarly, there is an induced cycle of length at least 4 exists that contains $u, v$ and the other vertices from $C^u_2$.
    In addition, there is an induced cycle of length at least 4 exists that contains $u, v$ and the vertices from $C^u_3$.
    Since $u, v \notin X$, it must be that $X$ contains at least one vertex from each of $C^u_1, C^u_2$ and $C^u_3$.
    We set $\widehat{X} := (\widehat{X} \cup \{u\}) \setminus (C^u_1 \cup C^u_2 \cup C^u_3 \cup C^u_4 \cup C^u_5)$.
    Observe that the size of $\widehat{X}$ reduces by at least 2.
    \item[Case-(ii):] $u$ and $v$ are adjacent and there exactly 2 components from $C^u_1,C^u_2,C^u_3,C^u_4,,C^u_5$ such that there is an induced cycle of length at least 4 containing only $u, v$ and the vertices of one component $C^u_j$.
    Without loss of generality assume that $C^u_1$ and $C^u_2$ are the 2 such components.
    Then, there is one induced cycle of length at least 4 that contains $u, v$ and other vertices from $C^u_1$ and one induced cycle of length 4 containing $u, v$ and the other vertices from $C^u_2$.
    In such a case, as $u$ and $v$ are adjacent, there must be a triangle that contains $u, v$ and the other vertex from $C^u_3$.
    Similarly, there is a triangle that contains $u, v$ and the other vertex from $C^u_4$.
    Additionally, there is a triangle containing $u, v$ and the other vertex from $C^u_5$.
    As $u, v \notin X$, note that one vertex from each of $C^u_1$ and $C^u_2$ is essential in $X$.
    Additionally, if no vertex from $C^u_3 \cup C^u_4 \cup C^u_5$ is chosen, then there is a claw $J'$ with center $u$ and 3 vertices from $C^u_3 \cup C^u_4 \cup C^u_5$ and a triangle $T'$ containing $u, v$ and the other vertex from $C^u_5$ that pairwise intersect each other.
    Hence, at least one vertex is required from $C^u_3 \cup C^u_4 \cup C^u_5$ to be added to $X$.
    So, at least three vertices from $C^u_1 \cup C^u_2 \cup C^u_3 \cup C^u_4 \cup C^u_5$ are present in $X$.
    We set $\widehat{X} := (\widehat{X} \cup \{u\}) \setminus (C^u_1 \cup C^u_2 \cup C^u_3 \cup C^u_4 \cup C^u_5)$.
    Observe that this procedure reduces the size of $\widehat{X}$ by 2.
    \item[Case-(iii):] $u$ and $v$ are adjacent and there is exactly one component from $C^u_1,C^u_2,C^u_3,C^u_4,C^u_5$ such that there is an induced cycle of length at least 4 containing only $u, v$ and the vertices of one component $C^u_j$.
    Without loss of generality, assume that $C^u_1$ is the component such that there is an induced cycle of length at least 4 containing $u, v$ and the other vertices from $C^u_1$ only.
    Then, as $uv$ is an edge, it must be that there is a triangle that contains $u, v$ and the other vertex from $C^u_2$.
    Due to similar reason, there is a triangle that contains $u, v$ and the other vertex from $C^u_3$, there is a triangle that contains $u, v$ and the other vertex from $C^u_4$ and there is a triangle that contains $u, v$ and the other vertex from $C^u_5$.
    As $u, v \notin X$, one vertex from $C^u_1$ must be chosen since a hole is created by $u, v$ and the vertices of $C^u_1$ only.
    Since there are 4 triangles that pairwise intersect at $u$ and $v$, and the other vertex from these triangles are from $C^u_2, C^u_3, C^u_4, C^u_5$, there must be at least two vertices from $C^u_2 \cup C^u_3 \cup C^u_4 \cup C^u_5$ in $X$.
    It implies that at least 3 vertices from $C^u_1 \cup C^u_2 \cup C^u_3 \cup C^u_4 \cup C^u_5$ are present in $X$.
    We set $\widehat{X} := (\widehat{X} \cup \{u\}) \setminus (C^u_1 \cup C^u_2 \cup C^u_3 \cup C^u_4 \cup C^u_5)$.
    Similar to the previous case also, this procedure reduces the size of $\widehat{X}$ by 2.
    \item[Case-(iv):] $u$ and $v$ are adjacent and there does not exist any induced cycle of length at least 4 that contains $u, v$ and the vertices of only one component from $C^u_1,C^u_2,C^u_3,C^u_4,,C^u_5$.
    Then, for every component from $C^u_1,C^u_2,C^u_3,C^u_4,,C^u_5$, there is a triangle containing only $u, v$ and the other vertex from $C^u_j$.
    Observe that these 5 triangles pairwise intersect at $u$ and $v$ and the other vertices from these 5 triangles form independent set.
    If $u, v \notin X$, then at least 3 vertices from $C^u_1 \cup C^u_2 \cup C^u_3 \cup C^u_4 \cup C^u_5$ are essential in $X$.
    We set $\widehat{X} := (\widehat{X} \cup \{u\}) \setminus (C^u_1 \cup C^u_2 \cup C^u_3 \cup C^u_4 \cup C^u_5)$.
    Similar to the previous cases, we observe that this procedure reduces the size of $\widehat{X}$ by 2.
    \item[Case-(v):] $u$ and $v$ are not adjacent.
    Then, both $u$ and $v$ are adjacent to some vertex from every component of $C^u_1, C^u_2, C^u_3, C^u_4, C^u_5$.
    Then, observe that for every pair of components $C^u_i$ and $C^u_j$, the vertices $u, v$ and the vertices of $C^u_j \cup C^u_j$ form induced cycle of length at least 4.
    Hence, at least 4 vertices from $C^u_1 \cup C^u_2 \cup C^u_3 \cup C^u_4 \cup C^u_5$ are essential in $X$.
    We set $\widehat{X} := (\widehat{X} \cup \{u\}) \setminus (C^u_1 \cup C^u_2 \cup C^u_3 \cup C^u_4 \cup C^u_5)$.
    In this case, note that the size of $\widehat{X}$ reduces by 3.
  \end{description}

  Additionally, note that due to the 5-expansion property, for two vertices $u, w \in A \setminus X$, there are 5 components $C^u_1, C^u_2, C^u_3, C^u_4, C^u_5$ adjacent to $u$ and there are 5 components $C^w_1, C^w_2, C^w_3, C^w_4, C^w_5$ adjacent to $v$ that are pairwise disjoint.
  Hence, the modification procedure (described above) involves two disjoint vertex subsets for two distinct vertices $u, w \in A$.
  Therefore, the size of $\widehat{X}$ strictly reduces by at least 2.
  Finally, we set $\widehat{X} = \widehat{X} \cup \{v\}$.
  Hence, $|\widehat{X}| < |X|$ and by construction, $A \cup \{v\} \subseteq \widehat{X}$.

  Now, we argue that $G - \widehat{X}$ is a {\pitgraph}.
  Consider a double-edge of $G$ between $x$ and $y$.
  Observe that any double-edge must be incident to a vertex in $S$.
  If $x$ or $y$ is in $A \cup \{v\}$, then $\widehat{X}$ intersects the double-edge between $x$ and $y$.
  For other double-edges not incident to any vertex of $A \cup \{v\}$, such double-edge cannot be incident to any connected component of $B$ also.
  Therefore, such double-edges are incident to the vertices that are outside the components of $B$.
  Such double-edges are already intersected by $X$, and those vertices from $X$ are unaltered.
  Hence, those double-edges are intersected by $\widehat{X}$.

  Now, we consider the holes in $G$.
  If a hole is incident to the vertices appearing in the components of $B$, then such holes pass through $v$ and some vertices of $A$.
  As $A \cup \{v\} \subseteq \widehat{X}$, such holes are intersected by $\widehat{X}$.
  Consider the other holes that are not intersected by the components of $B$.
  Observe that such holes are already intersected by $X$ and those vertices have remained unaltered by this modification process.
  Hence, such vertices are present in $\widehat{X}$.
  Therefore, those holes are intersected by $\widehat{X}$.

  Consider the nets in $G$.
  If a net $N$ of $G$ does not intersect the vertices in the components of $B$, then $N$ is already intersected by $X$.
  Those vertices appearing in $N \cap X$ are also in $\widehat{X}$.
  Hence, $N$ is intersected by $\widehat{X}$.
  Consider a net $N$ that intersects the vertices appearing in the connected components of $B$.
  Then $N$ has a pendant vertex $x$ its neighbor $y$ that participates in a triangle.
  Then, $y$ must have a neighbor in $S$.
  There are two subcase, $y \in C$ or $y \notin C$.
  If $y \in C$ for some component $C \in B$, then the triangle of $N$ must intersect the vertices of $A \cup \{v\}$, since $N_{G}(C) \subseteq A \cup \{v\}$.
  Since $A \cup \{v\} \subseteq \widehat{X}$, it follows that $N$ is intersected by $\widehat{X}$.
  If $y \notin C$, then $y \in N_G(C)$.
  It implies that $y \in A \cup \{v\}$ since $N_{\mathcal{Z}}(B) \subseteq A$ and $N_G(C) \subseteq A \cup \{v\}$.
  Since $A \cup \{v\} \subseteq \widehat{X}$, it follows that $N$ is intersected by $\widehat{X}$.

  Consider a tent $N$ in $G$.
  Then, every vertex of $N$ participates in a triangle.
  If $N$ does not intersect any component $C \in B$, then $N$ is already intersected by $X$.
  In such a case, those vertices from $N \cap X$ have remains in $\widehat{X}$.
  Then, $N$ is intersected by $\widehat{X}$.
  If $N$ is intersected by the vertices of a connected component $C \in B$, then the vertex of $C \cap N$ intersects in a triangle in $N$.
  In such a case, the triangle of $N$ that is intersected by $C$ must pass through a vertex in $S$.
  Since $N_G(C) \subseteq A \cup \{v\}$,
  the triangle of $N$ intersected by $C$ is also intersected by $A \cup \{v\}$.
  Since $A \cup \{v\} \subseteq \widehat{X}$, $N$ is intersected by $\widehat{X}$.

  Now, we prove that there does not exist claw and a triangle in any connected component of $G - \widehat{X}$.
  Consider any component $\widehat{D}$ of $G - \widehat{X}$.
  If $\widehat{D}$ intersects the connected components appearing in $B$, then $\widehat{D}$ is a tree.
  Hence, $\widehat{D}$ cannot contain a triangle.
  Otherwise, $\widehat{D}$ does not intersect connected components appearing in $B$.
  If $\widehat{D}$ is a component in $G - X$ itself, then it does not intersect with components in $B$.
  Then $\widehat{D}$ is a proper interval graph or a tree.
  If $\widehat{D}$ is not a connected component of $G - X$, a few subcase can occur.
  The first subcase is that $\widehat{D}$ intersects the components in $B$.
  Then, $N_G(\widehat{D}) \subseteq A \cup \{v\}$.
  Since $A \cup \{v\} \subseteq \widehat{X}$, observe that $\widehat{D}$ is a tree.
  The second subcase is that $\widehat{D}$ does not intersect the vertices of any component appearing in $B$.
  Then, $\widehat{D}$ is constructed by deleting the vertices of $A \cup \{v\}$ from a connected component $D'$ of $G - X$.
  But note that $D'$ is a proper interval graph or a tree.
  But deleting vertices from a proper interval graph or a tree preserves the respective property for the connected components.
  Hence, $\widehat{D}$ is a proper interval graph or a tree.
  As the cases are mutually exhaustive, there cannot exist any claw and a triangle in $\widehat{D}$.

  We have completed our arguments that $|\widehat{X}| < |X|$, $A \cup \{v\} \subseteq \widehat{X}$, and $G - \widehat{X}$ is a {\pitgraph}.
  But this contradicts the optimality of $X$. 
  Hence, the statement is true.
\end{proof}
\fi

%\ifshort
%\begin{lemma}[$\star$]
%\else
%\begin{lemma}
%\fi
%\label{lem:expansion-safe}
%  \Cref{rule:expansion} is safe.
%\end{lemma}

Now, we are ready to bound the number of vertices in $V_2$.

\ifshort
\begin{lemma}[$\star$]
\else
\begin{lemma}
\fi
\label{lemma:upper-bound-V2}
If none of the above reduction rules are applicable, then the number of vertices in $G[V_2]$ is $\cOO(|S|^2)$.
\end{lemma}

\ifshort
\else
\begin{proof}
  As we have partitioned the vertices $G[V_2]$ in $F_1$ and $F_2$, we will simply focus on giving the bound on these two partition of $G[V_2]$. 
  Since \Cref{rule:expansion} is not applicable, it implies that every $v\in S$ is adjacent to at most $7(|S| + |Z_v|) +5-1$ vertices in $G[V_2]$.
  Hence, $|F_1|$ is $\cOO(|S|^2)$.
  %Note that $G[V_2]$ is a collection of trees, the number of $x$-$y$ paths in $G[F_1 \cup F_3]$, where $x, y \in F_1$, are $|F_1| - 1$.
  %%%%
  The subgraph $G[F_1 \cup F_3]$ is a forest and the set of leaves in $G[F_1 \cup F_3]$ is a subset of $F_1$.
  Hence, the number of leaves in $G[F_1 \cup F_3]$ is at most $|F_1|$.
  In addition, every $F_3$-critical vertex is a vertex of degree at least 3 in $G[F_1 \cup F_3]$.
  Hence, $|F_3^c|$ is at most the number of leaves in $G[F_1 \cup F_3]$, implying that $|F_3^c| \leq |F_1|$.
  
  Let $x, y \in F_1 \cup F_3^c$ and $P_{x, y}$ denote the unique path between $x$ and $y$ in $G[F_1 \cup F_3]$.
  Since $G[F_1 \cup F_3]$ is a forest, the number of such unique paths is at most $|F_1| + |F_3^c| - 1 \leq 2|F_1| - 1$.
  Each of these paths can have most $2$ good $F_3$-hooks.
  Hence, the number of good $F_3$-hooks is at most $4|F_1| - 2$.
  As Reduction Rule \ref{rule:remove-F3-hangers} is not applicable, there does not exist any bad $F_3$-hook.
  Hence, the number of vertices that are good $F_3$-hooks or from $F_1 \cup F_3^c$ is at most $|F_1 \cup F_3^c| + 4|F_1| - 2 \leq 6|F_1| - 3$.
  
  We consider the vertices of $F_2 \setminus F_3$.
  Every vertex of $F_2 \setminus F_3$ is part of a pendant tree attached to a vertex that is a good $F_3$-hook or a vertex from $F_1 \cup F_3^c$.
  As Reduction Rule \ref{rul:pendant-trees} is not applicable, there can be at most 3 pendant trees attached to every possible vertex that is a good $F_3$-hook or a vertex from $F_1 \cup F_3^c$.
  Therefore, the number of pendant trees attached to the good $F_3$-hooks or to the vertices of $F_1 \cup F_3^c$ is at most $18|F_1| - 9$.
  Since Reduction Rule \ref{rul:pendant-tree} is not applicable, every such pendant tree attached to the good $F_3$-hooks or the vertices of $F_1 \cup F_3^c$ has at most 5 vertices.
  Hence, the number of vertices from $F_2 \setminus F_3$ is at most $90|F_1| - 45$.
  
  Now, we consider the number of vertices in $F_3 \setminus F_3^c$ that are not $F_3$-hooks.
  As Reduction Rule \ref{rule:remove-F3-hangers} is not applicable, there is no bad $F_3$-hook.
  Hence, all the vertices other than $F_3$-hooks are the ones that are part of degree-2-paths in $G$.
  The number of such paths between two vertices $x$ and $y$ that are good $F_3$-hooks or the vertices of $F_1 \cup F_3^c$ is at most $6|F_1| - 3$.
  As Reduction Rule \ref{rule:degree-2-overbridge} is not applicable, every degree-2-path can have at most 2 internal vertices.
  Hence, the number of vertices in $F_3 \setminus F_3^c$ that are not good $F_3$-hooks is at most $12|F_1| - 6$.
  
  Hence, the total number of vertices in $G[V_2]$ is $\cOO(|F_1|)$.
  As $|F_1|$ is $\cOO(|S|^2)$, the number of vertices in $G[V_2]$ is $\cOO(|S|^2)$.
\end{proof}
\fi

\subsection{Bounding the Proper Interval Graph Components of \texorpdfstring{$G - S$}{G - S}}
\label{sec:bounding-pig-components2}
As $G[V_{1}]$ is a collection of {\pig s}, there exists a proper interval ordering $\mathcal{V}$.
Let us fix such an ordering $\mathcal{V} = \langle v_{1}, v_2, \dots, v_{|V_1|}\rangle$ of $G[V_{1}]$.
Using this ordering, we can derive a clique partition $\mathcal{K} = \{K_1, K_2, \dots, K_t\}$ of $G[V_{1}]$ as described in \Cref{sec:prelims}.
In the following, we provide a series of reduction rules that aim to provide the upper bound on $|V_1|$ in 3 crucial steps.
%\begin{itemize}
%	\item 
First, we provide one reduction rule that provides an upper bound on the number of connected components in $G[V_1]$.
	%\item 
Second, we provide the next reduction rule that provides an upper bound on the number of cliques as per $\cKK$ in every connected component of $G[V_1]$.
	%\item 
Finally, we provide the last two reduction rules that provide an upper bound on the number of vertices in every clique of $\cKK$. 
%%\end{itemize}
%%%
Our following observation ensures us that if a proper interval graph is acyclic, then it must be a path.
%%%
%\subsubsection{Exapansion Lemma Approach}
% \todo[inline, color=green]{Arpit: I think We need some lemma/observation for it -- ``If a {\pig} component does not contain any triangle, then it is a claw-free tree, and hence, it is a path.''}

\ifshort
\begin{observation}[$\star$]
\else
\begin{observation}
\fi
\label{obs:triangle-free-pig}
If a connected component $C$ of $G - S$ is acyclic and a {\pig}, then $C$ is a path.
\end{observation}

\ifshort
\else
\begin{proof}
  Suppose that the premise of the statement is true.
  Let $C$ be a connected component of $G - S$ that is a proper interval graph and has no cycle.
  Then, $C$ cannot have any triangle, implying that $C$ is triangle-free.
  First crucial observation is that $C$ has no triangle and no hole (induced cycle of length at least 4).
  Hence, $C$ is acyclic.
  Suppose for the sake of contradiction that $C$ is not a path.
  Since $C$ is connected and not a path, therefore, there exists a vertex $v$ in $C$ such that degree of $v$ is at least $3$.
  Let $u_1, u_2, u_3$ be three distinct neighbors of $v$ in $C$.
  Since $C$ is acyclic, there cannot exist any edge between any two vertices in $\{u_1, u_2, u_3\}$.
  Therefore, the induced subgraph on the vertex set $\{v, u_1, u_2, u_3\}$ is a claw which contradicts the fact that $C$ is a {\pig}.
  Hence, $C$ must be a path.  
\end{proof}
\fi
% Due to \Cref{obs:triangle-free-pig}, any connected component of $G - S$ that is acyclic is in $G[V_2]$.
% Hence, it is natural for us to assume that every connected component of $G\left[V_{1}\right]$ contains a cycle.
% In such a case, every connected component of $G\left[V_{1}\right]$ contains at least one triangle.
% Subsequently, it is natural to assume that the cliques in $\cKK$ when restricted to the component $C$ appear in consecutive order.
Due to \Cref{obs:triangle-free-pig}, any acyclic component of $G-S$ lies in $G[V_2]$.
Thus, we assume every component of $G[V_1]$ contains a cycle, hence at least one triangle.
Consequently, it is natural to assume that the cliques in $\cKK$ restricted to such a component $C$ appear consecutively.
\subparagraph*{Bounding the Number of Components in $G[V_1]$.}
We focus on the first phase in which we provide an upper bound on the number of connected components of $G[V_1]$.
In order to state the corresponding reduction rule, we first define an auxiliary bipartite graph $\cBB$ with bipartition $(A, \cCCC)$ as follows.
\begin{definition}[Construction of Auxilliary Bipartite Graph]
  \label{def:auxiliary-bipartite}
  % If none of the previous reduction rules are applicable, we construct an auxiliary bipartite graph $\cBB$ with bipartition $(A, \cCCC)$ as follows.
  Let $A$ denote the vertex subset $S$ and ${\cCCC}$ denote the set of connected components in $G[V_1]$.
  For each vertex $v \in A$ and each component $D \in \cCCC$, we add an edge between $v$ and $D$ in $\cBB$ if and only if $v$ has at least one neighbor in $D$ in $G$.
  %%%In other words, we add an edge between $v$ and $D$ in $\cBB$ if $N(v) \cap V(D) \neq \emptyset$ and $D \in \cCCC$ represents the whole component as a single vertex in $\cBB$.
\end{definition}
% \todo[inline]{Diptapriyo: it maybe nicer to put it inside Definition environment so that we can refer to it in the next reduction rule.}
% \todo[inline, color=green]{Arpit: Done.}
%\textbf{Applying The Exapansion Lemma:} We apply the $3$-Exapansion Lemma \Cref{lem:q-expansion} and get non-empty sets
Our next reduction rule applies $q$-Expansion Lemma and provides an upper bound on $|\cCCC|$. 

\begin{reductionrule}
\label{rul:proper-interval-bounding}
We construct the auxiliary bipartite graph $\cBB = (A, \cCCC)$ as described in the \Cref{def:auxiliary-bipartite}.
If $|\cCCC| \geq 3|A|$, then we invoke \Cref{lem:q-expansion} ($q$-Expansion Lemma with $q = 3$) and obtain nonempty sets $\widehat{S}\subseteq A$ and $\widehat{\cCCC} \subseteq {\cCCC}$ such that
\begin{enumerate}[(i)]
  \item There is a $3$-expansion $M$ of $\widehat{S}$ into $\widehat{\cCCC}$ in $\cBB$, and
  \item $N_{\cBB}(\widehat{\mathcal{C}})\subseteq \widehat{S}$.
\end{enumerate}
%If there exists such sets $\widehat{S}$ and $\widehat{\mathcal{C}}$ as described above, then 
We remove all the vertices in $\widehat{\mathcal{S}}$ from $G$ to obtain the new instance $(G - \widehat{\mathcal{S}}, k - |\widehat{\mathcal{S}}|)$.
\end{reductionrule}

%\ifshort
%\begin{lemma}[$\star$]
%\else
%\begin{lemma}
%\fi
%  \Cref{rul:proper-interval-bounding} is safe.
%\end{lemma}

%%TO-BE-OMITTED-IF-FULL Proof to add

% as the {\pig} components that are triangle-free are paths and can be considered as tree components.

\subparagraph*{Bounding the Number of Cliques of $\cKK$ in $G[V_1]$.}
\ifshort
We now proceed to the second step, bounding the number of cliques of $\cKK$ in $G[V_1]$.
The reduction rule exploits the following observation: if $v \in S$ has a neighbor in $6k+5$ distinct cliques within a single component of $G[V_1]$, then $G[V_1 \cup \{v\}]$ contains at least $k+1$ obstructions whose pairwise intersection is exactly $\{v\}$.
Consequently, every solution of size at most $k$ must include $v$.
\else
Now, we move our focus for the second step in which we provide an upper bound on the number of cliques of $\cKK$ in $G[V_1]$.
The following reduction rule is based on the idea that if $v \in S$ is adjacent to at least one vertex of $6k + 5$ distinct cliques in a component of $G[V_1]$, then there are at least $k+1$ obstructions in $G[V_1 \cup \{v\}]$ that pairwise intersect only at $v$.
Hence, $v$ must be chosen in any solution of size at most $k$.
\fi
\begin{reductionrule}
  \label{rul:pig-high-nbd}
  Let $v \in S$ and $C$ be a connected component of $G[V_1]$.
  Additionally, let $\cKK_C$ denote the set of cliques from $\cKK$ that appear in $C$.
  If $v$ is adjacent to a vertex of at least $6k + 5$ distinct cliques of $\cKK_C$, then the new instance is $(G - v, k - 1)$.
  % \replace{Let $v \in S$ and $C$ be a connected component of $G[V_1]$.
  % Additionally, let $\cKK_C$ denote the set of cliques from $\cKK$ that appear in $C$.
  % If $v$ is adjacent to a vertex of at least $6k + 5$ distinct cliques of $\cKK_C$, then the new instance is $(G - v, k - 1)$.}{Let $v\in S$ and $C$ be a connected component of $G[V_1]$ such that $|N(v)\cap V(C)|\ge 5k+5$, then the new instance is $(G - v, k - 1)$.}
\end{reductionrule}

%\ifshort
%\begin{lemma}[$\star$]
%\else
%\begin{lemma}
%\fi
%  \Cref{rul:pig-high-nbd} is safe.
%\end{lemma}

\ifshort
{The backward direction ($\Leftarrow$) of the safeness of the \Cref{rul:pig-high-nbd} is trivial.
The forward direction ($\Rightarrow$) follows from the following intuitive reason that if $X$ is a solution of size at most $k$ to $(G,k)$, then $X$ must include $v$.
Assume that $v\notin X$.
As Reduction Rule \ref{rule:buss-rule-double-edge} is not applicable, there are at most $k$ of these $6k + 5$ cliques from $\cKK_C$ have one vertex adjacent to $v$ with a double-edge.
Then, there are $5k + 5$ cliques from $\cKK_C$ such that each of these $5k + 5$ cliques has a vertex adjacent to $v$ but not incident to any double-edge.
If $v$ is incident to 5 consecutive cliques in $\cKK_C$ and $v$ is not incident to any double-edge, then $v$ must be a part of a {\ctpair}.
Hence, there are $k+1$ quintuples of these cliques which provide the necessary vertices to construct an obstruction involving $v$.
Hence, $X$ must include at least one vertex from these $k+1$ quintuples of cliques, implying that $|X| > k$.
Hence, every solution of size at most $k$ must contain $v$, and the reduction rule is safe.

% Then, an obstruction in $G-(X\cup\{v\})$ must be created by the introduction of $v$.
% Since $v$ has neighbors in at least $6k+5$ distinct cliques of \(\cKK_C\), and $X$ contains at most $k-1$ vertices, sufficiently many such cliques remain unaffected by $X$.
}
\else
\begin{lemma}
  \Cref{rul:pig-high-nbd} is safe.
\end{lemma}
\begin{proof}
  Clearly, if $X$ is a solution to $(G - v, k-1)$ instance then $X \cup \{v\}$ is a solution to $(G, k)$ instance.
  Therefore, the backward direction ($\Leftarrow$) follows.
  
  Now, to prove the forward direction ($\Rightarrow$), we prove that if $X$ is a solution of size at most $k$ to $(G, k)$ instance then we claim that $X$ must include $v$.
  Suppose for the sake of contradiction that $v \notin X$.
  %Let $C^v_1, C^v_2,\ldots,C^v_{6k+5}$ denote the $6k + 5$ cliques of $\cKK_C$ such that each of these $C^v_i$ has a vertex adjacent to $v$.
  As Reduction Rule \ref{rule:buss-rule-double-edge} is not applicable, there are at most $k$ of these $6k + 5$ cliques from $\cKK_C$ have one vertex adjacent to $v$ with a double-edge.
  Then, there are $5k + 5$ cliques from $\cKK_C$ such that each of these $5k + 5$ cliques have a vertex that is adjacent to $v$ that is not a double-edge.
  Let $C^v_1, C^v_2,\ldots,C^v_{5k+5}$ denote order the $5k + 5$ cliques of $\cKK_C$ arranged in the proper-interval ordering such that each of these $C^v_i$ has a vertex adjacent to $v$.
  Since $|X| \leq k$, out of these $5k + 5$ cliques of $\cKK_C$, there are 5 cliques $C^v_{i}, C^v_{i+1},C^v_{i+2},C^v_{i+3},C^v_{i+4}$ in $\cKK_C$ appearing from left to right such that
  \begin{itemize}
  	\item each of $C^v_{i}, C^v_{i+1},C^v_{i+2},C^v_{i+3},C^v_{i+4}$ have a vertex adjacent to $v$,
  	\item $X \cap (C^v_{i} \cup  C^v_{i+1} \cup C^v_{i+2} \cup C^v_{i+3} \cup C^v_{i+4}) = \emptyset$, and
  	\item $X$ does not contain any vertex from any clique of $\cKK_C$ that appear between $C^v_i$ and $C^v_{i+4}$.
  \end{itemize}  
  
  Observe that any edge with one endpoint $v$ and other endpoint in these 5 cliques is not a double-edge.
  Then the following cases arise.
  \begin{description}
  	\item[Case-(i):] 
	  If $v$ is adjacent to 2 vertices of $C^v_{i+j}$ for some $j \in \{0, 1, 2, 3, 4\}$, then there is a triangle containing $v$ and two vertices in $C^v_{i+j}$.
	  In addition, consider $v$ and its 3 neighbors, one appearing in $C^v_i$, one appearing in $C^v_{i+2}$ and one appearing in $C^v_{i+4}$.
	  Note that no vertex of $C^v_{i}$ is adjacent to any vertex of $C^v_{i+2} \cup C^v_{i+4}$ and no vertex of $C^v_{i+2}$ is adjacent to any vertex of $C^v_{i+4}$.
	  Hence, there is a claw with center $v$.
	  Hence, there is a claw and a triangle that pairwise intersect at $v$ appearing in a same connected component
	  This contradicts that $G - X$ is a {\pitgraph}.
	  \item[Case-(ii):] If $v$ is adjacent to only one vertex in every $C^v_{i}, C^v_{i+1},C^v_{i+2},C^v_{i+3},C^v_{i+4}$.
	  Let $x_0, x_1, x_2, x_3, x_4$ be the neighbors of $v$ in $C^v_{i}, C^v_{i+1},C^v_{i+2},C^v_{i+3},C^v_{i+4}$, respectively.
	  Since $X$ does not contain any vertex that appears in the cliques of $\cKK_C$ between $C^v_i$ and $C^v_{i+4}$, and they appear in the same connected component, therefore, there is a path from $x_0$ to $x_4$ that passes through $x_1, x_2$ and $x_3$.
	  In addition, no vertex of $C^v_{i}$ is adjacent to any vertex of $C^v_{i+j}$ for any $j \geq 2$.
	  If $x_0 x_1 \in E(G)$, then there is a triangle $(v, x_0, x_1)$ and a claw with center $v$ and the other vertices $x_0, x_2, x_4$.
	  Then, we have a claw and a triangle that are pairwise intersecting in $G - X$. 
	  This contradicts that $G - X$ is a {\pitgraph}.
	  Similarly, if $x_1 x_2 \in E(G)$ or $x_2 x_3 \in E(G)$ or $x_3 x_4 \in E(G)$, then by similar argument, we can prove that there is a claw and a triangle in $G - X$ that are pairwise intersecting with each other.
	  It also implies that $G - X$ is not a {\pitgraph}.
	 
	  Hence, we assume that $\{x_0, x_1, x_2, x_3, x_4\}$ forms an independent set.
	  But, there is a path between $x_0$ and $x_{1}$ in $G - X$.
	  The length of a shortest path is at least 2 in $G - X$.
    Then, the vertices $v, x_0, x_1$ and a shortest path between $x_0$ and $x_1$ in $G - X$ forms a hole (induced cycle with at least 4 vertices).
    This also contradicts that $G - X$ is a {\pitgraph}.
	\end{description}
	
  As both the above cases are mutually exhaustive, therefore, $v \in X$ for any optimal solution of size at most $k$.
\end{proof}
\fi

%%%%If none of the above reduction rules are applicable, we construct an auxiliary bipartite graph that we use in some of our subsequent reduction rules.
%\todo[inline]{Diptapriyo: I believe that the next proposition is not required. There is no tent and net or even $C_4 ,C_5, C_6$ that intersect $G - S$.}
%By the definition of distance between two vertices in a graph, the distance between a pair of vertices in net, tent is at most $3$.
%Therefore, the following proposition is immediate.

%{\color{red}\begin{proposition}
%  \label{prop:net-tent-distance}
%  For each $\ell$ such that $3\le \ell\le t-2$, if $N(S)$ is disjoint from $K_i$ for all $\ell -2 \le i\le \ell +2$, then the vertices in $K_{\ell}$ and the first vertex of $K_{\ell+1}$ are not contained in any induced net or tent in $G$.
%\end{proposition}}

% Now, we consider bounding the number of cliques in $\cKK$ that are disjoint from $N_G(S)$.
% In order to achieve this, we first prove the following structural properties.
% Then, we state our next reduction rule that reduces the number of cliques $\widehat{K}$ from $\cKK$ such that $N(S) \cap \widehat{K} = \emptyset$.
% The safeness of the next reduction rule will follow from the next \Cref{lem:2-disjoint-triangle}.
% Recall that a vertex set $M$ is an $x - y$ seperator if $y$ and $y$ lie in different components of $G - M$.
% For any nonadjacent $x$ and $y$, an inclusion-wise minimal $x - y$ separator of a {\pig} is a clique \cite{dirac1961rigid}; consequently, such a seperator can intersect at most two cliques of $\mathcal{K}$.
Now, we consider bounding the number of cliques in $\cKK$ disjoint from $N_G(S)$. 
In order to achieve this, we first prove the following structural property.

\ifshort
\begin{lemma}[$\star$]
\else
\begin{lemma}
\fi
  \label{lem:2-disjoint-triangle}
  Let $K_{1} \dots K_{7}$ be the seven consecutive cliques in $\mathcal{K}\subseteq G[V_1]$ such that $K_{i}\cap N(S)=\emptyset$, for each $i\in [7]$.
  If none of the above reduction rules is applicable, then there are at least $2$ vertex-disjoint triangles in $G[\mathbb{K}]$, where $\mathbb{K}=\bigcup\limits_{i\in [7]}K_i$.
\end{lemma}

\ifshort
\else
\begin{proof}
  We assume for the sake of contradiction that every pair of triangles in $G[\mathbb{K}]$ intersect in at least one vertex.
  Hence, there is at most one triangle $T$ in $G[\mathbb{K}]$.
  It implies that 
  % $d_{G[\mathbb{K}]}(v)\le 2$, $\forall~v\in V(\mathbb{K})\setminus T$, i.e., 
  $G[\mathbb{K} \setminus T]$ induces a collection of paths or single edges.
  Due to \Cref{prop:pig-clique-nbd}, $T$ can lie in at most two consecutive cliques of $\mathbb{K}$.
  
  We consider the case when $T \subseteq K_2\cup K_3$ (the argument is similar when $T \subseteq K_1 \cup K_2$).
  Then, the vertices in $K_4, K_5, K_6$ induces a path with length at least 6.
  Note that any triangle from $G[\mathbb{K}]$ other than $T$ must contain at least one vertex from $K_2$ or from $K_3$.
 	Hence, such a triangle can be present only in $K_1 \cup K_2$ or in $K_2 \cup K_3$ or in $K_3 \cup K_4$.
 	Additionally, such a triangle cannot be present in $K_4 \cup K_5$ or in $K_5 \cup K_6$.
 	Hence, $|K_6| = |K_5| = |K_4| = 2$.
 	Consider the last vertex of $K_4$ as $u_1$.
 	Let us order the vertices as $u_1 ~{\uv} ~u_2 ~\uv ~u_3 ~\uv ~u_4 ~\uv ~u_5$ such that $K_5 = \{u_2, u_3\}$, $K_6 = \{u_4, u_5\}$.
 	In addition, $u_4$ is the first vertex of $K_6$, hence cannot have any neighbor in $K_7$.
 	Since the vertices of $\mathbb{K}$ have no neighbor in $S$, therefore each of $u_2, u_3, u_4$ have degree exactly 2.
 	Hence, $u_1 - u_2 - u_3 - u_4 - u_5$ forms a degree-2-path satisfying the premise of Reduction Rule \ref{rule:degree-2-overbridge}.
	This makes Reduction Rule \ref{rule:degree-2-overbridge} applicable leading to a contradiction. 	
  Therefore, if $K_2 \cup K_3$ contains a triangle, then there must exist at least one more triangle in $K_4 \cup K_5 \cup K_6 \cup K_7$.
  	
  Now, we consider the case when $T \subseteq K_3$.
  Then, there does not exist any triangle in $K_1 \cup K_2$.
  In addition, any other triangle must be part of $K_2 \cup K_3$ or part of $K_3 \cup K_4$ but not completely contained in $K_4$.
  Hence, $|K_4| = |K_5| = |K_6| = |K_7| = 2$.
  But no triangle can be part of $K_4 \cup K_5$ or $K_5 \cup K_6$, hence starting from the last vertex of $K_4$, the vertices of $K_4 \cup K_5 \cup K_6 \cup K_7$ forms a degree-2-path with at least 5 vertices.
  This implies that Reduction Rule \ref{rule:degree-2-overbridge} is applicable leading to a contradiction.
  Hence, one triangle must be present in $K_4 \cup K_5 \cup K_6 \cup K_7$.
  
  Now, we consider the case when $T \subseteq K_3 \cup K_4$ but $|K_3| = 2$.
  We have considered the case when $K_2 \cup K_3$ has a triangle that intersects $T$.
  Hence, we consider the case when $K_4 \cup K_5$ has a triangle that intersects $T$.
  Then, $|K_2| = |K_1| = 2$.
  In addition, the first vertex of $K_3$ is not part of any triangle.
  Let the vertices of $K_3$ are $v_4$ and $v_5$ such that $v_4 ~\uv ~v_5$.
  Additionally, we denote the vertices of $K_2$ as $v_2$ and $v_3$ such that $v_2 ~\uv ~v_3$.
  Let $v_1$ denote the last vertex of $K_1$.
  In such a case, when $\mathbb{K}$ has no neighbor in $S$, there is a degree-2-path $v_1 - v_2 - v_3 - v_4 - v_5$ such that only $v_1$ and $v_5$ are the only vertices with degree possibly more than 2.
  Then, Reduction Rule \ref{rule:degree-2-overbridge} is applicable leading to a contradiction.
  Hence, there should be a triangle in $K_1 \cup K_2 \cup K_3$.
  
  As the above-mentioned cases are mutually exhaustive, therefore there are 2 vertex-disjoint triangles in $\mathbb{K}$.
\end{proof}
\fi

Now, we state the next reduction rule that bypasses a clique which ensures us that there can be at most $14k + 5$ consecutive cliques in a component that are disjoint from $N_G(S)$.

\begin{reductionrule}
  \label{rul:bypass-rule}
  Let $K_{i-7k},\ldots,K_i,\ldots,K_{i+7k+4}$ is a set of $14k + 5$ consecutive cliques in a connected component of $G - S$ that are disjoint from $N_{G}(S)$.
  Consider the first vertex $x$ of $K_{i}$ and the first vertex $y$ of $K_{i+5}$.
  Find a minimum $(x, y)$-separator $L$ in $G - S$ and a clique $K_{\ell}$ such that $L$ is disjoint from $K_{\ell}$, where $\ell\in \{i+1, i+2, i+3\}$.
  %{a clique $K_j$ such that $j \in \{i+1,i+2,i+3\}$ and $L$ is disjoint from $K_j$}.
  Let $G'$ be the graph obtained from $G$ by bypassing the clique $K_{\ell}$.
  Then, the new instance is $(G', k)$.
\end{reductionrule}

Let $\widehat{Y} = \cup_{j = i-7k}^{i + 7k + 4} K_j$ as per \Cref{rul:bypass-rule}.
Note that there are at least $7k$ cliques appearing before $K_{\ell}$ and at least $7k$ cliques appearing after $K_{\ell}$.
Hence, there are $2k$ vertex-disjoint blocks of 7 cliques each.
Due to \Cref{lem:2-disjoint-triangle}, each of those $2k$ blocks contain at least 2 vertex-disjoint triangles.
Hence, for any solution $X$ of size at most $k$, $G[\widehat{Y} \setminus X]$ contains at least one triangle.
The following structural lemma can be proved using the above characterizations.
% us that \Cref{rul:bypass-rule} is safe.

\ifshort
\begin{lemma}[$\star$]
\else
\begin{lemma}
\fi
  \label{lemma:Ki-opt-soln-structure}
  Let $X \subseteq V(G)$ be a set of minimum cardinality and $|X| \leq k$ such that $G - X$ is a {\pitgraph}.
  Additionally, assume that $Y = X \cap (\cup_{j=i}^{i+4} K_j) \neq \emptyset$ where $K_j$s are as per the Reduction Rule \ref{rul:bypass-rule}.
  If $x, y \notin Y$, then $Y$ is an $(x, y)$-separator in $G - S$.
  Subsequently, for a minimum $(x, y)$-separator $\widehat{Y}$, $G - (X \setminus Y) \cup \widehat{Y}$ is a {\pitgraph}.
\end{lemma}

\ifshort
\else
\begin{proof}
  Consider a set $X \subseteq V(G)$ of the smallest cardinality such that $|X| \leq k$ and $G - X$ is a {\pitgraph} and $Y$ is defined as per the premise of the lemma statement.
  For the sake of simplicity, we assume that $B = \cup_{j=i}^{i+4} K_j$.
  As $X$ is a minimal solution to $(G, k)$, $X \setminus Y$ is not a solution to $(G, k)$.
  It implies that $G - (X \setminus Y)$ contains a component $D$ which is neither a proper interval graph nor a tree.
  However, this component $D$ contains the vertices of $Y$.
  Note that $Y$ cannot intersect any tent or net.
  Moreover, $Y$ cannot intersect any claw as $N(Y) \cap S = \emptyset$.
  %Hence, $D$ must contain a hole.
  Now, we consider the components that intersect $(\cup_{j=i-7k}^{i-1} K_j)$ and $(\cup_{j=i+5}^{i+7k+4} K_j)$.
  Due to Lemma \ref{lem:2-disjoint-triangle}, there are $2k$ vertex-disjoint triangles in $(\cup_{j=i-7}^{i-1} K_j)$ and $2k$ vertex-disjoint triangles in $(\cup_{j=i+5}^{i+7k+4} K_j)$.
  Since $|X| \leq k$, there is at least one connected component of $G - X$ intersecting $(\cup_{j=i-7k}^{i-1} K_j)$ has a triangle.
  Additionally, there is at least one connected component of $G - X$ that intersects $(\cup_{j=i+5}^{i+7k+4} K_j)$ and has a triangle.
  Hence, any component of $G - X$ that intersects $(\cup_{j=i-7k}^{i-1} K_j)$ or $(\cup_{j=i+5}^{i+7k+4} K_j)$ is a path or a triangle.

  Hence, if the vertices of $Y$ are adjacent to any vertex of a component $D_1$ of $G - X$ intersecting $(\cup_{j=i-7k}^{i-1} K_j)$, then $D_1$ is a {\pig}.
  Similarly, if the vertices of $Y$ are adjacent to any vertex of a component $D_2$ of $G - X$ intersecting $(\cup_{j=i+5}^{i+7k+4} K_j)$, then $D_2$ is a {\pig}.
  Additionally, $v$ is not incident to any claw of $G$.
  Therefore, the component $D$ of $G - (X \setminus Y)$ that contains the vertices of $Y$ cannot have both claw and a triangle.
  Hence, it must contain a hole $H$. 
  As $H$ intersects $Y$, $H \cap (B \setminus \{x\}) \neq \emptyset$.
  Consider the induced path $P$ such that $P = H \cap B$.
  Let $x'$ and $y'$ be two neighbors of $P$ in $K_{i-1}$ and $K_{i+5}$.

  We first argue that $Y$ separates $x'$ from $y'$ in $G[B \cup \{x', y'\}]$.
  Suppose not.
  Then, there is an induced path $P'$ from $x'$ to $y'$ in $G[(B \setminus Y) \cup \{x', y'\}]$. 
  Let $x''$ be the neighbor of $x'$ in $P'$ and $y''$ be the neighbor of $y'$ in $P'$.
  Note that $x' \uv x'' \uv y'' \uv y'$.
  In addition, suppose that $x^*$ be the neighbor of $x'$ in $H$ such that $x^* \neq x$ (not in $P'$) and $y^*$ be the other neighbor of $y'$ in $H$ (not in $P'$) such that $y \neq y^*$.
  Hence, by choice, $x^* \uv x'$ and $y' \uv y^*$.
  Then, $x^* \uv x' \uv x''$ and $y'' \uv y' \uv y^*$.

  We consider the situations depending on $x^* x'', y'' y^*$ are edges or not.
  In either of these situations, we consider the path $P^*$ which is $x''$-$P'$-$y''$.
  If $x^* x'' \in E(G)$, then we add the edge $x^* x''$ into $P^*$. 
  Otherwise, we add two edges $x^* x', x' x''$ into $P^*$.
  If $y^* y'' \in E(G)$, then we add the edge $y^* y''$ into $P^*$, otherwise we add two edges $y'' y', y' y^*$ into $P^*$.
  In every possible combination of situations, we get an induced path $\widehat{P}$ that starts from $x^*$, uses $x''$-$P$-$y''$ and subsequently moves to $y^*$.
  After that, it uses the edges of $H$ that is disjoint from the edges of $P$ and the edges of $H \setminus (P \cup \{x, y\})$.
  This creates a hole in $G - X$, contradicting that $G - X$ is a {\pitgraph}.
  Hence, $Y$ must separate $x'$ from $y'$ in $G[B \cup \{x', y'\}]$.

  Finally, after deleting an $(x', y')$-separator of $G[B \cup \{x', y'\}]$ ensures that no vertex of $B$ can participate in any hole of $G$.
  Hence, $Y$ must be a minimal $(x', y')$-separator of in $G[B \cup \{x, y\}]$.
  %%%As $G[B \cup \{x, y\}]$ is a chordal graph and $Y$ is a minimal separator, $Y$ is a clique. 
  As $x' \in K_{i-1}$ is adjacent to a vertex in $K_i$, $x'$ is adjacent to the first vertex of $K_i$ which is $x$.
  Since $y' \in K_{i+5}$, $yy'$ is an edge.
  As $Y$ separates $x'$ from $y'$, $Y$ separates $x$ from $y$.
  This completes the proof that $Y$ is an $(x, y)$-separator.

  Consider an arbitrary minimum $(x, y)$-separator $\widehat{Y}$ in $G - S$.
  Note that removal of $\widehat{Y}$ from $B \setminus \{x\}$ also separates $K_{i-1}$ from $K_{i+5}$.
  This will eventually ensure that there does not exist any hole of $G - X$ passing through the vertices of $B$.
  Hence, $(X \setminus Y) \cup \widehat{Y}$ is also an optimal solution to $(G, k)$.
\end{proof}
\fi

\ifshort
%{\color{blue}
%A crucial fact due to \cite{KE2018109} that if a clique from the clique partition of a proper interval graph is bypassed, then the resulting graph is also a proper interval graph.
%Since there are $14k + 5$ (respectively, $14k + 4$) consecutive cliques of $\cKK$ in a component of $G - S$ (respectively, of $G' - S$), there are $2k$ blocks of 7-consecutive cliques each and due to \Cref{lee:2-disjoint-triangle}, each such block has at least 2 vertex-disjoint triangles.
%Hence, after deleting any solution $X$ of size at most $k$, there is connected component of $G - X$
The forward direction ($\Rightarrow$) of the safeness of the \Cref{rul:bypass-rule} can be proved by \Cref{lem:2-disjoint-triangle} and \Cref{lemma:Ki-opt-soln-structure}.
For the backward direction ($\Leftarrow$), if $X$ is a solution of size at most $k$ to $(G',k)$ but not to $(G,k)$, then an obstruction in $G-X$ must intersect the bypassed clique $K_{\ell}$.
In addition, due to the structure of the graph, the obstruction cannot be net or tent.
%Since there are at least $7k$ cliques on each side of $K_{\ell}$, there are $2k$ vertex-disjoint blocks of seven consecutive cliques, each containing two vertex-disjoint triangles.
%Hence, at least one triangle survives after deleting $X$.
Hence, the obstruction intersecting $K_{\ell}$ is either a hole or a component containing \ctpair.
In both cases, it can be explained that $G' - X$ also contains a similar obstruction, justifying the safeness of the reduction rule. 
For additional details of the arguments, we refer to the appendix.
%}

\else
\begin{lemma}
  \Cref{rul:bypass-rule} is safe.
\end{lemma}
\begin{proof}
  Let $K_{\ell}$ be the clique such that $\ell \in \{i+1, i+2, i+3\}$ and $K_{\ell}$ is bypassed by Reduction Rule \ref{rul:bypass-rule}.
  Additionally, assume that
  \begin{itemize}
    \item $B_{pv} = \cup_{j=i-7k}^{i-1} K_j$,
    \item $B = \cup_{j=i}^{i+4} K_j$, and
    \item $B_{nt} = \cup_{i+5}^{i+7k+4} K_j$.
  \end{itemize} 
  We use a few very crucial observations here.
  As $B_{pv}$ contains $k$ blocks of 7 consecutive cliques in $\cKK$ in the same connected component, therefore due to Lemma \ref{lem:2-disjoint-triangle}, in a block of every consecutive 7 cliques, there are 2 vertex-disjoint triangles.
  Since there are $k$ such blocks, there are at least $2k$ such vertex-disjoint triangles in $B_{pv}$.
  Similarly, due to Lemma \ref{lem:2-disjoint-triangle}, in $B_{nt}$, there are $2k$ vertex-disjoint triangles.
  Therefore, for solution $X$ of size at most $k$, there exists a connected component $D^*$ of $G - X$ that intersects $B_{pv} \cup B \cup B_{nt}$, and has a triangle.

  First we give the backward direction ($\Leftarrow$) of the proof.
  Let $X$ be a set of at most $k$ vertices such that $G' - X$ is a {\pitgraph}, but $G - X$ has a connected component $D$ that is neither a proper interval graph nor a tree.
  Observe that the component $D$ intersects $K_{\ell}$.
  If $D$ contains a claw and a triangle, then $D$ must contain a path $P^*$ from $K_{\ell}$ to a claw $J$ in $G - X$.
  Note that no claw of $G$ can intersect $B_{pv} \cup B \cup B_{nt}$. 
  Hence, the path $P^*$ must pass through a vertex in $S$.
  In particular, $P^*$ passes through at least one vertex from every clique $K$ such that $K \subseteq B_{pv}$ or every clique $K \subseteq B_{nt}$.
  But note that $G[B_{nt}]$ (respectively, $G[B_{pv}]$) has at least $2k$ vertex-disjoint triangles.
  As $|X| \leq k$, both the subgraphs  $G[B_{nt}] \setminus X]$ and $G[B_{pv} \setminus X]$ have at least one triangle.
  Hence, $P^*$ passes through at least one triangle.
  But $G[B_{nt}] \setminus X]$ and $G[B_{pv} \setminus X]$ are induced subgraphs of $G' - X$.
  Additionally, the deletion of $K_i$ does not disconnect the vertices of $G[B_{nt}] \setminus X]$ from $J$ (respectively, the vertices of $G[B_{pv}] \setminus X]$ from $J$) in the graph $G' - X$.
  Hence, $G' - X$ also has a connected component that contains a claw and a triangle, leading to a contradiction that $G' - X$ is a {\pitgraph}.

  Suppose that $D$ contains a hole $H$.
  Observe that $H$ must pass through the vertices of $K_{\ell}, K_{\ell - 1}$, and $K_{\ell + 1}$.
  Consider the leftmost vertex $x$ of $H$ in $N(K_{\ell}) \cap K_{\ell-1}$ and the rightmost vertex $y$ of $H$ in $N(K_{\ell}) \cap K_{\ell + 1}$.
  As $x$ and $y$ both have neighbors in $K_{\ell}$, $xy$ is an edge in $G'$.
  Hence, $xy$ is an edge in $G' - X$.
  Additionally, $H$ has at least 7 vertices and at most 2 vertices of $H$ are in $K_{\ell}$.
  Observe that removing the vertices of $K_{\ell}$ and adding the edge $xy$ into $H$ creates a subgraph $H'$ that has at least 5 vertices.
  Hence, $H'$ is also a hole in $G' - X$, leading to a contradiction that $G' - X$ is a {\pitgraph}.

  For the forward direction ($\Rightarrow$), let $X$ be a set of at most $k$ vertices such that $G - X$ is a {\pitgraph}.
  If $X \cap B (= Y) \neq \emptyset$ and $x, y \notin X \cap B$, then due to Lemma \ref{lemma:Ki-opt-soln-structure}, $Y$ is a minimum $(x, y)$-separator and $Y \cap K_{\ell} \neq \emptyset$.
  Hence, we can assume that $X$ is disjoint from $K_{\ell}$.
  In addition, due to Lemma \ref{lem:2-disjoint-triangle}, $G[B_{pv}]$ and $G[B_{nt}]$ both have at least $2k$ vertex-disjoint triangles.
  Hence, $G - X$ has at least one triangle when restricted to $G[B_{nt} \setminus X]$ as well as when restricted to $G[B_{pv} \setminus X]$.

  Suppose that $E'$ denote the edges of $G'$ added that are not in $G$ and $Z$ denote the endpoints of the edges in $E'$.
  Assume for the sake of contradiction that $G' - X$ is not a {\pitgraph}.
  Then, there is a component $D$ of $G - X$ that is not a {\pitgraph}, and it contains the edges of $E'$.
  It implies that$D$ has a hole or a claw as well as a triangle.
  In case $D$ has a claw $J$ and a triangle $T$, then there is a path $P^*$ from $J$ to $T$.
  In addition, the claw $J$ cannot intersect the vertices of $(B \cup B_{pv} \cup B_{nt}) \setminus (K_{i-7k} \cup K_{i+7k+4})$.
  Hence, $P^*$ must pass through at least one vertex from every clique in $B_{nt} \setminus K_{i+7k+4}$ or at least one vertex from $B_{pv} \setminus K_{i-7k}$.
  As both $G[B_{nt} \setminus (K_{i+7k+4} \cup X)]$ and $G[B_{pv} \setminus (K_{i-7k} \cup X)]$ contains at least one triangle, hence such a triangle $T$ is already present in $G - X$ as well.
  Then, $G - X$ itself has a component containing both a claw and a triangle, leading to a contradiction.

  Now, we consider the case when $G' - X$ has a hole $H$.
  Note that the hole must use an edge from $E'$.
  Let $xy \in E'$ that is in $H$.
  It must be that $x \in K_{\ell - 1}$ and $y \in K_{\ell + 1}$.
  Consider the vertices of $K_{\ell}$ that are adjacent to $x$ or $y$.
  If there is a common neighbor $z$ of $x$ and $y$ in $K_{\ell}$, then $x$-$z$-$y$ forms an induced path with 3 vertices.
  In such a case, subdividing the edge $xy$ by a vertex $z$ creates a hole $H'$ that is in $G - X$. 
  This leads to a contradiction.
  If there is no common neighbor of $x$ and $y$ in $K_{\ell}$, then consider a neighbor $z_1$ of $x$ and $z_2$ of $y$ both in $K_{\ell}$.
  Note that $x$-$z_1$-$z_2$-$y$ forms an induced path.
  Then, subdividing the edge $xy$ by 2 internal vertices creates a hole $H'$ in $G - X$.
  This leads to a contradiction.

  Hence, this reduction rule is safe.
\end{proof}
%\else
\fi

Our next lemma provides an upper bound on the number of cliques in $\cKK$ with the help of Reduction Rules \ref{rul:proper-interval-bounding}, \ref{rul:pig-high-nbd} and \ref{rul:bypass-rule}.

\ifshort
\begin{lemma}[$\star$]
\else
\begin{lemma}
\fi
  \label{lem:boudning-number-of-cliques}
  Let $(G, k)$ be an instance of {\PITVD} to which none of the above reduction rules are applicable.
  Then, the number of cliques in $\cKK$ is $\cOO(k^2|S|^2)$.
  %{If there are more than $(72k+60)|S| + 12$ cliques in $\mathcal{K}$, then return a no-instance.}
\end{lemma}

\ifshort
\else
\begin{proof}
  First, we note that Reduction Rule \ref{rul:proper-interval-bounding} is not applicable.
  Therefore, there are $3|S|$ connected components in $G - S$ that are in $G[V_1]$.
  Subsequently, we consider a connected component $C$ of $G[V_1]$. 
  Since \Cref{rul:pig-high-nbd} is not applicable, for every $v\in S$ there are at most $6k+5$ cliques of $C$ intersecting $N(v)$.
  Hence, the number of cliques that intersecting $N(S)$ are at most $3(6k+5)|S|$.
  Hence, the total number of cliques in $\cKK$ that intersect $N(S)$ is at most $3|S|^2 (6k + 5)$.
  %%%%
 	Consider every block of consecutive cliques that are disjoint from $N(S)$ and that appear in a single component.
 	As \Cref{rul:bypass-rule} is not applicable, in a single component, the number of cliques that are disjoint from $N(S)$ is $(14k + 5)\{(6k + 5)|S| + 1\}$.
 	Hence, the total number of cliques in $\cKK$ that are disjoint from $N(S)$ is $3|S|[(14k + 5)\{(6k + 5)|S| + 1\}]$ which is $\cOO(k^2 |S|^2)$.
%  Consider an ordering of the cliques $\widehat{\cKK} = \{L_1, L_2, \ldots, L_t\}$ of $G[V_1]$ that intersect $N(S)$.
%  Observe that any clique of $\cKK$ that is disjoint from $N(S)$ appears between two consecutive cliques $L_i$ and $L_{i+1}$.
%  Hence, number of such positions where such a clique of $\cKK$ disjoint from $N(S)$ appears is at most $|\widehat{\cKK}| \leq 3(6k + 5)|S|$.
  %%Additionally, since , between every two clique $L_i$ and $L_{i+1}$ that intersect $N(S)$, there are at most $14k + 5$ consecutive cliques that are disjoint from $N(S)$.
  Therefore, the total number of cliques in $G[V_1]$ is $\cOO(|S|^2 k^2)$.
\end{proof}
\fi

\subparagraph{Bounding the Size of Every Clique of $\cKK$.}
\ifshort
We now turn to the final crucial step: a marking scheme, partially inspired by Cao et al. \cite{KE2018109}, that yields an upper bound on the number of vertices of any clique $K \in \cKK$ intersecting $N(S)$.
\else
Now, we focus on the final crucial step that illustrates a marking scheme and provides an upper bound on the number of vertices in every clique $K \in \cKK$ that intersect $N(S)$.
This marking scheme is partially inspired by a marking scheme illustrated by Cao et al. \cite{KE2018109}.
\fi
\subparagraph*{{\sf Mark-Clique}($K$):} 
    % {\color{blue}The first point \ref{double-edge} considers the presence of a double-edge.
    % In all the subsequent points, we assume that no double-edge is involved.}
    We mark the vertices of $K$ in the following way:
\begin{enumerate}
    \item\label{mark:double-edge} For every $u\in S$, mark every vertex $v\in K$ such that $uv$ is a double-edge in $G$.
    \item\label{mark:subet-3} For every $Z \in \binom{S}{\leq 3}$ and for every function $f: Z \to \{0, 1\}$, a vertex $v \in K$ {\em matches} $(Z, f)$ when for every $z \in Z$, (i) $f(z) = 1$ then $vz \in E(G)$ and (ii) $f(z) = 0$ then  $vz \notin E(G)$.
    Let $\mathcal{V}$ be a proper interval ordering of $G[V_1]$.
    For every $Z \in \binom{S}{\leq 3}$ and for every $f: Z \to \{0, 1\}$, let $P_{Z, f}$ denote the vertices of $K$ that match $(Z, f)$.
    Let $p:=|P_{Z, f}|$, and mark the first $min(2k+3, p)$ and the last $min(2k+3, p)$ vertices of $P_{Z, f}$ with respect to $\mathcal{V}$.
    \item\label{mark:k**2} Let $K_{pv}$ and $K_{nt}$ be the previous and the next clique of $K$ in the clique partition $\mathcal{K}$. For each $x \in S$, and each $y$ of the last $2k + 1$ non-neighbors of $x$ in $K_{pv}$\footnote{In case $x \in S$ has fewer than {$2k+1$} non-neighbors in $K_{pv}$, then we consider every non-neighbor of $x$ in $K_{pv}$. Moreover, the same goes for the two subsequent points and situations.}, we mark the last $2k + 3$ common neighbors of $x$ and $y$ in $K$; for each $x \in S$, and each $z$ of the first $2k + 1$ non-neighbors of $x$ in $K_{nt}$, we mark the first $2k + 3$ common neighbors of $x$ and $z$ in $K$.
    Let them be denoted by $K^2_{pv}(x, y)$ and by $K^2_{nt}(x, z)$ respectively.
    \item\label{mark:k**3} For each $x\in S$, and each $y$ of the last $2k + 1$ neighbors of $x$ in $K_{pv}$, we mark the last $2k + 3$ vertices in $K$ that are neighbors of $y$ but not $x$; for each $x \in S$, and each $z$ of the first $2k +1$ neighbors of $x$ in $K_{nt}$, we mark the first $2k + 3$ vertices in $K$ that are neighbors of $z$ but not $x$. Let them be denoted by $K^3_{pv}(x, y)$ and by $K^3_{nt}(x, z)$ respectively.
    \item\label{mark:K**4} For each $x \in S$, and each $y$ of the first $2k + 1$ neighbors of $x$ in $K_{pv}$, we mark the first $2k + 3$ vertices in $K$ that are neighbors of $x$ but not $y$; for each $x \in S$, and each $z$ of the last $2k + 1$ neighbors of $x$ in $K_{nt}$, we mark the last $2k + 3$ vertices in $K$ that are neighbors of $x$ but not $z$. Let them be denoted by $K^4_{pv}(x, y)$ and by $K^4_{nt}(x, z)$ respectively.
\end{enumerate}

Consequently, the number of marked vertices in $K$ is at most $\eta(k) = |S|k + 2(2k+3)\left(\sum\limits_{i=1}^{3}2^{i}\binom{|S|}{i}\right) + 6|S|(2k+1)(2k+3)$, which is $\cOO(k^{19})$.
Since the obstructions in $\mathcal{F}$ of size at most $6$ are already taken care of by \Cref{lem:small-obstruction}, we need to focus on the obstructions of size at least $7$, i.e., holes of length at least $7$ and {\ctpair s}.
%%%%%
%\todo[inline]{Diptapriyo: what if $N(S) \cap K = N(S) \cap K_{pv} = N(S) \cap K_{nt} = \emptyset$?  Does our marking scheme also guarantee marking $2(k+3)\left(\sum\limits_{i=1}^{3}2^{i}\binom{|S|}{i}\right)$ vertices from $K$?}
%\todo[color=green, inline]{Arpit: Then, first point of marking scheme comes in the play, $f:Z\to \{0\}$. Moreover, it is written that this number is at most. So, I believe it is not a problem.}
%%%%%

\begin{reductionrule}
  \label{rul:mark-clique}
  Let a clique $K$ in the clique partition $\mathcal{K}$ of $G[V_{1}]$.
  If a vertex $v \in K$ remains unmarked after applying $\mathsf{Mark}$-$\mathsf{Clique}(K)$, then the new instance is $(G - v, k)$.
\end{reductionrule}

By \Cref{rule:buss-rule-double-edge}, each vertex of $G$ is incident to at most $k$ double-edges.
Moreover, ${\sf Mark\mbox{-}Clique}(K)$ marks every vertex in $V_1$ incident to a double-edge.
Hence, no unmarked vertex of $V_1$ is incident to a double-edge.
We have the following observation.
% \begin{observation}
%     \label{obs:double-edge}
%     There is no unmarked vertex $v\in V_1$ under the procedure {\sf Mark-Clique}($K$) such that $v$ is incident to a double-edge in $G$.
% \end{observation}
% {\color{blue}
\ifshort
\begin{observation}[$\star$]
\else
\begin{observation}
\fi
    \label{obs:mark-signature}
    Let $v$ be an unmarked vertex by the procedure {\sf Mark-Clique}($K$), and $X\subseteq V(G)\setminus \{v\}$ with $|X|\le k$.
    \begin{enumerate}
        \item\label{obs:double-edge} There is no unmarked vertex $v\in V_1$ under the procedure {\sf Mark-Clique}($K$) such that $v$ is incident to a double-edge in $G$.
        % , $G-X$ or $G-(X\cup \{v\})$.
        % \item There exists $Z$ and $f$ such that $v\in P_{Z, f}$ and $4k+3$ vertices are marked that matches ($Z, f$).
        % $2k+3$ out of these $4k+3$ vertices satisfies $v'\uv v$ and the rest satisfies $v\uv v'$. 
        % \item If $v\in P_{Z,f}$, then there are at least $2k+3$ vertices of $P_{Z,f}$ preceding $v$, i.e. $v'\uv v$ and at least $2k+3$ vertices of $P_{Z,f}$ succeeding $v$, i.e. $v\uv v'$ in $\mathcal{V}$.
        % At most $k$ many vertices out of these $2k+3$ can be incident to any double-edge.
        % Moreover, there are at least $3$ marked vertices $v'\in P_{Z, f}$ such that $v'\uv v$ and $v'$ is not incident to any double-edge in $G-X$ or $G-(X\cup \{v\})$ and at least 3 marked vertices $v\in P_{Z, f}$ such that $v\uv v'$ and $v'$ is not incident to any double-edge in $G-X$ or $G-(X\cup\{v\})$.
        % Note that $|P_{Z, f}|\ge 4k+3$ in $G$ and $|P_{Z, f}|\ge 3k+3$ in $G-X$ or $G-(X\cup \{v\})$. 
        \item\label{obs:subset-3} Let $v\in P_{Z,f}$. Then the following hold:
        \begin{enumerate}
            \item\label{obs:2k-3-marked} There are $2k+3$ marked vertices $v'\in P_{Z,f}$ preceding $v$, i.e. $v'\uv v$, and $2k+3$ marked vertices $v'\in P_{Z,f}$ succeeding $v$, i.e. $v\uv v'$, in $\mathcal V$.
            Among these $2k + 3$ marked vertices that appear before $v$ (respectively, appear after $v$) in $\mathcal{V}$, at most $k$ of them are incident to a double-edge.
            \item\label{obs:3-marked} Consequently, there are at least three marked vertices $v'\in P_{Z,f}$ such that $v'\uv v$, and at least three marked vertices $v'' \in P_{Z,f}$ such that $v\uv v''$, such that neither $v'$ nor $v''$ is incident to any double-edge in either of $G-X$ or $G-(X\cup\{v\})$.
            \item\label{obs:mark-num} In particular, $|P_{Z,f}|\ge 4k+6$ in $G$ and there are at least $3k + 6$ marked vertices of $P_{Z,f}$ present in $G-X$. 
            %  \text{ and } G-(X\cup\{v\}).$
        \end{enumerate}
    \end{enumerate}
\end{observation}
\ifshort 
\else
\begin{proof}
    We prove the observation in the same order as it is written.
    For Item (\ref{obs:double-edge}), since there are at most $k$ many double-edge can be incident to a vertex in $G$ as \Cref{rule:buss-rule-double-edge} is not applicable, and we mark all the vertices $u\in K$ that are incident to a double-edge, it implies that all vertices in $K$ that are incident to a double-edge are marked.
    Hence, there is no unmarked vertex in $K$ which is incident to a double-edge.

    Now, for Item (\ref{obs:subset-3}), we prove it by the definition of $P_{Z, f}$.
    Since $v \in P_{Z, f}$ is an unmarked vertex, it implies that there must exist a subset $Z \in \binom{S}{\leq 3}$ and a function $f:Z\to \{0,1\}$ such that $f(z)=1$ if $vz\in E(G)$ and $f(z)=0$ otherwise, for all $z\in Z$.
    Clearly, $v$ matches $(Z, f)$, i.e., $v\in P_{Z, f}$ by the definition of $P_{Z, f}$. 
    As $v$ is unmarked, implying that $|P_{Z, f}|\ge 4k+6$.
    Moreover, $2k+3$ vertices out of these vertices $v'\in P_{Z, f}$ satisfy that $v'\uv v$ and the rest $2k+3$ vertices satisfy that $v\uv v'$, by step-(\ref{mark:subet-3}) of {\sf Mark-Clique}($K$).
    Note that at most $k$ out of these $2k+3$ vertices satisfy that $v'\uv v$ and are incident to double-edges.
    Similarly, at most $k$ out of these $2k + 3$ vertices satisfy that $v \uv v''$ and $v''$ is incident to a double-edge.
    This completes the justification of \Cref{obs:2k-3-marked}.
    Since $|X|\le k$, at most $k$ vertices are deleted in $G - X$ as well as in $G - (X \cup \{v\})$.
    Hence, due to item-(a), there are at least $k + 3$ marked vertices in $G - X$ (respectively, in $G - (X \cup \{v\})$) that appear before $v$ and at least $k + 3$ marked vertices in $G - X$ (respectively, in $G - (X \cup \{v\})$) that appear after $v$ with respect to the ordering $\mathcal{V}$.
    Since out of these $k + 3$ marked vertices, at most $k$ of them are incident to a double-edge, it implies that there are at least three marked vertices $v'\in P_{Z,f}$ such that $v'\uv v$, and at least three marked vertices $v''\in P_{Z,f}$ such that $v\uv v''$, such that neither $v'$ nor $v''$ is incident to any double-edge in either $G-X$ or $G-(X\cup\{v\})$.
    This completes the proof of \Cref{obs:3-marked}. 
    
    For the \Cref{obs:mark-num}, note that $v$ is unmarked by the procedure {\sc Mark-Clique}($K$).
    Therefore, the step-(\ref{mark:subet-3}) of {\sf Mark-Clique}($K$) did not mark $v$, and it follows that $|P_{Z, f}| \geq 4k + 6$ in $G$.
    As $|X| \leq k$, in $G - X$, it follows that $|P_{Z, f} \setminus X|\ge 3k+6$.
    Hence, at least $3k + 6$ vertices of $P_{Z, f}$ are present in $G-X$.
\end{proof}
\fi
% }
\ifshort
The safeness of the above reduction rule crucially relies on the above observation and the next two lemmas; \Cref{obs:hole-or-claw-triangle1}, \Cref{lem:safe-hole-7} and \ref{lem:safe-ctpair}.
Intuitively, the idea is that if a subset $X$ of size at most $k$ is a solution to $(G - v, k)$, then $X$ is a solution to $(G, k)$.
\else
To prove the safeness of  \Cref{rul:mark-clique}, we will use the following observation and lemmas in which we show that if an unmarked vertex $v$ in $K$ is a part of an obstruction $O$ in $G-X$ for some solution $X$ to $(G, k)$, then there is another obstruction in $G-(X\cup \{v\})$ which does not contain $v$.
Let $R = V(G) \setminus (S \cup K)$.
Moreover, for any obstruction $O$, for the sake of simplicity, we assume $O_S=O\cap S$, $O_K=O\cap K$ and $O_R=O\cap R$.
These notations will be changed accordingly for different obstructions $O$.
%It implies that $v\in O$ and $O\cap S\neq \emptyset$ otherwise, $O$ must be completely contained in $G[V_1]$ and $G[V_1]$ is a collection of {\pig s}.
%Clearly, $O=O_S\uplus O_K \uplus O_R$.
\fi

\ifshort
\begin{observation}[$\star$]
\else
\begin{observation}
\fi
    \label{obs:hole-or-claw-triangle1}
    Let $C = v - u - u_1 - u_2 - \ldots - u_{\ell} - u' - v$ be a cycle of length at least 6 where $u - u_1 - u_2 - \ldots - u_{\ell} - u'$ is an induced path $P$.
    Then, $G$ has an induced cycle of length at least 4 or $G$ has a claw $J$ and a triangle $T$ such that $J \cap T \neq \emptyset$.
\end{observation}

\ifshort
\else
\begin{figure}[ht]
    \centering
    \includegraphics[scale=0.8]{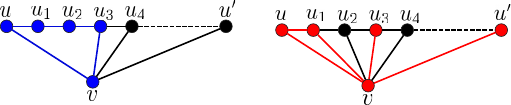}
    \caption{Schematic depiction of \Cref{obs:hole-or-claw-triangle1}}
    \label{fig:cycle-reduction}
\end{figure}
\fi

\ifshort
\else
\begin{proof}
  As $P$ is an induced path, there may exist an edge between $v$ and the vertices in $P$.
  We have the following two cases:
  \begin{description}
    \item[Case-(i)] When there exists a vertex $x\in P$ such that $vx\notin E(G)$.
    Let ${\ell}_x$ be the first vertex on $P$ that is adjacent to $x$ if traversed from $x$ towards $u$, and $r_x$ be the first vertex on $P$ that is adjacent to $x$ if traversed from $x$ towards $u'$.
    Clearly, ${\ell}_x$ and $r_x$ are two different vertices.
    Therefore, $v$-${\ell}_x$-$P$-$r_x$-$v$ is an induced cycle of length at least 4.    
    \item[Case-(ii)] When $v$ is adjacent to all the vertices in $P$.
    Given that the length of $C$ is at least 6, there must be at least 5 vertices $u, u_1, u_2, u_3, u_4$ in $P$ in this sequence.
    Observe that out of these 5 vertices, $\{u, u_2, u_4\}$ forms an independent set.
    As $v$ is adjacent to all the vertices in $P$, $J:=\{v, u, u_2, u_4\}$ induces a claw such that $v$ is the center of $J$.
    Additionally, observe that $T:=\{v, u, u_1\}$ induces a triangle.
    Therefore, there exists a claw $J$ and a triangle $T$ such that $J \cap T$ contains $v$ implying that $J \cap T \neq \emptyset$.
  \end{description}
  This completes the proof of the observation.
\end{proof}
\fi
\ifshort
\else
% {\color{blue}By the consequences of \Cref{obs:double-edge}, the arguments of the proofs of \Cref{lem:safe-ctpair,lem:safe-hole-7} omits the fact $v$ is incident to any double-edge. 
% Moreover, a hole of length at least 7 and \ctpair\ can not have double-edge as an induced subgraph of $G$, $G-X$ or $G-(X\cup \{v\})$.
% Therefore, all the edges involved in the proofs of \Cref{lem:safe-hole-7,lem:safe-ctpair} are simple edge and the cardinality of $P_{Z, f}$ excludes the number of vertices that may match $(Z, f)$.
% }
%\todo[inline, color=green]{Arpit: all the red color texts are not needed.}
\fi
\ifshort
\begin{lemma}[$\star$]
\else
\begin{lemma}
\fi
\label{lem:safe-hole-7}
    Let $v\in K$ be an unmarked vertex by the procedure {\sf Mark-Clique}($K$), and $X\subseteq V(G)\setminus \{v\}$ with $|X|\le k$.
    If $G-X$ has a hole $H$ of length at least $7$ then $G-(X\cup \{v\})$ also has a hole of length at least $4$ or a claw $J$ and a triangle $T$ such that $J \cap T \neq \emptyset$.
\end{lemma}

\ifshort
\else
\begin{proof}
    For the sake of contradiction, let us assume that $G-(X\cup \{v\})$ has neither any hole of length at least $4$ nor any claw and a triangle that are pairwise intersecting, but $G - X$ has a hole $H$ of length at least 7.
    It implies that $v\in H$, $|H_S|\ge 1$ and $|H_K| \le 2$.
    %Moreover, $|H_S|\neq \emptyset$ otherwise, $H$ must be completely contained in $G[V_1]$ and $G[V_1]$ is a collection of {\pig s}.
    Now, we have the following three mutually exclusive cases based on the structure of $H$ in $G-X$:
    \begin{description}
        \item[Case-(i)] When $|N(v)\cap H_S|=0$.
        Let $N(v)\cap H = \{u, w\}$ and $u'$, $w'$ be the neighbors of $u$ and $w$ in $H\setminus \{v\}$ respectively.
        WLOG, assume that $u \uv v \uv w$.
        Let us pick an arbitrary vertex $x$ from $H_S$.
        Define a function $g:\{x\}\to \{0, 1\}$ such that $g(x)=0$.
        By the definition in step-(\ref{mark:subet-3}) of {\sf Mark-Clique}($K$), $v\in P_{\{x\}, g}$.
        %{Clearly, $v\in P_{\{x\}, g}$}.
        Now, we have three situations: (a) $u\in H_R$ and $w\in H_K$; (b) $u\in H_K$ and $w\in H_R$; and (c) $u\in H_R$ and $w\in H_R$.
        The argument is symmetric for situations (a) and (b).
        We give the argument for situations (a) and (c).

        In situation (a), since $v$ is unmarked and $v \in P_{\{x\}, g}$, due to by \Cref{obs:mark-num} of \Cref{obs:mark-signature}, $|P_{\{x\}, g}| \geq 3k + 6$.
        It follows from \Cref{obs:3-marked} of \Cref{obs:mark-signature} that there are at least three marked vertices in $K\setminus X$ that match $(\{x\}, g)$ and each appear before $v$.
        Consider $v' \uv v$ as one of those 3 marked vertices.
        As $v' \uv v$ and $u \uv v'$, it follows that $u \uv v' \uv v$.
        As $u$ is adjacent to $v$, $u$ is adjacent to $v'$.
        Additionally, $v' \in K$ implies that $v'$ is adjacent to $w$.
        Now, we consider the path  $P$ that we obtain by removing $v$ from $P$.
        This is the path that starts from $u$, proceeds along the edges of $H$ and ends at $w$.
        Observe that $P$ is an induced path in $G-(X\cup \{v\})$ and $P \cup \{v'\} = (H \setminus \{v\}) \cup \{v'\}$ is a cycle of length {at least 7} in $G - (X \cup \{v\})$.
        % of length at least 6 and $v'$ is adjacent to $u$ and $w$ in $G-(X\cup \{v\})$.
        Then, $(H \setminus \{v\}) \cup \{v'\}$ is a cycle of length at least 7 and $P$ is an induced path in $G-(X\cup \{v\})$.
        Therefore, due to \Cref{obs:hole-or-claw-triangle1}, we have an induced cycle of length at least $4$, or we have a claw $J$ and a triangle $T$ such that $J \cap T \neq \emptyset$ in $G-(X\cup \{v\})$.
        
        	The situation (b) is symmetric to situation (a).

        In situation (c), since $v$ is unmarked and due to \Cref{obs:3-marked} of \Cref{obs:mark-signature}, there are at least {three} marked vertices in $(K\setminus X)$ that match $(\{x\}, g)$ and appear before $v$ in $\mathcal{V}$.
        Suppose that $v' \in K \setminus X$ is one such vertex such that $v'\uv v$.
        In addition, there are at least three marked vertices in $(K\setminus X)$ that match $(\{x\}, g)$ and appear after $v$ in $\mathcal{V}$.
        We assume that $v'' \in K \setminus X$ is one such vertex such that $v\uv v''$.
        Note that $P$, the path starting from $u$ and ending at $w$ containing the edges of $H$ is an induced path in $G-(X\cup \{v\})$.
        With choice of $v'$ and $v''$, it is easy to see that $uv', v''w\in E(G-(X\cup \{v\}))$ by \Cref{prop:pig-clique}.
        If $v'$ is adjacent to $w$ or $v''$ is adjacent to $u$, then $H'=(H\setminus \{v\})\cup \{v'\}$ or $H''=(H\setminus \{v\})\cup \{v''\}$ is a cycle with length at least 7 and $P$ is an induced path in $G-(X\cup \{v\})$.
        Therefore, we have an induced cycle of length at least $4$ or a claw $J$ and a triangle $T$ such that $J \cap T \neq \emptyset$ in $G-(X\cup \{v\})$ by \Cref{obs:hole-or-claw-triangle1}.
        
        Now, consider the case when $v'w, v''u \notin E(G-(X\cup \{v\}))$.
        Then $H'= (H\setminus \{v\})\cup \{v', v''\}$ is a cycle of length at least 8 and $P$ is an induced path of length at least 5 in $G-(X\cup \{v\})$.
        Observe that there is $x \in H_S$ and by choice of $v', v''$, it follows that $x$ is neither adjacent to $v'$ nor adjacent to $v''$.
        Consider the first vertex $x_{\ell}^1$ that is adjacent to $v'$ and appears in the path $x$-$P$-$u$ and the first vertex $x_r^1$ that is adjacent to $v'$ and appears in the path $x$-$P$-$w$.
        Note that $x_{\ell}^1$ definitely exists while $x_r^1$ may not exist.
        Similarly, consider the first vertex $x_{\ell}^2$ that is adjacent to $v''$ and appears in the path $x$-$P$-$u$ and the first vertex $x_r^2$ that is adjacent to $v''$ and appears in the path $x$-$P$-$w$.
        Note that $x_r^2$ definitely exists but $x_{\ell}^2$ may not exist.
        
        If both $x_{\ell}^1$ and $x_r^1$ exist, then $v'$-$x_{\ell}^1$-$P$-$x$-$P$-$x_r^1$-$v'$ forms an induced cycle of length at least 4 in $G - (X \cup \{v\})$.
        Similarly, if both $x_{\ell}^2$ and $x_r^2$ exist, then $v'$-$x_{\ell}^2$-$P$-$x$-$P$-$x_r^2$-$v$ forms an induced cycle of length at least 4 in $G - (X \cup \{v\})$.
        This takes care of the scenario when one of $x_r^1$ and $x_{\ell}^2$ exists.
        
        In case neither $x_{r}^1$ and $x_{\ell}^2$ do not exist, then let us observe the following.
        Note that $v'$ is not adjacent to any vertex appearing in the path $x$-$P$-$w$ and $v''$ is not adjacent to any vertex appearing in the path $x$-$P$-$u$.
        Then, consider $v'$-$x_{\ell}^1$-$P$-$x$-$P$-$x_r^2$-$v''$.
        Note that this subgraph is an induced cycle with at least 5 vertices.
        Hence, there is an induced cycle of length at least 5 in $G - (X \cup \{v\})$.
%        	\todo[inline]{Diptapriyo: this part is not clear.}
%        	In particular, $v'$-$u$-$P$-$x_4$ 
%        	Note that $v'$ is adjacent to $u$
%        Assume that there is a vertex $z \in P$ such that $v'z$ is not an edge.
%        Then traversing $P$ from $z$ towards $u$, there exists $\ell_z$ that is the first vertex towards $u$ such that $\ell_z v'$ is an edge of $G - (X \cup \{v\})$.
%        If there exists $r_z$ that is the first vertex towards $w$ when traversing through satisfying that $ v'r_z \in E(G-(X\cup \{v\}))$, then there is $v'$-$\ell_z$-$P$-$r_z$-$v'$ is an induced cycle of length at least $4$ in $G-(X\cup \{v\})$.
%       \todo[inline]{Diptapriyo: the above type $r_z$ may not exist.}
%       If no vertex of $P$ from $z$ towards $w$ is adjacent to $v'$, then note the following.
%       The path $P^* = \ell_z$-$P$-$w$ is an induced path.
%        Similarly, we can argue the same for $v''$ when $v''$ is not adjacent to a vertex $z \in P$. and find an induced cycle of length at least $4$ in $G-(X\cup \{v\})$.
%        
%        Now, if $v'$ and $v''$ are adjacent to all the vertices in $P'\setminus \{u, w\}$ then it is easy to see that $P'':= u'$-$u''$-$P'$-$w'$-$w'$ is an induced path of length at least $5$ and $P''\cup \{v\}$ is a cycle of length at least 6 in $G-(X\cup \{v\})$.
%        Therefore, we have an induced cycle of length at least $4$ or a claw $J$ and a triangle $T$ such that $J \cap T \neq \emptyset$ in $G-(X\cup \{v\})$ by \Cref{obs:hole-or-claw-triangle1}.

        \item[Case-(ii)] When $|N(v)\cap H_S|=1$.
        Consider $N(v)\cap H_S=\{x\}$.
        Let us define a function $g:\{x\}\to \{0, 1\}$ such that $g(x)=1$.
        In addition, let $u$ be the other neighbor of $v$ in $H$.
        Assume that $v \uv u$ (for $u\uv v$, the argument is symmetric).
        Consider the path $P$ from $x$ to $u$ that follows the edges of $H \setminus \{v\}$.
        Note that $x$-$P$-$u$ is an induced path.
        In $G-(X\cup \{v\})$ containing at least 6 vertices.
        Now, we have two situations: either $u \in H_K$ or $u \in H_R$.
        Since $v$ is unmarked and due to \Cref{obs:3-marked} of \Cref{obs:mark-signature}, there are at least three marked vertices $v'$ in $K\setminus X$ that matches $(\{x\}, g)$, appear before $v$ in $\mathcal{V}$, and $v' \uv v$ is one of them.
        There are at least three marked vertices in $K\setminus X$ (including $w'$) that matches $(\{x\}, g)$ such that $v \uv w'$. 
        Now, pick any $v'$ and $w'$ out of these 3 vertices.
        If $u \in H_K$, then we consider the vertex $v'$ such that $v' \uv v$.
        Then, $P \cup \{v'\}$ is a cycle with 7 vertices and $P$ is an induced path with at least 6 vertices.
        Hence, due to \Cref{obs:hole-or-claw-triangle1}, then $G - (X \cup \{v\})$ has an induced cycle of length at least 4 or a claw $J$ and a triangle $T$ such that $J \cap T \neq \emptyset$.
        	
        Similarly, if $u \in H_R$, then we consider the $w'$ such that $v \uv w'$.
        Observe that $w'$ is adjacent to $x$ and $u$.
        Then, $P \cup \{w'\}$ is a cycle of length 7.
        Hence, due to \Cref{obs:hole-or-claw-triangle1}, $G - (X \cup \{v\})$ has an induced cycle of length at least 4 or a claw $J$ and a triangle $T$ such that $J \cap T \neq \emptyset$.
        
        \item[Case-(iii)] When $|N(v)\cap H_S|=2$.
        Say $N(v)\cap H_S=\{x, y\}$.
        Let us define a function $g:\{x, y\}\to \{0, 1\}$ such that $g(x)=1$ and $g(y)=1$.
        Clearly, $v\in P_{\{x,y\}, g}$.
        Since $v$ is unmarked and due to \Cref{obs:mark-num} of \Cref{obs:mark-signature}, there are at least $3k+6$ marked vertices in $K\setminus X$ that match $(\{x, y\}, g)$.
        Now, at most $k$ vertices out of these $3k+6$ are incident to a double-edge.
        We pick one marked vertex $v'$ from $K\setminus X$ that matches $(\{x, y\}, g)$ and $v'$ is not incident to a double-edge.
       	Observe that $H \setminus \{v\}$ is an induced path $P$ with at least 6 vertices.
        Now, we consider the cycle $P \cup \{v'\}$.
        Observe that $P \cup \{v'\}$ forms a cycle of length at least 7.
        Then, due to \Cref{obs:hole-or-claw-triangle1}, $G - (X \cup \{v\})$ has an induced cycle of length at least $4$ or a claw $J$ and a triangle $T$ such that $J \cap T \neq \emptyset$.
    \end{description}
    This completes the proof of the lemma.
\end{proof}
\fi

\ifshort
\begin{lemma}[$\star$]
\else
\begin{lemma}
\fi
\label{lem:safe-ctpair}
    Let $v\in K$ be an unmarked vertex by the procedure {\sf Mark-Clique}($K$), and let $X\subseteq V(G)\setminus \{v\}$ with $|X|\le k$.
    If $G-X$ has a connected component containing a claw $J$ and a triangle $T$, then $G-(X\cup \{v\})$ has a connected component that has a claw and a triangle.
\end{lemma}

\ifshort
\else
\begin{proof}
  For the sake of contradiction, let us assume that $G - X$ has a connected component containing a claw $J$ and a triangle $T$, but $G-(X\cup \{v\})$ is a has no connected component containing both a claw and a triangle.
  Then, $v$ must be part of the component $D$ of $G - X$ that contains the claw $J$ and the triangle $T$.
  In addition, $v \in K$ but $v$ is unmarked as per {\sf Mark-Clique}($K$).
  Then, due to \Cref{obs:mark-num} of \Cref{obs:mark-signature}, $G - X$ contains at least $3k+6$ marked vertices excluding $v$ as per the procedure {\sf Mark-Clique}($K$).
  In particular, all those $3k + 6$ vertices are present in $K \setminus X$.
  Hence, in $G - (X \cup \{v\})$, there is a component that contains 3 vertices {out of these $3k+6$ vertices} of $K$ which form a triangle.
  Additionally, in $G - X$, the component $D$ that contains $v$ also contains these 3 vertices of $K$ that are not deleted by $X$.
  %Hence, $v$ is part of a $K_4$, implying that $v$ is part of a triangle in the component $D$ of $G - X$ where $v$ is part of.
  We use this crucial observation.
  Now, our situation is divided into the following cases, each concerning if $v$ is part of a claw in $D$ or $v$ is not part of a claw in $D$.
  
  Consider the situation when $v$ is part of a claw $J$ in $D$.
  Then, there must exist at least one vertex $x$ in $J \setminus \{v\}$ such that $x\in S$.
  Because $J$ cannot be completely contained in $G[V_1]$; due to the fact that $G[V_1]$ is a collection of {\pig s}.
  Let $J_S = J \cap S$, $J_R = J \cap R$, and $J_K = J \cap K$.
  We have the following three mutually exclusive cases based on the structure of ${J}$ in $G-X$:
  \begin{description}
    \item[Case-(i)] When $v$ is the center of the claw $J$ in $D$.
    Now, we have the following three situations:
    \begin{description}
        \item[Subcase-(a)] When $|J_S|=1$, say $J_S = \{x\}$.
        Clearly, $x$ is a pendant vertex of the claw $J$ in $D$.
        It means that the other pendant vertices of the claw $J$ are in $J_K \cup J_R$; say $u$ and $w$ be are the other pendant vertices of $J_K \cup J_R$ such that $u \uv v \uv w$.
        Now, we have following three situations: (a-i) $u \in {J}_R$ and $w \in {J}_K$; (a-ii) $u \in J_K$ and $w\in {J}_R$; and (a-ii) $u, w \in {J}_R$.
        The argument is symmetric for situations (a-i) and (a-ii).
        We give the arguments for situations (a-i) and (a-iii).

        In situation (a-i), define a function $g: J_S \to \{0,1\}$ such that $g(x)=1$.
        Clearly, $v \in P_{\{x\}, g}$.
        Since $v$ is unmarked and {due to \Cref{obs:3-marked} of \Cref{obs:mark-signature}}, there are at least 3 marked vertices $v_1', v_2', v_3'$ in $K \setminus X$ that match $(\{x\}, g)$ such that $v_1', v_2', v_3' \uv v$.
        Observe that each of $v_1', v_2', v_3'$ is adjacent to $w$ as $w \in K$.
        Moreover, each of $v_1', v_2', v_3'$ is adjacent to $u$ since these vertices appear between $u$ and $v$ in $\mathcal{V}$.
        Hence, $v_1', v_2', v_3'$ are adjacent to $u$ by \Cref{prop:pig-clique}. 
        Now, we consider $\{v_1', x, u, w\}$.
        Observe that $\{u, x, w\}$ forms an independent set.
        Hence, $\{v_1', x, u, w\}$ forms a claw in $G - (X \cup \{v\})$.
        Subsequently, there is clearly  a triangle $\{v_1', v_2', v_3'\}$ in $G - (X \cup \{v\})$.
        Therefore, a connected component of $G - (X \cup \{v\})$ has a claw and a triangle. 
%        	
%        Now, $J':= (J\setminus \{v\})\cup \{v'\}$ induces a claw in $G-(X\cup \{v\})$.
%        {\color{blue}If $v$ is not the only vertex in $\widehat{J}$ that connects the claw $J$ and the triangle $T$, then it is easy to see that $(\widehat{J}\setminus \{v\})\cup \{v'\}$ is an induced {\ctpair} in $G-(X\cup \{v\})$.
%        Now, if $v$ is the only vertex in $\widehat{J}$ that connects the claw $J$ and the triangle $T$, then the vertex, say $y$, in $P\cup T$ that is adjacent to $v$ must be adjacent to $v'$ in $G-(X\cup \{v\})$ by \Cref{prop:pig-clique} if $y\in K\cup R$.
%        If $y\in S$, then define a function $g':\{x, y\}\to \{1\}$.
%        Clearly, $v\in P_{\{x, y\}, g'}$.
%        Since $v$ is unmarked, there are at least one marked vertex $v''$ in $K\setminus X$ that matches $(\{x, y\}, g')$ such that $v''\uv v$.
%        Now, $J'':= (J\setminus \{v\})\cup \{v''\}$ induces a claw in $G-(X\cup \{v\})$ and $v''$ is adjacent to $y$ in $G-(X\cup \{v\})$.
%        Therefore, $(\widehat{J}\setminus \{v\})\cup \{v''\}$ is an induced {\ctpair} in $G-(X\cup \{v\})$.}
%        \todo[inline, color=green]{Arpit: These blue colored texts are exactly some times. Can we write it like the rest of the argument is symmetric with subcase (a)?}

        In situation (a-iii), consider the pair $x$ and $u$.
        Note that $u \in K_{pv}$ is a nonneighbor of $x$ in $K_{pv}$.
        {Since $v$ is unmarked and $u$ exists in $G-X$, it implies that $u$ is a nonneighbor of $x$ in $K_{pv} \setminus X$.
        % there are last $2k+1$ vertices $u'\in K_{pv}$ that are non-neighbor to $x$.
        %It is possible that $u=u'$ for some of these vertices in $K_{pv}$.
        Either $u$ is one of the last $2k + 1$ nonneighbors of $x$ in $K_{pv}$, or there exists a vertex in $K_{pv}$ appearing after $u$ and is one of the last $2k + 1$ nonneighbors of $x$ in $K_{pv}$.
        	In both the cases, note that $u$ is not deleted by $X$.
        	Additionally, note that in $K_{pv} \setminus X$, there exists a vertex $u'$ (which could be $u$ itself) which is a nonneighbor of $x$ in $K_{pv}$.
        	As $u \euv u' \uv v$, and $uv$ is an edge $vu'$ is an edge.
        	Observe that $v$ is a common neighbor of $x$ and $u'$, $v$ is unmarked, and $v$ is not incident to a double-edge.
        	Then, there are $2k + 3$ vertices of $K$ that appear are common neighbors of $x$ and $u'$ that are marked by the step-(\ref{mark:k**2}) {\sf Mark-Clique}($K$).
        	As $|X|\le k$ and at most $k$ many vertices out of these $2k+3$ common neighbors are incident to a double-edge, hence, there exist at least one vertex $v' \in K \setminus X$ such that $v'$ is a common neighbor of $x$ and $u'$.}
        %However, our argument works in both cases.
        %As $|X|\le k$ and at most $k$ many vertices out of these $2k+3$ common neighbors are incident to a double-edge, hence, there exist at least one vertex $u'\in K_{pv}\setminus X$ out of these $2k+1$ vertices in $K_{pv}$ such that $u\euv u'$ and $u'$ is not incident to any double-edge.
        %{If $u$ is one of the last $k+1$ nonneighbors of $x$ in $K_{pv}$, then $v$ is unmarked vertex among the common neighbors of $x$ and $u$ in $K$.}
        Now, there are at least $k+3$ common neighbors of $x$ and $u'$ in $K\setminus X$ that are marked, and they appear after $v$ in the proper interval ordering $\mathcal{V}$.
        At least 3 vertices out of these $k+3$ marked vertices in $K$ are not incident to any double-edge.
        Hence, there are 3 vertices $v_1', v_2', v_3 \in K \setminus X$ such that $v \uv v_1', v_2', v_3'$ and they are not incident to any double-edge.
        In such a case, $v_1'$ is adjacent to $w$ {since $v\uv v_1' \uv w$ and $vw\in E(G-X)$.}
        Therefore, $\{v_1', x, w, u'\}$ induces a claw.
        In addition, $v_1', v_2', v_3'$ forms a triangle.
        Hence, $G - (X \cup \{v\})$ has a connected component that contains a claw and a triangle.

        \item[Subcase-(b)] When $|J_S| = 2$; say $J_S=\{x, y\}$.
        Clearly, $x$ and $y$ are pendant vertices of the claw $J$ in the component $D$.
        It means that the other pendant vertex of the claw $J$ is in $J_K \cup J_R$, say $u$ be that pendant vertex.
        It is possible that $u \uv v$ or $v \uv u$.
        Without loss of generality, assume that $u \uv v$ (if $v \uv u$, the argument is symmetric).
        We define a function $g:\{x, y\}\to \{0, 1\}$ as $g(x) = g(y) = 1$.
        Clearly, $v\in P_{\{x, y\}, g}$.
        % \replace{}{but $v$ is unmarked}
        {Due to \Cref{obs:3-marked} of \Cref{obs:mark-signature} and $v$ is unmarked}
        %{Then, there are $k+3$ vertices in $P_{\{x, y\}, g} \cap K$ that are marked and appear before $v$ and after $v$ as per the ordering $\mathcal{V}$.In particular}, 
        there are 3 vertices $v_1', v_2', v_3'$ in $P_{\{x, y\}, g}$ that are marked, are present in $K \setminus X$, such that $v_1', v_2', v_3' \uv v$.
                	
        Now, we have the following two situations: (a) $u \in {J}_K$ and (b) $u \in {J}_R$.
        In both the cases, $u \uv v$.
        We consider the vertices $v_1', v_2', v_3'$ in $P_{\{x, y\}, g}$ that are marked, are present in $K \setminus X$, such that $v_1', v_2', v_3' \uv v$.
        	As $u \uv v$ and $v_1' \uv v$, it follows that $v_1'$ is adjacent to $u$.
        	%%%%
   %     Since $v$ is unmarked, there must exist at least one vertex $v'$ in $K\setminus X$ that matches $(\{x, y\}, g)$ such that $v\uv v'$ if $u\in \widehat{J}_K$; and $v'\uv v$ if $u\in \widehat{J}_R$.
        Now, $J':= (J\setminus \{v\}) \cup \{v_1'\} = \{v_1', u, x, y\}$ induces a claw in $G-(X\cup \{v\})$.
        Additionally, $\{v_1', v_2', v_3'\}$ forms a triangle in $G - (X \cup \{v\})$.
        Hence, there is a connected component containing a claw and a triangle in $G - (X \cup \{v\})$.
        	
%        {\color{blue}If $v$ is not the only vertex in $\widehat{J}$ that connects the claw $J$ and the triangle $T$, then it is easy to see that $(\widehat{J}\setminus \{v\})\cup \{v'\}$ is an induced {\ctpair} in $G-(X\cup \{v\})$.
%        Now, if $v$ is the only vertex in $\widehat{J}$ that connects the claw $J$ and the triangle $T$, then the vertex, say $y$, in $P\cup T$ that is adjacent to $v$ must be adjacent to $v'$ in $G-(X\cup \{v\})$ by \Cref{prop:pig-clique} if $y\in K\cup R$.
%        If $y\in S$, then define a function $g':\{x, y\}\to \{1\}$.
%        Clearly, $v\in P_{\{x, y\}, g'}$.
%        Since $v$ is unmarked, there are at least one marked vertex $v''$ in $K\setminus X$ that matches $(\{x, y\}, g')$ such that $v''\uv v$.
%        Now, $J'':= (J\setminus \{v\})\cup \{v''\}$ induces a claw in $G-(X\cup \{v\})$ and $v''$ is adjacent to $y$ in $G-(X\cup \{v\})$.
%        Therefore, $(\widehat{J}\setminus \{v\})\cup \{v''\}$ is an induced {\ctpair} in $G-(X\cup \{v\})$. }

        \item[Subcase-(c)] When $|{J}_S|=3$, say $J_S = \{x, y, z\}$.
        Define a function $g:\{x, y, z\}\to \{0, 1\}$ such that $g(x) = g(y) = g(z) = 1$.
        Clearly, $v \in P_{\{x, y, z\}, g}$.
        Since $v$ is unmarked and due to \Cref{obs:3-marked} of \Cref{obs:mark-signature},
        %{}, 
        there are at least 3 marked vertices $v_1', v_2', v_3'$ in $K \setminus X$ that matches $(\{x, y, z\}, g)$ such that $v \uv v_1', v_2', v_3'$.
        Now, $J':= (J\setminus \{v\})\cup \{v_1'\}$ and $v_1'$ is adjacent to each of $x, y$ and $z$.
        Observe that $J'$ induces a claw in $G-(X\cup \{v\})$ and $\{v_1', v_2', v_3'\}$ is a triangle in $G - (X \cup \{v\})$.
        Hence, $G - (X \cup \{v\})$ has a connected component that has a claw and a triangle.
%        {\color{blue}If $v$ is not the only vertex in $\widehat{J}$ that connects the claw $J$ and the triangle $T$, then it is easy to see that $(\widehat{J}\setminus \{v\})\cup \{v'\}$ is an induced {\ctpair} in $G-(X\cup \{v\})$.
%        Now, if $v$ is the only vertex in $\widehat{J}$ that connects the claw $J$ and the triangle $T$, then the vertex, say $y$, in $P\cup T$ that is adjacent to $v$ must be adjacent to $v'$ in $G-(X\cup \{v\})$ by \Cref{prop:pig-clique} if $y\in K\cup R$.
%        If $y\in S$, and $vy\in E(G-(X\cup \{v\}))$ then define a new function $g':\{y\}\to \{1\}$.
%        Clearly, $v\in P_{\{y\}, g'}$.
%        Since $v$ is unmarked, there are at least one marked vertex $v''$ in $K\setminus X$ that matches $(\{y\}, g')$ such that $v''\uv v$.
%        Now, $\widehat{J}'':= (\widehat{J}\setminus \{v\})\cup \{v', v''\}$ induces a claw in $G-(X\cup \{v\})$ and the claw $J'$ is connected to triangle $T$ via $v''$ and $y$ because $v''$ is adjacent to $y$ in $G-(X\cup \{v\})$.
%        Therefore, $(\widehat{J}\setminus \{v\})\cup \{v', v''\}$ is an induced {\ctpair} in $G-(X\cup \{v\})$.}
    \end{description}
    
    \item[Case-(ii)]\label{ct-pendant} When $v$ is a pendant vertex of the claw $J$ in $D$.
    Now, we have the following three situations:
    \begin{description}
        \item[Subcase-(a)] When the center of the claw $J$ is in $S$.
        Say the center of the claw $J$ is $x$, and $y, z$ be the other pendant vertices of $J$.
        Now, we have following three situations:
        \begin{description}
            \item[Sub-subcase-(a-i)] When $y, z \in {J}_R$ and suppose that $K_{nt}$ and $K_{pv}$ appear in the same connected component of $G - S$.
            Clearly, $x \in J_S$.
            %	\todo[inline]{Diptapriyo: what if $K_{nt}$ or $K_{pv}$ does not appear in the same component of $G - S$?}
         
        		Consider the situation when $y, z \notin K_{pv} \cup K_{nt}$.
           	Then, we consider the function $g: \{x\} \rightarrow \{0, 1\}$ such that $g(x) = 1$.
           	Since $v$ is unmarked and due to \Cref{obs:mark-num} of \Cref{obs:mark-signature}, there are at least $3k+6$ marked vertices in $K \setminus X$ that match $(\{x\}, g)$.
            Out of these $3k + 6$ vertices, $2k+6$ vertices are not incident to any double-edge.
            Now, pick any three such vertices $v_1', v_2', v_3'$ from these $2k + 6$ vertices.
            As $y, z \notin K_{pv} \cup K_{nt}$, due to \Cref{prop:pig-clique-nbd}, $v_1'$ is not adjacent to either of $y$ and $z$.
%            {Observe that $v$ is unmarked but $v$ matches $(\{x\}, g)$.
%           	Then, there are $k+3$ marked vertices that match $(\{x\}, g)$.
%           	Hence, there is $v_1', v_2', v_3' \in K \setminus X$ that are adjacent to $x$.}
           	Then, $\{x, y, z, v_1'\}$ forms a claw and $\{v_1', v_2', v_3'\}$ forms a triangle.
           	Therefore, a claw and a triangle is in a connected component of $G - (X \cup \{v\})$. 
           	%\todo[inline]{Diptapriyo: this part is not correct, as it seems.}
           	            	
            Now, we consider the situation when $y \in K_{pv}$ and $z \in K_{nt}$.
            The arguments for the situations: $y \in K_{pv}$ but $z \notin K_{nt}$ and $y \notin K_{pv}$ and $z \in K_{nt}$ are similar.
            Since $v$ is unmarked and $y\in K_{pv}$ exists in $G-X$, note that $y$ is a neighbor of $x$ in $K_{pv} \setminus X$.
            Either $y$ is one of the first $2k + 1$ neighbors of $x$ in $K_{pv}$ or there are $2k + 1$ neighbors neighbors of $x$ in $K_{pv}$ that appear before $y$ by the step-(\ref{mark:K**4}) of {\sf Mark-Clique}($K$).
            In both cases, note that in $K_{pv} \setminus X$, there is a vertex $y'$ (which could be $y$ itself) that is neighbor to $x$ and $v$ is a neighbor of $x$ but not a neighbor of $y'$ and not incident to any double-edge.
            	As $v \in K$ is unmarked, there are $2k + 3$ vertices in $K$ that are neighbors of $x$, non-neighbors of $y'$, and appear before $v$ in $K$.
            %However, our argument works in each case.
            %As $|X|\le k$ and at most $k$ vertices out of these $2k+3$ are incident to a double edge, hence, there exist at least one vertex $y'\in K_{pv}\setminus X$ such that $y'\euv y$ and $y'$ is not incident to any double-edge.
            %Since $v$ is unmarked and by the procedure of step-(\ref{mark:K**4}) of marking scheme, there are at least $2k+3$ first marked vertices $v'$ in $K$ that are neighbor of $x$ but not $y'$, i.e., $v\in K^{4}_{pv}(x, y')$.
            As $|X|\le k$ and at most $k$ many vertices are incident to double-edges, there are at least three marked $v_1', v_2', v_3'\in K\setminus X$ such that $v_1', v_2', v_3'\uv v$ and $v_1', v_2', v_3'$ are not incident to any double-edge.
            In addition, $v_1'$ is a neighbor of $x$ but not a neighbor of $y'$.
            %	Then, $\{v_1, x, y', 
            Since $y'\euv y\uv v_1'\uv v\uv z$ and $vy, vz\notin E(G-X)$, implying that $v_1'y, v_1'z\notin E(G-X)$ by \Cref{prop:pig-clique}. 
            %}{
%            Since $v$ is unmarked and , there must exist at one vertex $y' \in K_{pv}$ and at least one vertex $z'\in K_{nt}$ such that $y'$ and $z'$ are adjacent to $x$ in $G$.
%            %% in $G-(X\cup \{v\})$.
%            It is possible that $y=y'$ and $z=z'$.
%            If $y \neq y'$, then $y'$ is one of the first $k+1$ neighbors of $x$ in $K_{pv}$ and $y' \uv y$.
%            Similarly, if $z \neq z'$, then $z'$ is one of the last $k+1$ neighbors of $x$ in $K_{nt}$ and $z \uv z'$.
%            As $|X| \leq k$, both $z'$ and $y'$ are in $G - (X \cup \{v\})$.
%            Then, $v \in K$ is a vertex that is adjacent to $x$ but neither adjacent to $y$ nor adjacent to $z$.
%            Then, $v$ is neither adjacent to $y$ nor adjacent to $z$, $v$ is neither adjacent to $y'$ nor adjacent to $z'$.
%            But, at least $k+3$ vertices of $K$ are marked that are adjacent to $x$ but not adjacent to $y'$ but appear before $v$ in $\mathcal{V}$.
%            Hence, in $K \setminus X$, there are 3 vertices $v_1', v_2', v_3'$ that are adjacent to $x$ but not adjacent to $y'$ and $v_1', v_2', v_3' \uv v$.
%            Additionally, $v_1' \uv v \uv z \uv z'$.
%       		Hence, $v_1'$ is not adjacent to $z'$ either.}
            	
            Then, $J':= \{x, y', v_1', z\}$ induces a claw in $G-(X\cup \{v\})$ such that $x$ is the center.
            Additionally, $\{v_1', v_2', v_3'\}$ forms a triangle in $K \setminus X$.
            Therefore, a $G - (X \cup \{v\})$ contains a claw and a triangle that pairwise intersect with each other.
            	Hence, a component exists in $G - (X \cup \{v\})$ that has a claw and a triangle.
            	
            	The case when $K_{pv}$ or $K_{nt}$ does not appear in the same component of $G - S$ work similarly.
            	If $K_{pv}$ appears in the same component of $G - S$ while $K_{nt}$ is not, then the case is similar to the situation when $y \in K_{pv}$ but $z \notin K_{nt}$.
            	Similarly, if $K_{nt}$ appears in the same component of $G - S$ while $K_{pv}$ does not, then the situation is symmetric.
            	If neither $K_{pv}$ nor $K_{nt}$ appear in the same component of $G - S$, then the case is similar to $y, z \notin K_{pv} \cup K_{nt}$.

            \item[Sub-subcase-(a-ii)] When $y \in {J}_R$ and $z \in {J}_S$.
            WLOG, assume that $v \uv y$ (the arguments for $y \uv v$ are symmetric).
            Define a function $g:\{x, z\}\to \{0, 1\}$ such that $g(x)=1$ and $g(z)=0$.
            Clearly, $v\in P_{\{x, z\}, g}$.
            Since $v$ is unmarked and due to \Cref{obs:3-marked} of \Cref{obs:mark-signature}, there are at least 3 marked vertices $v_1', v_2', v_3'$ in $K\setminus X$ that match $(\{x, z\}, g)$ such that $v_1', v_2', v_3' \uv v$.
            We consider $J':= (J\setminus \{v\})\cup \{v_1'\}$ and $\{v_1', v_2', v_3'\}$.
            Note that $v_1' \uv v \uv y$ and $vy\notin E(G-X)$, hence $v_1'$ is not adjacent to $y$ by \Cref{prop:pig-clique}.
            Additionally, $v_1'$ is adjacent to $x$ but not adjacent to $z$.
            Hence, $J' = (J \setminus \{v\}) \cup \{v_1'\}$ induces a claw in $G-(X\cup \{v\})$.
            Additionally, $\{v_1', v_2', v_3'\}$ forms a triangle appearing in $K \setminus X$.
            Therefore, $G - (X \cup \{v\})$ contains a claw and a triangle that are pairwise intersecting.
            Hence, $G - (X \cup \{v\})$ contains a component that contains both a claw and a triangle.
%            {\color{blue}If $v$ is not the only vertex in $\widehat{J}$ that connects the claw $J$ and the triangle $T$, then it is easy to see that $(\widehat{J}\setminus \{v\})\cup \{v'\}$ is an induced {\ctpair} in $G-(X\cup \{v\})$.
%            Now, if $v$ is the only vertex in $\widehat{J}$ that connects the claw $J$ and the triangle $T$, then the vertex, say $u$, in $P\cup T$ that is adjacent to $v$ must be adjacent to $v'$ in $G-(X\cup \{v\})$ by \Cref{prop:pig-clique} if $u\in K\cup R$.
%            If $u\in S$, then define a function $g':\{x, z, u\}\to \{0, 1\}$ such that $g'(x)=g'(u)=1$ and $g'(z)=0$.
%            Clearly, $v\in P_{\{x, z, u\}, g'}$.
%            Since $v$ is unmarked, there are at least one marked vertex $v''$ in $K\setminus X$ that matches $(\{x, z, u\}, g')$ such that $v''\uv v$.
%            Now, $J'':= (J\setminus \{v\})\cup \{v''\}$ induces a claw in $G-(X\cup \{v\})$ and $v''$ is adjacent to $u$ in $G-(X\cup \{v\})$.
%            Therefore, $(\widehat{J}\setminus \{v\})\cup \{v''\}$ is an induced {\ctpair} in $G-(X\cup \{v\})$.}

            \item[Sub-subcase-(a-iii)] When $y, z\in {J}_S$, i.e., $J_S = \{x, y, z\}$.
            Define a function $g:\{x, y, z\}\to \{0, 1\}$ such that $g(x)=1$ and $g(y)=g(z)=0$.
            Clearly, $v\in P_{\{x, y, z\}, g}$.
            Since $v$ is unmarked and due to \Cref{obs:mark-num} of \Cref{obs:mark-signature}, there are at least $3k+6$ marked vertices in $K\setminus X$.
            At most $k$ vertices are incident to a double-edge, hence, there are at least $2k+6$ marked vertices such that match $(J_S, f)$ and not incident to any double-edge.
            Pick any such three vertices $v_1', v_2', v_3'$
            %{, there are at least 3 marked vertices $v_1', v_2', v_3'$ in $K \setminus X$ that match $(\{x, y, z\}, g)$ such that $v \uv v_1', v_2', v_3'$.}
            Consider $J' = (J \setminus \{v\}) \cup \{v_1'\}$.
            Observe that $v_1'$ is adjacent to $x$ but not adjacent to $y$ and $z$.
            Hence, $J'$ forms a claw.
            Additionally, as $v_1', v_2', v_3' \in K$, hence $\{v_1', v_2', v_3'\}$ is a triangle.
            Then, $J'$ and $\{v_1', v_2', v_3'\}$ are pairwise intersecting.
            Hence, $G - (X \cup \{v\})$ has a claw and a triangle that appear in the same connected component since they are pairwise intersecting.
        \end{description}
        \item[Subcase-(b)]\label{ct-ct-k} When the center of the claw $J$ is in $K$ but $v$ is a pendant vertex of $J$.
        Say the center of the claw $J$ is $u$, and $x, y$ be the other pendant vertices of the claw $J$.
		Without loss of generality assume that $x \in J_S$.
		Then, $y \in J_S$ or $y \in J_R$.
        WLOG, assume that $v \uv u \uv y$.
         If $y \in {J}_R$ (the case of $y \uv u \uv v$ is symmetric).
        As the claw $J$ must intersect with $S \cup V_2$; as it cannot be completely contained in $G[V_1]$ and $G[V_1]$ is a collection of {\pig s}.
        We have the following two situations:
        	%%%%%EDIT TO RESUME
        \begin{description}
            \item[Sub-subcase-(b-i)] When $x, y\in {J}_S$.
            Define a function $g: \{x, y\}\to \{0, 1\}$ such that $g(x) = g(y) = 0$.
            Clearly, $v\in P_{\{x, y\}, g}$.
            Since $v$ is unmarked and due to \Cref{obs:3-marked} of \Cref{obs:mark-signature}, there are at least 3 marked vertices $v_1', v_2', v_3'$ in $K\setminus X$ that match $(\{x, y\}, g)$ such that $v_1', v_2', v_3' \uv v$.
            Now, consider $J':= (J\setminus \{v\})\cup \{v_1'\}$ and $\{v_1', v_2', v_3'\}$.
            Note that $J' = (J \setminus \{v\}) \cup \{v_1'\}$ induces a claw in $G-(X\cup \{v\})$ and $\{v_1', v_2', v_3'\}$ induces a triangle in $G - (X \cup \{v\})$ and these are pairwise intersecting.
            Hence, $G - (X \cup \{v\})$ has a component that contains a claw and a triangle.
%            {\color{blue}If $v$ is not the only vertex in $\widehat{J}$ that connects the claw $J$ and the triangle $T$, then it is easy to see that $(\widehat{J}\setminus \{v\})\cup \{v'\}$ is an induced {\ctpair} in $G-(X\cup \{v\})$.
%            Now, if $v$ is the only vertex in $\widehat{J}$ that connects the claw $J$ and the triangle $T$, then the vertex, say $z$, in $P\cup T$ that is adjacent to $v$ must be adjacent to $v'$ in $G-(X\cup \{v\})$ by \Cref{prop:pig-clique} if $z\in K\cup R$.
%            If $z\in S$, then define a function $g':\{x, y, z\}\to \{0, 1\}$ such that $g'(x)=0$, $g'(y)=0$, and $g'(z)=1$.
%            Clearly, $v\in P_{\{x, y, z\}, g'}$.
%            Since $v$ is unmarked, there are at least one marked vertex $v''$ in $K\setminus X$ that matches $(\{x, y, z\}, g')$ such that $v''\uv v$.
%            Now, $J'':= (J\setminus \{v\})\cup \{v''\}$ induces a claw in $G-(X\cup \{v\})$ and $v''$ is adjacent to $z$ in $G-(X\cup \{v\})$.
%            Therefore, $(\widehat{J}\setminus \{v\})\cup \{v''\}$ is an induced {\ctpair} in $G-(X\cup \{v\})$.}
            \item[Sub-subcase-(b-ii)] When $x \in {J}_S$ and $y \in {J}_R$.
            Define a function $g:\{x\}\to \{0, 1\}$ such that $g(x) = 0$.
            Clearly, $v \in P_{\{x\}, g}$.
            Since $v$ is unmarked and due to \Cref{obs:3-marked} of \Cref{obs:mark-signature}, there are at least 3 marked vertices $v_1', v_2', v_3'$ in $K\setminus X$ that match $(\{x\}, g)$ and $v_1', v_2', v_3' \uv v$.
            Then, $v_1' \uv v \uv y$.
            Since $v$ is not adjacent to $y$, $v_1'$ is also not adjacent to $y$.
            Additionally, $v_1'$ is not adjacent to $x$ also.
            Therefore, $J':= (J\setminus \{v\}) \cup \{v_1'\}$ induces a claw in $G-(X\cup \{v\})$.
            Additionally, $\{v_1', v_2', v_3'\}$ is a triangle in $G - (X \cup \{v\})$.
            Hence, $G - (X \cup \{v\})$ has a claw and a triangle that are pairwise intersecting.
            Therefore, $G - (X \cup \{v\})$ has a component that has a claw and a triangle.
%            {\color{blue}If $v$ is not the only vertex in $\widehat{J}$ that connects the claw $J$ and the triangle $T$, then it is easy to see that $(\widehat{J}\setminus \{v\})\cup \{v'\}$ is an induced {\ctpair} in $G-(X\cup \{v\})$.
%            Now, if $v$ is the only vertex in $\widehat{J}$ that connects the claw $J$ and the triangle $T$, then the vertex, say $z$, in $P\cup T$ that is adjacent to $v$ must be adjacent to $v'$ in $G-(X\cup \{v\})$ by \Cref{prop:pig-clique} if $z\in K\cup R$.
%            If $z\in S$, then define a function $g':\{x, z\}\to \{0, 1\}$ such that $g'(x)=0$ and $g'(z)=1$.
%            Clearly, $v\in P_{\{x, z\}, g'}$.
%            Since $v$ is unmarked, there are at least one marked vertex $v''$ in $K\setminus X$ that matches $(\{x, z\}, g')$ such that $v''\uv v$.
%            Now, $J'':= (J\setminus \{v\})\cup \{v''\}$ induces a claw in $G-(X\cup \{v\})$ and $v''$ is adjacent to $z$ in $G-(X\cup \{v\})$.
%            Therefore, $(\widehat{J}\setminus \{v\})\cup \{v''\}$ is an induced {\ctpair} in $G-(X\cup \{v\})$.}
        \end{description} 
        \item[Subcase-(c)] When the center of the claw $J$ is in $R$.
        Say the center of the claw $J$ is $u$.
        Then, $u \in J_R$, in addition $y \in J_R$ or $y \in J_S$.
        Without loss of generality suppose that $u \in K_{nt}$.
        If $y \in J_S$, then the argument is similar to sub-subcase-(b-i) of Subcase-(b). 
        % (\ref{ct-ct-k}).

        Consider the situation when $y \in J_R$.
        Observe that $y \notin K_{pv}$, since no vertex of $K_{pv}$ is adjacent to $K_{nt}$ due to \Cref{prop:pig-clique}.
        %The argument is symmetric to \hyperref[ct-ct-k]{Subcase-(b)} of \hyperref[ct-pendant]{Case-(ii)}
        Then, $y \in K_{nt}$ or $y \notin K_{nt}$.
%        If $y \notin K_{nt}$, then $y \notin K_{pv}$ \replace{otherwise $y$ is adjacent to $v$, and it is not possible as $v$ and $y$ are the pendants of the claw $J$ centered at $u$.}{}.
        Consider the case that $y\notin K_{nt}$.
        We already have that $y \in J_R$.
        Hence, it must be that $v\uv u\uv y$.
        Now, define a function $g:\{x\}\to \{0,1\}$ such that $g(x)=0$.
        Clearly, $v\in P_{\{x\}, g}$.
        Since $v$ is unmarked and due to \Cref{obs:3-marked} of \Cref{obs:mark-signature}, there are at least 3 vertices $v_1', v_2', v_3'$ in $K\setminus X$ that match $(\{x, y\}, g)$ such that $v\uv v_1', v_2', v_3'$.
        Note that $v\uv v_1'\uv u\uv y$, $uv\in E(G-X)$ and $y\notin K_{nt}$, $v_1'u\in E(G-X)$ and $v_1'y\notin E(G-X)$.
        Now, consider $J':= (J\setminus \{v\})\cup \{v_1'\}$ and $\{v_1', v_2', v_3'\}$.
        Note that $J' = (J \setminus \{v\}) \cup \{v_1'\}$ induces a claw in $G-(X\cup \{v\})$ and $\{v_1', v_2', v_3'\}$ induces a triangle in $G - (X \cup \{v\})$ and these are pairwise intersecting.
        Hence, $G - (X \cup \{v\})$ has a component that contains a claw and a triangle.
        %{In such a situation, the argument is similar to Sub-subcase-(a-i) under Subcase-(a).}

        Now, consider the case when $y\in K_{nt}$. 
        %Since $v$ is unmarked and $u$ exists in $G-X$, it implies that there are first $2k+1$ vertices $u'\in K_{nt}$ that are non-neighbor to $x$.
        	Since $v$ is unmarked and $u$ exists in $G-X$, it implies that there are $2k+1$ neighbors of $x$ in $K_{nt}$ appearing before $u$ or $u$ itself is one of the first (at most $2k + 1$) neighbors of $x$ itself in $K_{nt}$.
        	If $u$ itself if one of these first $2k + 1$ neighbors of $x$ itslef, then $u \in K_{nt} \setminus X$.
        	Otherwise, at most $k$ of these $2k + 1$ vertices appearing before $u$ are incident to double-edges.
        %It is obvious that $u'\euv u$. 
        %It is possible that $u=u'$ for some of these vertices in $K_{nt}$.
        %However, our argument works in both cases.
        Hence, there is $u' \in K_{nt} \setminus X$ such that $u' \euv u$ and $u$ is a neighbor of $x$ in $K_{nt}$ and $u'$ is not incident to any double-edge. 
        %Because at most $k$ of these $2k + 1$ vertices are incident to a double-edge.
        	Observe that $u'$ is adjacent to $y$.
        	%As $v$ is unmarked when considering $u'$ and $x$, hence, there are $2k + 3$ neighbors of $u'$ but not neighbors of $x$, and are marked in $K^3_{nt}(x, u')$ and each appearing before $v$ in $\mathcal{V}$.
        %As $|X|\le k$ and at most $k$ many vertices out of these $2k+3$ from $K^3_{nt}(x, u')$ are incident to a double-edge, hence, there exist at least one vertex $u'\in K_{nt}\setminus X$ out of these $2k+1$ vertices in $K_{pv}$ such that $u'\euv u$ and $u'$ is not incident to any double-edge.
        Since $v$ is unmarked and by the procedure of step-(\ref{mark:k**3}) of marking scheme, there are at least $2k+3$ marked vertices $v'$ in $K$ that are neighbor of $u'$ but not $x$ and appear before $v$ i.e., $v \in K^{3}_{nt}(x, u')$.
        As $|X|\le k$ and at most $k$ many vertices are incident to double-edges, there are at least three $v_1', v_2', v_3' \in K \setminus X$ such that $v_1', v_2', v_3' \uv v$, $v_1'$ is adjacent to $u'$ but not adjacent to $x$, and $v_1', v_2', v_3'$ are not incident to any double-edge.
        Since $v_1' \uv v \uv u' \euv u \uv y$ and $vy \notin E(G-X)$, implying that $v_1'y\notin E(G-X)$ by \Cref{prop:pig-clique}. 
%        {Hence, $y \in K_{nt}$ is a nontrivial situation.
%        Note that $u$ is adjacent to $x \in S$ and $u \in K_{nt}$.
%		Then, either $u$ or some $u' \in K_{nt}$ with $u' \uv u$ is one of the first $k+1$ neighbors of $x$ in $K_{nt}$.
%		As $|X| \leq k$, $u'$ exists which could be $u$ itself.
%		Then, $v$ is one such vertex that is neighbor to $u'$ (because $v$ is neighbor of $u$ such that $u' \uv u$) but not $x$.
%		As $v$ is unmarked, there are $k+3$ marked vertices in $K$ that are neighbors of $u'$ but not $x$ and each such vertex appear before $v$ in $\mathcal{V}$.
%		As $v$ is not adjacent to $y \in K_{nt}$, these $k+3$ vertices appearing before $v$ also are not adjacent to $y$.
%		Since $|X| \leq k$, there are 3 vertices $v_1', v_2', v_3' \in K \setminus X$ such that $v_1' \uv v$, $v_1'$ is adjacent to $u$ but not adjacent to $x$.}
		Then, consider $J' = (J \setminus \{v, u\}) \cup \{v_1', u'\}$ and $\{v_1', v_2', v_3'\}$.
		%\replace{}{Note that $v_1'$ is not adjacent to $y$, hence}
        Hence,   $J'$ forms a claw in $G - (X \cup \{v\})$ and $\{v_1', v_2', v_3'\}$ forms a triangle in $G - (X \cup \{v\})$.
		As $J'$ and $\{v_1', v_2', v_3'\}$ are pairwise intersecting, $G - (X \cup \{v\})$ has a component containing both a claw and a triangle.
    \end{description}

    \item[Case-(iii)]\label{triangle-ct} When $v$ does not appear in any claw of $D$.
   	Then $v$ appears either in a triangle containing 2 other vertices of $G - X$ or $v$ appears in a path that connects a triangle and a claw in $D$.
    Consider the component $D$ of $G - X$ that contains $v$ and also has a claw $J$.
    
    First situation is when $v$ is in a triangle of $D$.
    Without loss of generality, consider $J$ being the claw closest to a triangle of $G - X$ that contains $v$.
    %Note that $D \cap V_1$ induces a path due to Observation \ref{obs:triangle-free-pig}.
    Then, $D$ must contain some vertices from $S \cup V_2$ and the claw $J$ intersects $S \cup V_2$.
    Consider the path $P$ from the triangle containing $v$ to $J$ and $x \in S$ be the first vertex of $S$ appearing in the path $P$.
    If $P$ starts from a vertex $u$ of $K \setminus (X \cup \{v\})$, then triangle $T$ in $D$ that contains $v$ also contains $u$.
    %Note that the path starts from $u$, implying that $v$ is not adjacent to $x \in S$.
    As $v$ is unmarked, the clique $K$ has at least $2k + 4$ vertices.
    Then, the triangle $T$ has another vertex $w \in K$.
    Additionally, $K \setminus X$ has another $k + 4$ vertices, out of which at most $k$ of them are incident to double-edges.
    Hence, $K \setminus\ X$ has a 4th vertex $z$ apart from $v, u, w$.
    Moreover, even in $D \setminus \{v\}$, the path $P$ is present.
    Then, $D \setminus \{v\}$ has a triangle $T' = \{u, w, z\}$.
    In particular, $D \setminus \{v\}$ forms a component of $G - (X \cup \{v\})$ that contains a triangle $T' = \{u, w, z\}$, a claw $J$, and a path $P$ from $z$ to $T'$.
    This leads to a contradiction that $G - (X \cup \{v\})$ has no connected component containing a claw and a triangle.
    
    Now, we consider when the path $P$ starts from $v$.
    The next vertex of $P$ is a neighbor of $v$.
	If the next vertex of $P$ starting from $v$ is $x \in S$, then $x$ is adjacent to $v$.
	Then, we define a function $g(x) = 1$.
	Note that $v \in K$ matches $(\{x\}, g)$ and $v$ is unmarked.
	Hence, there are at least three marked vertices in $K$ that match $(\{x\}, g)$, and not incident to any double-edge due to \Cref{obs:3-marked} of \Cref{obs:mark-signature}.
	Then, there are 3 vertices $v_1, v_2, v_3 \in K \setminus X$ such that $v_1, v_2, v_3 \uv v$ and $v_1$ is adjacent to $x$.
	Then, $\{v_1, v_2, v_3\}$ is a triangle and $x$ is adjacent to $v_1$.
	Consider the vertex set $D \setminus \{v\}$.
	This forms a component of $G - (X \cup \{v\})$ that contains the triangle $\{v_1, v_2, v_3\}$, a path $(P \setminus \{v\}) \cup \{v_1\}$ and a claw $J$, leading to a contradiction that $G - (X \cup \{v\})$ has no component containing both a claw and a triangle.
	
	The other case is when $x$ is not the next vertex of $P$ when starting from $v$.
	Then, the next vertex must contain some other vertex $w' \in V_1$.
	Without loss of generality, assume that $v \uv w'$ (the other case $w' \uv v$ has similar arguments).
	In such a case, $x$ is not adjacent to $v$. 
	Then, we define the function $g:\{x\}\to \{0,1\}$ such that $g(x) = 0$.
    Clearly, $v$ matches $(\{x\}, g)$ but not marked.
	Due to \Cref{obs:3-marked} of \Cref{obs:mark-signature}, there are 3 marked vertices $v_1, v_2, v_3 \in K \setminus X$ such that $v_1, v_2, v_3 \uv v$ and $x$ is not adjacent to $v_1, v_2, v_3$ and they are not incident to any double-edge.
	But then $w'$ is adjacent to $v_1$.
	In such a case, $D \setminus \{v\}$ has a triangle $\{v_1, v_2, v_3\}$ and a claw, in addition to a path that starts from $v_1$ then uses the subpath $w'$-$P$-$x$-$J$.
	Then $D \setminus \{v\}$ is part of a connected component of $G - (X \cup \{v\})$ containing a claw and a triangle.
	This leads to a contradiction that $G - (X \cup \{v\})$ has a component that contains both a claw and a triangle.
	
	Finally, if $v$ appears in a path $P$ of $D$ that connects a claw $J$ and a triangle $T$ in $D$, but $v \notin T \cup D$, then note that $v$ is in $K$ and $v$ is unmarked.
	If the next vertex after $v$ in $P$ is $x \in S$, then $x$ is adjacent to $v$.
	Define a function $g: \{x\} \rightarrow \{0, 1\}$ such that $g(x) = 1$.
	Note that $v$ matches $(\{x\}, g)$ and due to \Cref{obs:3-marked} of \Cref{obs:mark-signature}, there are 3 vertices $v_1, v_2, v_3 \in K$ that are adjacent to $x$ appear both before as well as after $v$ in $G - (X \cup \{v\})$; none of them is incident to any double-edge.
	Hence, there are 6 vertices in $K$ that connect $J$ and $T$ and are adjacent to $x$.
	Hence, a component of $G - (X \cup \{v\})$ contains a claw and a triangle, leading to a contradiction.
	
	If the next vertex of the path is not a vertex of $S$, then assuming $x \in S$ is the first vertex of $S$, we define a function $g(x) = 0$.
	Note that $v$ again matches $(\{x\}, g)$ and by similar arguments as before, due to \Cref{obs:3-marked} of \Cref{obs:mark-signature}, there are 3 vertices $v_1, v_2, v_3 \in K$ appearing before $v$ and and $v_1', v_2', v_3' \in K$ appearing after $v$ in $G - (X \cup \{v\})$.
	As all these 6 vertices are in $K$, and one of them is connected by the same path to $x \in S$, hence, this provides the existence of a component in $G - (X \cup \{v\})$ that contains a claw and a triangle; leading to a contradiction. 
  \end{description}
  As the above cases are mutually exhaustive, this completes the proof of the lemma.
\end{proof}
\fi

%\ifshort
%\begin{lemma}[$\star$]
%\else
%\begin{lemma}
%\fi
%\label{lemma:safeness-clique-marking}
%Reduction Rule \ref{rul:mark-clique} is safe.
%\end{lemma}

\ifshort
{The forward direction ($\Rightarrow$) of the safeness of the \Cref{rul:mark-clique} is obvious.
For the backward direction ($\Leftarrow$), if $X$ is a solution of size at most $k$ to $(G-v,k)$ but not to $(G,k)$,
then an obstruction in $G-X$ contains $v$.
By \Cref{obs:double-edge} of \Cref{obs:mark-signature}, $v$ cannot be incident to any double-edge. 
Moreover, $v$ cannot be part of an obstruction of size at most 6 by \cref{lem:small-obstruction}.
Hence, only two possible obstructions are hole of length at least 7 and a \ctpair. 
By \Cref{lem:safe-hole-7,lem:safe-ctpair}, it is possible to construct an obstruction in $G - (X \cup \{v\})$ by replacing the unmarked vertex $v \in K$ with some of the marked vertices from $K$.
It implies that $X$ is not a solution to $(G - v, k)$.
%Hence, the marking procedure guarantees that the vertices needed to witness any obstruction containing $v$ have corresponding marked vertices in $K\setminus\{v\}$.
Hence, the reduction rule is safe.
}
\else
\begin{lemma}
  \Cref{rul:mark-clique} is safe.
\end{lemma}
\begin{proof}
    The forward direction ($\Rightarrow$) of the reduction rule is trivial due to Observation \ref{obs:reduction}.

    We consider the backward direction ($\Leftarrow$) of the reduction rule.
    Suppose for the sake of contradiction that $X \subseteq V(G) \setminus \{v\}$ is a set of at most $k$ vertices such that $G - (X \cup \{v\})$ is a {\pitgraph} but $G - X$ is not a {\pitgraph}.
    Then, there must be an obstruction $O$ in $G - X$ that contains $v$ and is isomorphic to an induced cycle of length at least 4, or a tent, or a net, or $G - X$ has a component containing $v$ that has a claw and a triangle.
    By \Cref{obs:double-edge}, $O$ can not be a double-edge as $v$ is an unmarked vertex.
    But note that every net, every tent, and every induced $C_4, C_5, C_6$ of $G[S]$ is intersected by $X$.
    Due to Lemma \ref{lem:small-obstruction}, every net, every tent, and every induced $C_4, C_5, C_6$ of $G$ are also intersected by $X$.
    Hence, the obstructions of $G - X$ can only be induced cycles of length at least 6, or $G - X$ has a component $D'$ containing $v$ that contains a claw and a triangle.

    If $G - X$ has an induced cycle of length at least 7, then due to Lemma \ref{lem:safe-hole-7}, $G - (X \cup \{v\})$ also contains hole of length at least 4, or has a claw and a triangle that are pairwise intersecting.
    This leads to a contradiction that $G - (X \cup \{v\})$ is a {\pitgraph}.
    Similarly, if $G - X$ has a component that has a claw and a triangle, then due to Lemma \ref{lem:safe-ctpair}, $G - (X \cup \{v\})$ has a component that has a claw and a triangle, leading to a contradiction that $G - (X \cup \{v\})$ is a {\pitgraph}.
    Hence, this reduction rule is safe.
\end{proof}
\fi
\subparagraph*{Bounding the Number of Vertices in $G[V_1]$ and Putting Things Together.}
Now, we bound the number of vertices in $G[V_1]$ and use it to prove our main result, restated below.

{\FinalMainThm*}

%\begin{lemma}
%\label{lemma:upper-bound-V1}
%
%If none of the above reduction rules are applicable, then the number of vertices in $G[V_1]$ is $\cOO(k^{33})$.
%\end{lemma}

\begin{proof}
    Consider an instance $(G, k)$ in which none of the reduction rules are applicable.
    Our kernelization algorithm first invokes Lemma \ref{lem:small-obstruction} that either outputs that $(G, k)$ is a no-instance, or computes a set $S$ of size $\cOO(k^6)$ such that $G - S$ has no tent, no net, and no $C_4, C_5, C_6$ as induced subgraph.
    Additionally, any set $X$ of at most $k$ vertices from $S$ intersects all nets, all tents, and all induced $C_4, C_5, C_6$ of $G[S]$ if and only if $X$ intersects all nets, all tents, and all induced $C_4, C_5, C_6$ of $G$.
    After that, we apply all our Reduction Rules \ref{rul:isolated-component-reomval}--\ref{rul:mark-clique} exhaustively.
    It follows from Lemma \ref{lem:boudning-number-of-cliques} that $G[V_1]$ has $\cOO(k^2 |S|^2)$ cliques which is $\cOO(k^{14})$ cliques.
    Subsequently, due to our marking procedure {\sf Mark-Clique}($K$), we have that every clique of $\cKK$ in $G[V_1]$ that intersects $N(S)$ contains at most $\eta(k)$ vertices which is $\cOO(k^{19})$ vertices.
    Our procedure additionally also marks at most $\eta(k)$ vertices for every clique $K$ such that $K \cap N(S) = \emptyset$.
    % but for the previous and next cliques also (in the ordering $\mathcal{V}$) $K_{pv} \cap N(S), K_{nt} \cap N(S) \neq \emptyset$.
    %\todo[inline]{Diptapriyo: argument to add here for the other cliques. I believe it takes care of every clique.}
    As $|S|$ is $\cOO(k^6)$, it follows that $G[V_1]$ is $\cOO(k^{33})$.
    It follows from Lemma \ref{lemma:upper-bound-V2} that $G[V_2]$ is $\cOO(|S|^2)$ which is $\cOO(k^{12})$.
    Hence, an irreducible instance $(G, k)$ has $\cOO(k^{33})$ vertices.
    Therefore, {\PITVD} admits a kernel with $\cOO(k^{33})$ vertices. 
\end{proof}

\section{Conclusion}
\label{sec:conclusion}
%%%%%%
% We provide the nontrivial polynomial kernelization result for {\PITVD} problem; in particular, our result is a first result for polynomial kernel for any pairs of scattered graph classes $\cGG_1$ and $\cGG_2$ such that both $\cGG_1$ and $\cGG_2$ are defined by infinite forbidden families.
% A natural future direction is to improve our upper bound result? For instance, can we provide a kernel with $\cOO(k^6)$ vertices?
% While we believe that is possible, but it will require significantly specific approach to mitigate the obstructions like tent and net.
We provide the first nontrivial polynomial kernel for {\PITVD} where both $\cGG_1$ and $\cGG_2$ are defined by infinite forbidden families.
A natural open question is to improve the upper bound, e.g., to $\cOO^*(k^6)$ vertices.
\ifshort
\else
We believe that it is possible but requires a specific approach to handle obstructions like tent and net.
Jacob et al. \cite{JACOB2023280} gave $7^k n^{\cOO(1)}$-time FPT algorithm for {\PITVD}; improving this algorithm is an interesting research direction.
\fi
It would also be interesting to explore if our ideas can be adopted to provide a polynomial kernel when $\cGG_1$ is interval graphs (respectively, chordal graphs) and $\cGG_2$ is forest (respectively, bipartite permutation graphs).
Even more interesting is designing polynomial kernels for vertex deletion into three or more scattered hereditary graph classes $(\cGG_1,\cGG_2,\cGG_3, \ldots, \cGG_d)$ such that $\cGG_i$-\textsc{Deletion} admits a polynomial kernel for each $i\in [d]$.

\subparagraph*{AI Disclosure.} Every observation, claim, lemma, theorem etc. statements proved in this paper are contributions of the authors themselves. ChatGPT (free) has been used to polish texts in various parts of the paper.

\ifshort
\newpage
\else
\fi
\bibliographystyle{plain}
\bibliography{references}

%\ifshort
%\else
%\section{Just Trying}
%\input{MS/marking-scheme.tex}
%\fi
\end{document}